\pdfoutput=1
\documentclass[10pt,journal,compsoc]{IEEEtran}
\ifCLASSOPTIONcompsoc
  \usepackage[nocompress]{cite}
\else
  \usepackage{cite}
\fi
\usepackage{listings}
\makeatletter
\lst@AddToHook{OnEmptyLine}{\addtocounter{lstnumber}{-1}}
\usepackage{subcaption}
\usepackage{algorithm}
\usepackage{lipsum}

\ifCLASSINFOpdf
 	\usepackage[pdftex]{graphicx}
\else
\fi
\usepackage{amsthm}
\def\ContinueLineNumber{\lstset{firstnumber=last}}

\begin{document}
\captionsetup[lstlisting]{singlelinecheck=false, margin=0pt}
\renewcommand\lstlistingname{Algorithm}
\lstnewenvironment{sflisting}[1][]
  {\lstset{#1}}{}
\newtheorem{theorem}{Theorem}
\newtheorem{lem}{Lemma}
\newtheorem{lemm}{Lemma}
\newtheorem{cor}{Corollary}
\newtheorem{coro}{Corollary}
\newtheorem{prop}{Proposition}
\newtheorem{obs}{Observation}
\newtheorem{obse}{Observation}
\newtheorem{ill}{Illustration}
\newtheorem{inv}{Invariant}
\newtheorem{defi}{Definition}
\newtheorem{defin}{Definition}
\newtheorem{claim}{Claim}
%
\title{Quadboost: A Scalable Concurrent Quadtree}
%
%
%
%

\author{Keren~Zhou,
		Guangming~Tan,
		Wei~Zhou
\IEEEcompsocitemizethanks{\IEEEcompsocthanksitem K.Zhou and G.Tan are with Institute of Computing Technology, Chinese
Academy of Sciences.
 \protect\\
\IEEEcompsocthanksitem W.Zhou is with School of Software, Yunnan University.}%
}

\IEEEtitleabstractindextext{%
\begin{abstract}
Building concurrent spatial trees is more complicated than binary search trees since a space hierarchy should be preserved during modifications. We present a non-blocking quadtree--{\em quadboost}--that supports concurrent insert, remove, move, and contain operations. To increase its concurrency, we propose a decoupling approach that separates physical adjustment from logical removal within the remove operation. In addition, we design a continuous find mechanism to reduce its search cost. The move operation combines the searches for different keys together and modifies different positions with atomicity. The experimental results show that quadboost scales well on a multi-core system with 32 hardware threads. More than that, it outperforms existing concurrent trees in retrieving two-dimensional keys with up to 109\% improvement when the number of threads is large. The move operation proved to perform better than the best-known algorithm, with up to 47\%.
\end{abstract}

\begin{IEEEkeywords}
Concurrent Data Structures, Quadtree, Continuous Find, Decoupling, LCA
\end{IEEEkeywords}}

\maketitle

\IEEEdisplaynontitleabstractindextext

%
\IEEEpeerreviewmaketitle

\IEEEraisesectionheading{\section{Introduction}\label{sec:introduction}}

%
%
%
%
\IEEEPARstart{M}{ulti-core} processors are currently the universal computing engine in computer systems. Therefore, it is urgent to develop data structures that provide an efficient and scalable multi-thread execution. At present, concurrent data structures~\cite{moir2007concurrent}, such as stacks, linked-lists, and queues have been extensively investigated. As a fundamental building block of many parallel programs, these concurrent data structures provide significant performance benefits~\cite{shavit2011data, david2015asynchronized}.

Recently, research on concurrent trees is focused on binary search trees (BSTs)~\cite{ellen2010non, howley2012non, natarajan2014fast, ramachandran2015castle, ramachandran2015fast, chatterjee2014efficient, drachsler2014practical}, which are at the heart of many tree-based algorithms. The concurrent paradigms of BSTs were extended to design concurrent spatial trees like R-Tree~\cite{ng1993concurrent, chen1997study}. However, there remains another unaddressed spatial tree--quadtree, which is widely used in applications for multi-dimensional data. For instance, spatial databases, like PostGIS~\cite{obe2011postgis},  adopt octree, a three-dimensional variant of quadtree, to build spatial indexes. Video games apply quadtree to handle collision detection~\cite{Qollide}. In image processing~\cite{sullivan1994efficient}, quadtree is used to decompose pictures into separate regions. 

There are different categories of quadtree according to the type of data a node represents, where two major types are region quadtree and point quadtree~\cite{mehta2004handbook}. Point quadtree stores actual keys in each node, yet it is hard to design concurrent algorithms for point quadtree since an insert operation might involve re-balance issues, and a remove operation needs to re-insert the whole subtree under a removed node. The region quadtree divides a given region into several sub-regions, where internal nodes represent regions and leaf nodes store actual keys. Our work focuses on region quadtree for two reasons: (1) The shape of region quadtree is independent of insert/remove operations' order, which allow us to avoid complex re-balance rules and devise specific concurrent techniques for it. (2) Further, region quadtree could be regarded as a typical external tree that we can adjust existing concurrent techniques from BSTs. Therefore, we will refer to region quadtree as {\em quadtree} in the following context.

In this paper, we design a non-blocking quadtree, referred to as \textit{quadboost}, that supports concurrent insert, remove, contain, and move operations.
Our key contributions are as follows:
	
\begin{itemize}
\item We propose the first non-blocking quadtree. It records traversal paths, compresses \textit{Empty} nodes if necessary, adopts a decoupling technique to increase the concurrency, and devises a continuous find mechanism to reduce the cost of retries induced by \texttt{CAS} failures. 

\item We design a lowest common ancestor (LCA) based move operation, which traverses a common path for two different keys and modifies two distinct nodes with atomicity. 

\item We prove the correctness of quadboost algorithms and evaluate them on a multi-core system. The experiments demonstrate that quadboost is highly efficient for concurrent updating at different contention levels. 
\end{itemize}

The rest of this paper is organized as follows. In Section 2, we overview some basic operations. Section 3 describes a simple \texttt{CAS} quadtree to motivate this work. Section 4 provides detailed algorithms for quadboost. We provide a sketch of correctness proof in Section 5. Experimental results are discussed in Section 6. Section 7 summarizes related works. Section 8 concludes the paper.

\section{Preliminary}
\label{sec:Preliminary}
Quadtree can be considered as a dictionary for retrieving two-dimensional keys, where $<$\textit{keyX, keyY}$>$ is never duplicated. Figure~\ref{fig:sample quadtree} illustrates a sample quadtree and its corresponding region, where we use numbers to indicate keys. Labels on edges are routing directions--Southwest (sw), Northwest (nw), Southeast (se), and Northeast (ne). The right picture is a mapping of the quadtree on a two-dimensional region, where keys are located according to their coordinates, and regions are divided by their corresponding width (\textit{w}) and height (\textit{h}). There are three types of nodes in the quadtree, which represent different regions in the right figure. \textit{Internal} nodes are circles on the left figure, and each of them has four children which indicate four equal sub-regions on different directions. The root node is an \textit{Internal} node, and it is the largest region. \textit{Leaf} nodes and \textit{Empty} nodes are located at the terminal of the quadtree; they indicate the smallest regions on the right figure. \textit{Leaf} nodes are solid rectangles that store keys; they represent regions with the same numbers on the right figure. \textit{Empty} nodes are dashed rectangles without any key. 
\begin{figure}[htbp]
\centering
\includegraphics[width=0.4\textwidth]{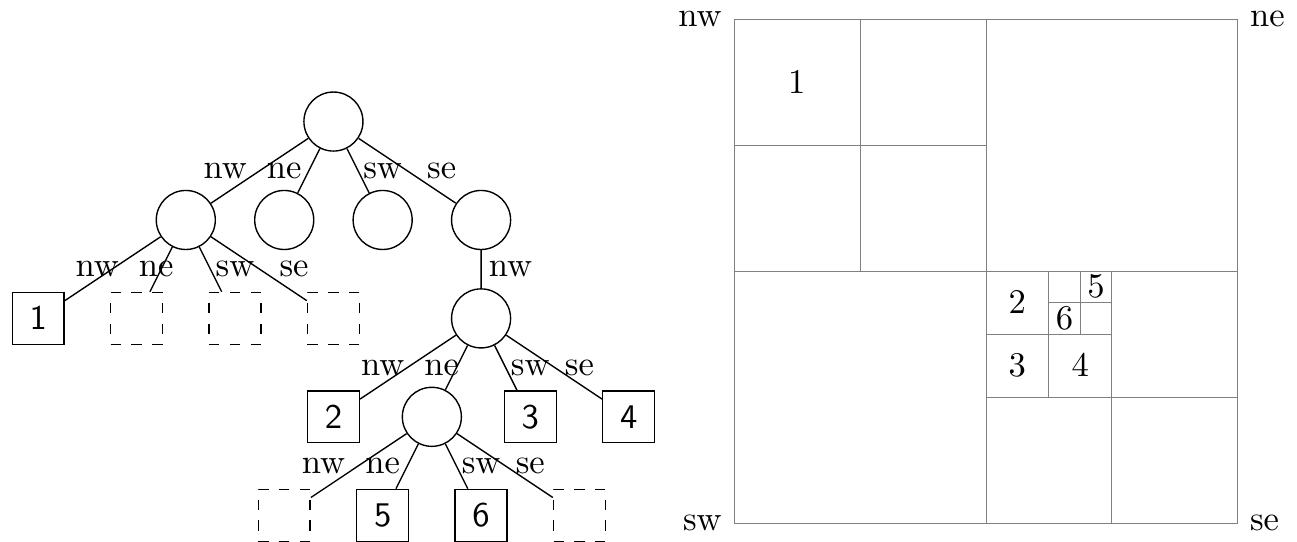}
\caption{A sample quadtree and its corresponding region}
\label{fig:sample quadtree}
\end{figure}

We describe the detailed structures of a quadtree in Figure~\ref{fig:quadtree nodes}. An \textit{Internal} node maintains its four children and a two-dimensional routing structure--$<$\textit{x, y, w, h}$>$, where $<$\textit{x, y}$>$ stands for the upper left coordinate and $<$\textit{w, h}$>$ are the width and height of the region. A \textit{Leaf} node contains a key $<$\textit{keyX, keyY}$>$ and its corresponding \textit{value}. An \textit{Empty} node does not have any field. To avoid some corner cases, we initially split the root node and its children to form two layers of dummy \textit{Internal} nodes with a layer of \textit{Empty} nodes at the terminal. Also, we present routing functions \textit{find} and \textit{getQuadrant} in the figure. For instance, if we have to locate key 1 in Figure~\ref{fig:sample quadtree}, we start with the root node and compare key 1 with its routing structures by \textit{getQuadrant}. Then we reach its \textit{nw} and perform a comparison again. In the end, we find the terminal node that contains key 1.

\begin{figure}[htbp]
\centering
\begin{sflisting}[language=Java, numberblanklines=false, numbers=left,xleftmargin=2em, basicstyle=\rmfamily\tiny, breaklines=true, breakatwhitespace=true, framexleftmargin=2em]
class Node<V> {}

class Internal<V> extends Node<V> {
    final double x, y, w, h;
    Node nw, ne, sw, se;
}

class Leaf<V> extends Node<V> {
    final double keyX, keyY;
    final V value;
}

class Empty<V> extends Node<V> {}

void find(Node& l, double keyX, double keyY) {
    while (l.class() == Internal) {
        getQuadrant(l, keyX, keyY);
    }
}

void getQuadrant(Node& l, double keyX, double keyY) {
    if (keyX < l.x + l.w / 2) {
        if (keyY < l.y + l.h / 2) l = l.nw;
        else l = l.sw;
    } else {
        if (keyY < l.y + l.h / 2) l = l.ne;
        else l = l.se;
    }
}
\end{sflisting}
\caption{Quadtree nodes and routing functions}
\label{fig:quadtree nodes}
\end{figure}
\ContinueLineNumber

There are four basic operations that rely on \textit{find}:
\begin{enumerate}
\item \textit{insert(key, value)} adds a node that contains \textit{key} and \textit{value} into a quadtree.
\item \textit{remove(key)} deletes an existing node with \textit{key} from a quadtree.
\item \textit{contain(key)} checks whether \textit{key} is in a  quadtree.
\item \textit{move(oldKey, newKey)} replaces an existing node with \textit{oldKey} and \textit{value} by a node with \textit{newKey} and \textit{value} that is not in a quadtree.
\end{enumerate}
\begin{figure}[htbp]
\begin{subfigure}{0.5\textwidth}
\centering
\includegraphics[width=0.8\textwidth]{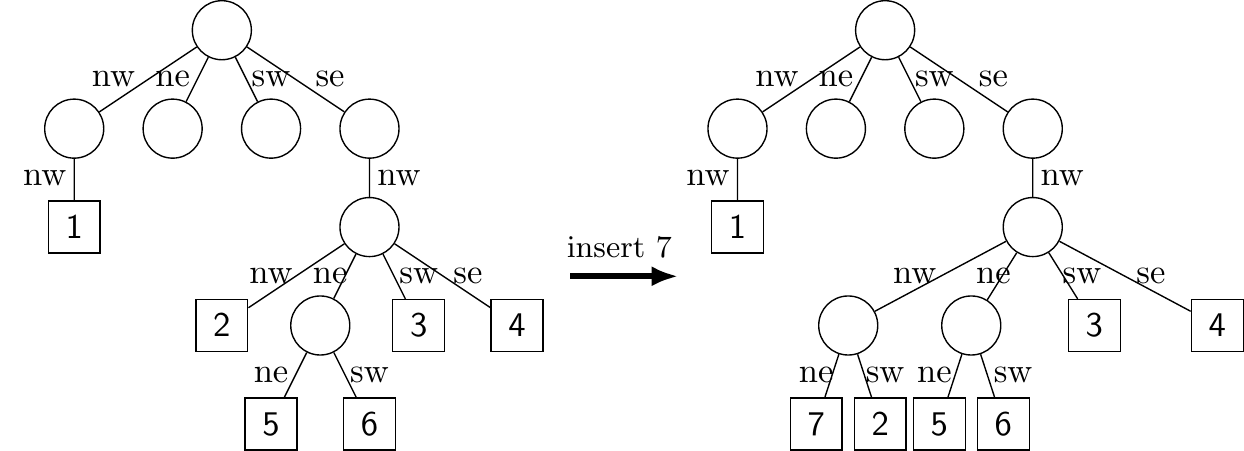}
\caption{Insert node 7 into the sample quadtree}
\label{fig:sample quadtree insert}
\end{subfigure}
\begin{subfigure}{0.5\textwidth}
\centering
\includegraphics[width=0.8\textwidth]{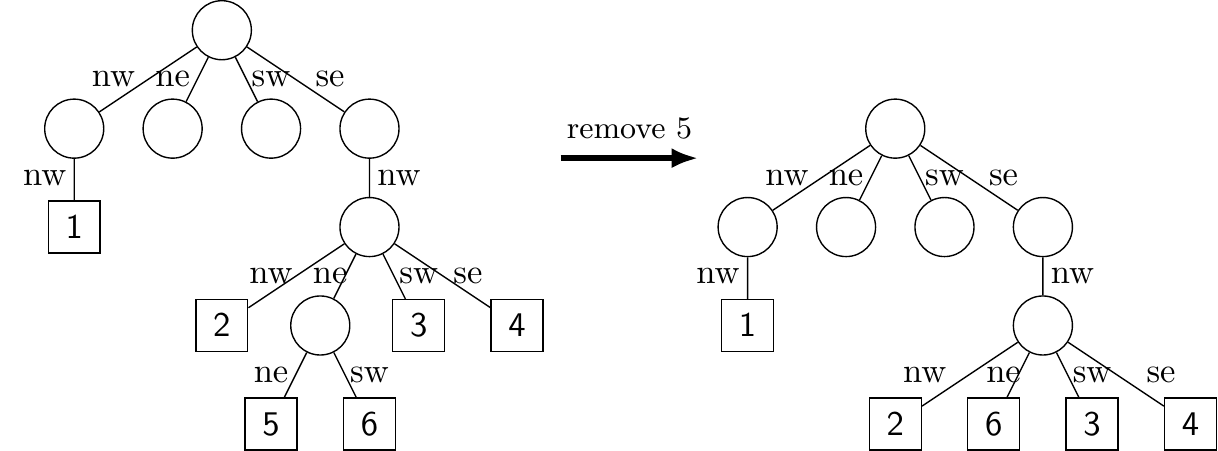}
\caption{Remove node 5 from the sample quadtree}
\label{fig:sample quadtree remove}
\end{subfigure}
\begin{subfigure}{0.5\textwidth}
\centering
\includegraphics[width=0.8\textwidth]{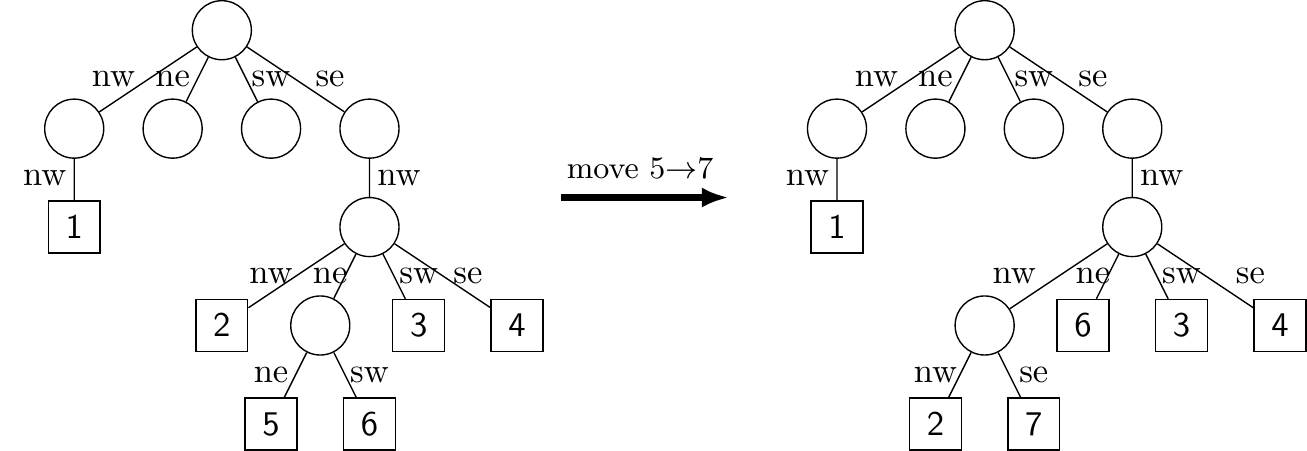}
\caption{Move the value of node 5 to node 7 from the sample quadtree}
\label{fig:sample quadtree move}
\end{subfigure}
\caption{Sample quadtree operations\protect\footnotemark}
\label{fig:sample quadtree operations}

\end{figure}
To insert a node, we first locate its position by calling \textit{find}. If we find an \textit{Empty} node at the terminal, we directly replace it by the node. Otherwise, before replacing the candidate node, we recursively divide a \textit{Leaf} node into four corresponding sub-regions until the candidate region is \textit{Empty}. Figure~\ref{fig:sample quadtree insert} illustrates a scenario of inserting node 7 (\textit{insert(7)}) as a neighborhood of node 2. The parent node is split, and node 7 is added on the \textit{ne} direction. Likewise, to remove a node, we also begin by locating it and then erase it from a quadtree. Next, we check whether its parent contains a single \textit{Leaf} node. If so, we record the node and traverse up until reaching a node that contains at least a \textit{Leaf} node or the child of the root node. Finally,  we use that node as a new parent and re-insert the recorded node as its child. Take Figure~\ref{fig:sample quadtree remove} as an example, if we remove node 5 (\textit{remove(5)}) from a quadtree, node 6 will be linked to the upper level. The move operation is a combination of the insert operation and the remove operation. It first removes the node with \textit{oldKey} and then adds a new node with \textit{newKey} into a  quadtree. Consider the scenario in Figure~\ref{fig:sample quadtree move}, after removing node 5 and inserting  node 7 with the same value (\textit{move(1, 7)}), the new tree appears on the right part. Because \textit{contain(key)} just checks whether the node returned by \textit{find} has the same \textit{key}, it does not need an extra explanation. 

\footnotetext{\footnotesize\textit{Empty} nodes are not drawn}

\section{CAS Quadtree}
\label{sec:CAS Quadtree}

\begin{figure}[htbp]
\centering
\lstset{escapeinside={(*@}{@*)}}
\begin{sflisting}[language=Java, numberblanklines=false, numbers=left,xleftmargin=2em, basicstyle=\rmfamily\tiny, breaklines=true, breakatwhitespace=true, framexleftmargin=2em]
bool contain(double keyX, double keyY) {
    Node p, l = root;
    // l: terminal node for retrieving <keyX, keyY>
    // p: parent of l
    find(p, l, keyX, keyY);
    if (inTree(l, keyX, keyY)) return true;(*@\label{cas contain in tree}@*) 
    return false;
}

bool insert(double keyX, double keyY, V value) {
    while (true) {
        Node p, l = root;
        find(p, l, keyX, keyY);
        if (inTree(l, keyX, keyY)) return false;(*@\label{cas: insert in tree}@*)
        Node newNode = createNode(l, p, keyX, keyY, value);(*@\label{cas: insert createNode}@*)  // creates a sub-tree by split l until the candidate region of <keyX, keyY> is not a Leaf node and returns the sub-tree's root node.
        if (helpReplace(p, l, newNode)) return true;(*@\label{cas: insert replace}@*)
    }
}

bool remove(double keyX, double keyY, V value) {
    Node newNode = new Empty();(*@\label{cas: remove createNode}@*)
    while (true) {
        Node p, l = root;
        find(p, l, keyX, keyY);
        if (!inTree(l, keyX, keyY)) return false;(*@\label{cas: remove in tree}@*)
        if (helpReplace(p, l, newNode)) return true;(*@\label{cas: remove replace}@*)
    }
}

void find(Node& p, Node& l, double keyX, double keyY) {
    while (l.class() == Internal) {
    	p = l;  // record the parent node
        getQuadrant(l, keyX, keyY);
    }
}

bool helpReplace(Internal p, Node oldChild, Node newChild) {
    if (p.nw == oldChild) return CAS(p.nw, oldChild, newChild);
    else if (p.ne == oldChild) return CAS(p.ne, oldChild, newChild);
    else if (p.sw == oldChild) return CAS(p.sw, oldChild, newChild);
    else if (p.se == oldChild) return CAS(p.se, oldChild, newChild);
    return false;
}

bool inTree(Node node, double keyX, double keyY) {
    return node.class == Leaf && node.keyX == keyX && node.keyY == keyY;
}
\end{sflisting}
\caption{CAS quadtree algorithm}
\label{fig:CAS quadtree}
\end{figure}

There are a plenty of concurrent tree algorithms, yet a formal concurrent quadtree has not been studied. Intuitively, we can devise a concurrent quadtree by modifying the sequential algorithm described in Section~\ref{sec:Preliminary} and adopting CAS instruction, which we name CAS quadtree. Figure~\ref{fig:CAS quadtree} depicts CAS quadtree's algorithms. Both the insert operation and the remove operation follow such a paradigm: it starts by locating a terminal node; then, it checks whether the node satisfies some conditions. If conditions are not satisfied, it returns false. Otherwise, it tries to apply a single CAS to replace the terminal node by a new node. After a successful CAS, it returns true. Or else, it restarts locating a terminal node from the root node. For a insert operation, if $<$keyX, keyY$>$ is not in the tree, it creates a corresponding sub-tree that contains $<$keyX, keyY$>$ and $value$ (line~\ref{cas: insert createNode}). Then it uses the root of the sub-tree to replace the terminal node by a single CAS (line~\ref{cas: insert replace}). To remove a node, it creates an \textit{Empty} node at the beginning (line~\ref{cas: remove createNode}). If $<$keyX, keyY$>$ is in the tree, it uses the \textit{Empty} node to replace the terminal node (line~\ref{cas: remove replace}). Different from the sequential algorithm, we do not adjust the structure as it involves several steps that cannot be implemented with atomicity. The algorithm is non-blocking. In other words, the algorithm provides a guaranteed whole progress, even if some threads starve.

\begin{figure}[htbp]
\centering
\hbox{\hspace{-1.0em}
\includegraphics[width=0.45\textwidth]{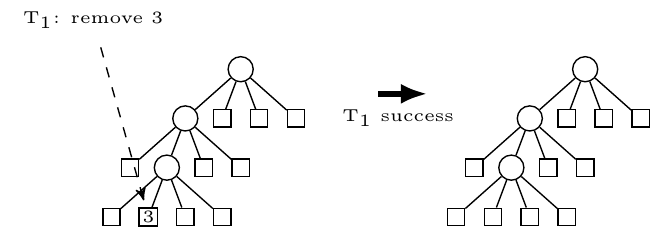}
}
\caption{Thread T$_{1}$ intends to remove node 3 from quadtree. After its removal, there remains a chain of \textit{Empty} nodes}
\label{fig:quadtree remove}
\end{figure}

\begin{figure}[htbp]
\centering
\hbox{\hspace{-1.0em}
\includegraphics[width=0.5\textwidth]{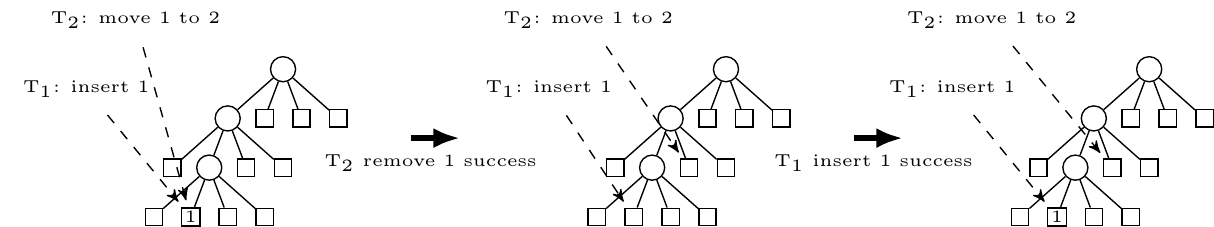}
}
\caption{An example of an incorrect move operation. Thread T$_{1}$ intends to insert node 1 into a quadtree, and thread T$_{2}$ plans to move the value from node 1 to node 2. T$_{2}$ first successfully removes node 1 and then attempts to insert node 2 into the quadtree. In the interval node 1 is added back by T$_{1}$, but T$_{2}$ is not aware of the action and reports success.}
\label{fig:quadtree move}
\end{figure}

\begin{figure*}[htbp]
\begin{subfigure}{\textwidth}
\centering
\includegraphics[scale=0.9]{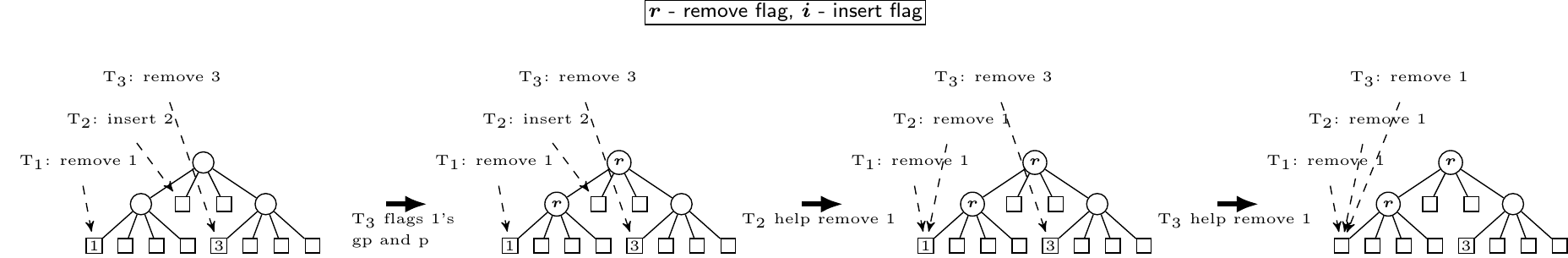}
\caption{Apply traditional removal}
\label{fig:BST's remove}
\end{subfigure}

\begin{subfigure}{\textwidth}
\centering
\includegraphics[scale=0.9]{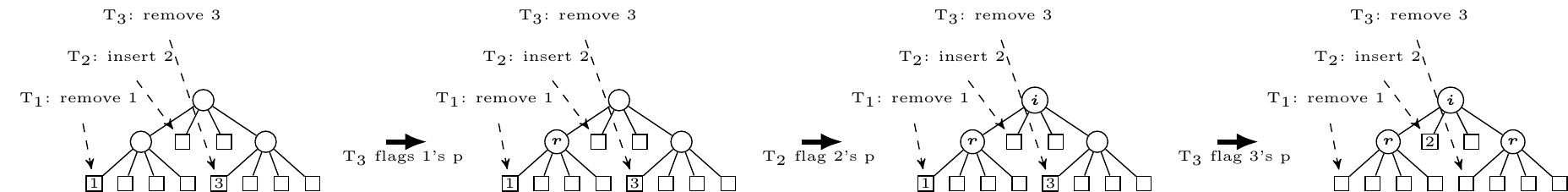}
\caption{Our decoupling approach}
\label{fig:Our decoupling approach}
\end{subfigure}

\caption{At the beginning, three threads are performing different operations concurrently: (a) T$_{1}$ removes key 1 in the lower level. (b) T$_{2}$ inserts key 2 in the upper level. (c) T$_{3}$ removes key 3 in the lower level. Consider the scenario that T$_{1}$ precedes others threads, and then both 1's parent and grandparent will be flagged. Hence, T$_{2}$ and T$_{3}$ will help T$_{1}$ remove key 1 before restarting their operations. By applying our decoupling method, only the parent node should be flagged. Hence three threads could run without intervention.}
\label{fig:decoupling}
\end{figure*}

However, CAS quadtree has several limitations. First, consider if there are a considerable proportion of remove operations. By applying the simple mechanism, we still have a large number of nodes in the quadtree because we substitute existing nodes with \textit{Empty} nodes without structural adjustment. Figure~\ref{fig:quadtree remove} illustrates a detailed example showing that there remains a chain of \textit{Empty} nodes. Hence, not only we have to traverse a long path to locate the terminal node for basic operations, but also a plenty of nodes are maintained in the memory. 

Second, we cannot implement the move operation in the simple algorithm. There might be two different nodes in a quadtree that were under modifying. If we directly erase the node with \textit{oldKey} and add the node with \textit{newKey}, the move operation cannot be correctly linearized. Figure~\ref{fig:quadtree move} shows such an incorrect move operation by combining an insert operation and a remove operation together. 

Therefore, these drawbacks motivate us to develop a new concurrent algorithm to make the move operation correct and employ an efficient mechanism to compress nodes.

\section{Quadboost}

\subsection{Rationale}

In this section, we describe how to design \textit{quadboost} algorithms to solve the two problems addressed in Section~\ref{sec:CAS Quadtree}.

To make the move operation correct, we should ensure that other threads know whether a terminal node is under moving. Hence, we attach an internal node with a separate object--\textit{Operation} (\textit{op}) to represent a node's state and record sufficient information to complete the operation. We instantiate the attachment behavior as a CAS and call it a \textit{flag} operation. We design different \textit{op}s for the insert operation, the remove operation and the move operation. The detailed description of structures and a state transition mechanism is presented in Section~\ref{sec:Structures and State Transitions}.

To erase \textit{Empty} nodes from a quadtree, we can apply a similar paradigm in the concurrent BST's removal~\cite{ellen2010non} as shown in Figure~\ref{fig:BST's remove}, which flags both the parent and the grandparent of a terminal node. This mechanism mixes logical removal and physical removal together. Different from the BST's removal in which every time the parent node has to be adjusted, we only have to compress the parent when there is only a single \textit{Leaf} child. Hence, we could separate the removal of a node and the adjustment of the structure into two phases. Meanwhile, we design an \textit{op}--\textit{Compress} to indicate a node is underlying the structural adjustment. Figure~\ref{fig:Our decoupling approach} illustrates how this method increases concurrency: three threads that handle different \textit{op}s could run in parallel. Further, for simplicity, we relax the adjustment condition to where the parent could be compressed if all children are \textit{Empty}. 

However, there is still a problem left after applying the above two methods. Recall the example in Figure~\ref{fig:quadtree remove}. We flag the bottom \textit{Internal} node to indicate that it should be compressed, but after replacing the bottom \textit{Internal} node with an \textit{Empty} node, it results in four \textit{Empty} nodes in the last level. How do we remove a series of nodes from a quadtree in a bottom-up way? We record the entire traversal path from the root to a terminal node in a stack. Since the traversal path will be altered when a node is compressed, we only have to restart locating the terminal node from its parent if any flag operation other than the flag of \textit{Compress} fails. This is called the continuous find mechanism. 

\subsection{Structures and State Transitions}
\label{sec:Structures and State Transitions}
\begin{figure}[htbp]
\centering
\lstset{escapeinside={(*@}{@*)}}
\begin{sflisting}[language=Java, numberblanklines=false, numbers=left,xleftmargin=2em, basicstyle=\rmfamily\tiny, breaklines=true, breakatwhitespace=true, framexleftmargin=2em]
class Node<V> {}

class Internal<V> extends Node<V> {
    final double x, y, w, h;
    volatile Node nw, ne, sw, se;
    volatile Operation op = new Clean();(*@\label{quadboost: structure Internal op Clean}@*) 
}

class Leaf<V> extends Node<V> {
    final double keyX, keyY;
    final V value;
    volatile Move op;
}

class Empty<V> extends Node<V> {}

class Operation {}

class Substitute extends Operation {
    Internal parent;
    Node oldChild, newNode;
}

class Compress extends Operation {
    Internal grandparent, parent;
}

class Move extends Operation {
    Internal iParent, rParent;
    Node oldIChild, oldRChild, newIChild;
    Operation oldIOp, oldROp;
    volatile bool allFlag = false, iFirst = false;
}

class Clean extends Operation {}
\end{sflisting}
\caption{quadboost structures}
\label{fig:structure}
\end{figure}

As mentioned before, we add an \textit{Operation} object to handle concurrency issues. Figure~\ref{fig:structure} shows the data structures of \textit{quadboost}. Four sub-classes of \textit{Operation}, including \textit{Substitute}, \textit{Compress}, \textit{Move}, and \textit{Clean}, describe all states in our algorithm. \textit{Substitute} provides information on the insert operation and the remove operation that are designed to replace an existing node by a new node. Hence, we shall let other threads be aware of its parent, child, and a new node for substituting. \textit{Compress} provides information for physical adjustment. We erase the parent node, previously connected to the grandparent, by swinging the link to an \textit{Empty} node. \textit{Move} stores both \textit{oldKey}'s and \textit{newKey}'s terminal nodes, their parents, their parents' prior \textit{op}s, and a new node. Moreover, we use a bool variable--\textit{allFlag} to indicate whether two parents have been attached to a \textit{Move} \textit{op} or not. Another bool variable--\textit{iFirst} is used to indicate the attaching order. For instance, if \textit{iFirst} is true, \textit{iParent} will be attached with a \textit{Move} \textit{op} before \textit{rParent}. \textit{Clean} means that there is no thread modifying the node. In contrast to Figure~\ref{fig:CAS quadtree}, the \textit{Internal} class adds an \textit{op} field to hold its state and related information. Note that a flag operation (CAS) is only applied on \textit{Internal} nodes. We atomically set \textit{Move} \textit{op}s in \textit{Leaf} nodes to linearize the move operation correctly.

Each basic operation, except for the contain operation, starts by changing an \textit{Internal} node's \textit{op} from \textit{Clean} to other states. The three basic operations, therefore, generate a corresponding state transition diagram which provides a high-level description of our algorithm. After locating a terminal node,  insert, remove, compress, and move transitions are executed as shown in Figure~\ref{fig:State transition diagram}. In the figure, we specify the flag operation that restores an \textit{op} to \textit{Clean} as \textit{unflag}, the flag operation that changes an \textit{op} from \textit{Clean} to \textit{Substitute} in \textit{insert} as \textit{iflag}, the flag operation that changes an \textit{op} from \textit{Clean} to \textit{Substitute} in \textit{remove} as \textit{rflag}, and the flag operation that changes an \textit{op} from \textit{Clean} to \textit{Compress} as \textit{cflag}. We describe how these transitions execute when a thread detects a state as follows:

\begin{figure}[htbp]
\centering
\includegraphics[scale=0.8]{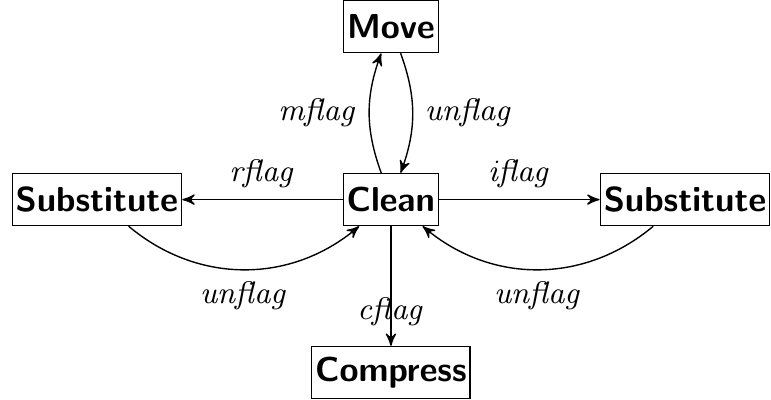}
\caption{State transition diagram}
\label{fig:State transition diagram}
\end{figure}

\begin{itemize}
\item \textbf{Clean}. For the insert transition, a thread constructs a new node and changes the parent's \textit{op} to \textit{Substitute} by \textit{iflag}. For the remove transition, the thread creates an \textit{Empty} node and changes the parent's \textit{op} from \textit{Clean} to \textit{Substitute} by \textit{rflag}. For the move transition, the thread flags both \textit{newKey}'s and \textit{oldKey}'s parents by \textit{mflag}. For the compress transition, the thread uses a \textit{cflag} operation to flag the parent if necessary. 

\item \textbf{Substitute}. The thread uses a CAS to change the existing node by a given node stored in the \textit{op}. It then restores the parent's state to \textit{Clean} by \textit{unflag}.

\item \textbf{Move}. The thread first determines the flag order by comparing \textit{oldKey}'s parent with \textit{newKey}'s parent. Suppose it flags \textit{newKey}'s parent first, it will flag \textit{oldKey}'s parent later. Then, the thread replaces \textit{oldKey}'s terminal with an \textit{Empty} node and replaces \textit{newKey}'s terminal with a new node. Finally, it \textit{unflag}s their parents' \textit{op} to \textit{Clean} in the reverse order.

\item \textbf{Compress}. The thread erases the node from the tree so that it cannot be detected. Different than other states, the node with a \textit{Compress} \textit{op} cannot be set to \textit{Clean} by \textit{unflag}.
\end{itemize}

\subsection{Concurrent Algorithms}

Figure~\ref{fig:quadboost insert} reflects quadboost's insert and contain operations. Both of them start by a find process that locates a terminal node. The find operation (line~\ref{quadboost: contain find}) pushes \textit{Internal} nodes into a stack (line~\ref{quadboost: contain path}) and keeps recording the parent node's \textit{op} (line~\ref{quadboost: contain pOp}).

The contain operation executes in a similar way as CAS quadtree. It calls the find function to locate a terminal node (line~\ref{quadboost: contain find}). The \textit{op} at line~\ref{quadboost: contain pOp} and the stack at line~\ref{quadboost: contain path} are created for a modular presentation. We can omit them in a real implementation.

\begin{figure}[htbp]
\centering
\lstset{escapeinside={(*@}{@*)}}
\begin{sflisting}[language=Java, numberblanklines=false, numbers=left,xleftmargin=2em, basicstyle=\rmfamily\tiny, breaklines=true, breakatwhitespace=true, framexleftmargin=2em]
bool contain(double keyX, double keyY) {
    Stack<Node> path;(*@\label{quadboost: contain path}@*)
    Operation pOp;(*@\label{quadboost: contain pOp}@*)
    Node l = root;(*@\label{quadboost: contain root}@*)
    find(l, pOp, path, keyX, keyY);(*@\label{quadboost: contain find}@*)
    if (inTree(l, keyX, keyY) && !moved(l)) return true;
    return false;
}

bool insert(double keyX, double keyY, V value) {
    Stack<Node> path;(*@\label{quadboost: insert path}@*)
    Operation pOp;(*@\label{quadboost: insert pOp}@*)
    Node l = root;(*@\label{quadboost: insert root}@*)
    find(l, pOp, path, keyX, keyY);(*@\label{quadboost: insert find}@*)
    while (true) {
        if (inTree(l, keyX, keyY) && !moved(l)) return false;(*@\label{quadboost: insert intree}@*)
        p = path.pop();(*@\label{quadboost: insert pop}@*)
        if (pOp.class == Clean) {(*@\label{quadboost: insert pop Clean check}@*)
            Node newNode = createNode(l, p, keyX, keyY, value);(*@\label{quadboost: insert createNode}@*)
            Operation op = new Substitute(p, l, newNode);(*@\label{quadboost: insert Substitute}@*)
            if (helpFlag(p, pOp, op)) {(*@\label{quadboost: insert helpFlag}@*)
                helpSubstitute(op);(*@\label{quadboost: insert helpSubstitute}@*)
                return true;
            } else pOp = p.op;(*@\label{quadboost: insert p op}@*)
        }
        help(pOp); (*@\label{quadboost: insert help}@*)  // help complete the operation of pOp   
        continueFind(pOp, path, l, p);(*@\label{quadboost: insert continueFind}@*)  // find a new terminal in a bottom up way 
    }
}

void continueFind(Operation& pOp, Stack& path, Node& l, Internal p) {(*@\label{quadboost: continueFind}@*)
    if (pOp.class != Compress) l = p;  // start from the parent node(*@\label{quadboost: continueFind p}@*)
    else {
    	while (!path.isEmpty()) {  // find a node in the path that is not in the Compress state
    	    l = path.pop();(*@\label{quadboost: continueFind path pop}@*)
    	    pOp = l.op;(*@\label{quadboost: continueFind pOp}@*)
    	    if (pOp.class == Compress) helpCompress(pOp);
    	    else break;  // find the node
    	}
    }
    find(l, pOp, path, keyX, keyY);(*@\label{quadboost: continueFind find}@*)
}

void find(Node& l, Operation& pOp, Stack& path, double keyX, double keyY) (*@\label{quadboost: find}@*)
    while (l.class() == Internal) {(*@\label{quadboost: find check}@*)
        path.push(l);(*@\label{quadboost: find path}@*)
    	pOp = l.op;(*@\label{quadboost: find op}@*)
        getQuadrant(l, keyX, keyY);
    }
}


bool helpFlag(Internal node, Operation oldOp, Operation newOp) {
    return CAS(node.op, oldOp, newOp);
}

void helpSubstitute(Substitute op) {
    helpReplace(op.parent, op.oldChild, op.newNode);(*@\label{quadboost: helpSubstitute helpReplace}@*)
    helpFlag(op.parent, op, new Clean()); (*@\label{quadboost: helpSubstitute helpFlag}@*)  // unflag the parent node to Clean
}

void help(Operation op) {
    if (op.class == Compress) helpCompress(op);(*@\label{quadboost: help helpCompress}@*)
    else if (op.class == Substitute) helpSubstitute(op);(*@\label{quadboost: help helpSubstitute}@*)
    else if (op.class == Move) helpMove(op);(*@\label{quadboost: help helpMove}@*)
}

bool moved(Node l) {
    return l.class == Leaf && l.op != null && !hasChild(l.op.iParent, l.op.oldIChild);
}


bool hasChild(Internal parent, Node oldChild) {
    return parent.nw == oldChild || parent.ne == oldChild || parent.sw == oldChild || parent.se == oldChild;
}
\end{sflisting}
\caption{quadboost \textit{insert} and \textit{contain}}
\label{fig:quadboost insert}
\end{figure}

In the insert operation, we create a stack at the beginning to record the traversal path (line~\ref{quadboost: contain find}) and a \textit{pOp} to record the parent node's \textit{op} (line~\ref{quadboost: contain pOp}). After locating a terminal node, we check whether the key of the node is in the tree and whether the node is moved at line~\ref{quadboost: insert intree}. We will show the reason why we have to check whether a node is moved in Section~\ref{sec: LCA-based Move Operation}. Then, we flag the parent of the terminal node before replacing it. In the next step, we call the helpSubstitute function at line~\ref{quadboost: insert helpSubstitute}, which first invokes the helpReplace function at line~\ref{quadboost: helpSubstitute helpReplace} to replace the terminal node and then \textit{unflag} the parent node at line~\ref{quadboost: helpSubstitute helpFlag}. If the flag operation fails, we update \textit{pOp} at line~\ref{quadboost: insert p op} and help it finish at line~\ref{quadboost: insert help}. More than that, we have to execute the continueFind function to restart from the nearest parent. The continueFind function (line~\ref{quadboost: continueFind}) pops nodes from the path until reaching a node whose \textit{op} is not \textit{Compress}. If the node's \textit{op} is \textit{Compress}, it helps the \textit{op} finish. Or else, it breaks the loop and performs the find operation from the last node popped at line~\ref{quadboost: continueFind find}. 

The remove operation has a similar paradigm as the insert operation. It first locates a terminal node and checks whether the node is in the tree or moved. After that, it flags the parent node and replaces the terminal node with an \textit{Empty} node. An extra step in the remove operation is the compress function at line~\ref{quadboost: remove compress}. In the algorithm, we perform the compress function before the remove operation returns true. In this way, the linearization point of the remove operation belongs to the execution of itself, and extra adjustment induced by the compress function do affect the effectiveness. The compress function should examine three conditions before compressing the parent node. First, because the remove operation must return true, we do not compress nor help if the state of the parent is not \textit{Clean} (line~\ref{quadboost: compress clean}). Second, we check if the grandparent node (\textit{gp}) is the root (line~\ref{quadboost: compress root}) as we maintain two layers of dummy \textit{Internal} nodes. At last, we check whether four children of a parent node are all \textit{Empty} (line~\ref{quadboost: compress check}). 

\begin{figure}[htbp]
\centering
\lstset{escapeinside={(*@}{@*)}}
\begin{sflisting}[language=Java, numberblanklines=false, numbers=left,xleftmargin=2em, basicstyle=\rmfamily\tiny, breaklines=true, breakatwhitespace=true, framexleftmargin=2em]
bool remove(double keyX, double keyY, V value) {
    Stack<Node> path;(*@\label{quadboost: remove path}@*)
    Node l = root, newNode = new Empty();(*@\label{quadboost: remove root}@*)
    Operation pOp;
    find(l, keyX, keyY, pOp, path);(*@\label{quadboost: remove find}@*)
    while (true) {
        if (!inTree(l, keyX, keyY) || moved(l)) return false;(*@\label{quadboost: remove intree}@*)
        p = path.pop();(*@\label{quadboost: remove pop}@*)
        if (pOp.class == Clean) {(*@\label{quadboost: remove pop Clean check}@*)    
            Operation op = new Substitute(p, l, newNode);(*@\label{quadboost: remove Substitute}@*)
            if (helpFlag(p, pOp, op) {(*@\label{quadboost: remove helpFlag}@*)
                helpSubstitute(op);(*@\label{quadboost: remove helpSubstitute}@*) 
                compress(path, p);(*@\label{quadboost: remove compress}@*)  // compress the path if necessary
                return true;
            } else pOp = p.op;(*@\label{quadboost: remove p op}@*)
        }
        help(pOp);(*@\label{quadboost: remove help}@*)
        continueFind(pOp, path, l, p);(*@\label{quadboost: remove continueFind}@*)
    }
}

void compress(Stack& path, Internal p) {
    while (true) {
        Operation pOp = p.op;(*@\label{quadboost: compress p op}@*)
        if (pOp.class == Clean) {(*@\label{quadboost: compress clean}@*)
            gp = path.pop();
            if (gp == root) return;(*@\label{quadboost: compress root}@*)  // two layers of Internal nodes are maintained
            Operation op = new Compress(gp, p);(*@\label{quadboost: compress Compress}@*)
            if (!check(p) || !helpFlag(p, pOp, op)) return;(*@\label{quadboost: compress check}@*)
            else helpCompress(op);(*@\label{quadboost: remove helpCompress}@*)
            p = gp;
        } else return;
    }
}


bool helpCompress(Compress op) {
    return helpReplace(op.grandparent, op.parent, new Empty<V>());(*@\label{quadboost: helpCompress helpReplace}@*)
}

bool check(Internal node) {
    return node.nw.class == Empty && node.ne.class == Empty && node.sw.class == Empty && node.se.class == Empty;
}
\end{sflisting}
\caption{quadboost \textit{remove}}
\label{fig:quadboost remove}
\end{figure}

\subsection{LCA-based Move Operation}
\label{sec: LCA-based Move Operation}

\begin{figure}[htbp]
\centering
\includegraphics[width=0.4\textwidth]{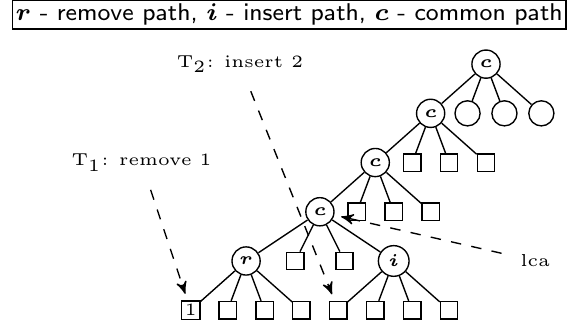}
\caption{Paths of two operations--insert node 2 and remove node 1 that share the LCA node}
\label{fig:lca}
\end{figure}

Tree-based structures share a property that two nodes in the tree have a common path starting from the root, and the lowest node in the path is called the lowest common ancestor (LCA). Figure~\ref{fig:lca} demonstrates the LCA node of a quadtree by a concrete example. Based on the observation, our LCA-based move operation is defined to find two different terminal nodes sharing a common path, remove a node with \textit{oldKey}, and insert a node with \textit{newKey}.

\begin{figure}[t]
\centering
\lstset{escapeinside={(*@}{@*)}}
\begin{sflisting}[language=Java, numberblanklines=false, numbers=left,xleftmargin=2em, basicstyle=\rmfamily\tiny, breaklines=true, breakatwhitespace=true, framexleftmargin=2em]
bool move(double oldKeyX, double oldKeyY, double newKeyX, double newKeyY) {
    Stack rPath, iPath;(*@\label{quadboost: move two paths}@*)
    // rl: terminal node of the remove path
    // il: terminal node of the insert path
    // rp: parent node of the remove path
    // ip: parent node of the insert path
    Node rl = root, il, rp, ip;(*@\label{quadboost: move root}@*)
    // rOp: op object of rp
    // iOp: op object of ip
    Operation rOp, iOp;
    // rFail: whether the remove path fail
    // iFail: whether the insert path fail
    // cFail: whether the common path fail
    bool rFail = false, iFail = false, cFail = false;
    if (!findCommon(il, rl, lca, rOp, iOp, iPath, rPath, oldKeyX, oldKeyY, newKeyX, newKeyY)) return false;(*@\label{quadboost: move findCommon}@*)
    ip = iPath.pop(); (*@\label{quadboost: move ip pop}@*)
    rp = rPath.pop();(*@\label{quadboost: move rp pop}@*)
    while (true) {
        if (rOp.class != Clean) rFail = true;(*@\label{quadboost: move rOp check}@*)
        if (iOp.class != Clean) iFail = true;(*@\label{quadboost: move iOp check}@*)
        if (iOp != rOp && ip == rp) cFail = true;(*@\label{quadboost: move ip rp same}@*)  // the same parent node has two different states
        if (!cFail && !iFail && !rFail) {
            if (il.class == Empty || il == rl) newNode = new Leaf(newKeyX, newKeyY, rl.value);(*@\label{quadboost: move create begin}@*)
            else newNode = createNode(il, ip, newKeyX, newKeyY, rl.value);(*@\label{quadboost: move createNode}@*)
            //iFirst: flag which node first            
            bool iFirst = getSpatialOrder(ip, rp);(*@\label{quadboost: move getSpatialOrder}@*)
            Operation op = new Move(ip, rp, il, rl, newNode, iOldOp, rOldOp, iFirst);(*@\label{quadboost: move create end}@*)
            if (rp != ip) {  // two parents are different
            	bool hasFlag = iFirst ? helpFlag(ip, iOp, op) : helpFlag(rp, rOp, op);(*@\label{quadboost: move helpFlag}@*)
                if (hasFlag) {
                    if (helpMove(op)) {(*@\label{quadboost: move helpMove first}@*)
                        compress(rp, rPath);(*@\label{quadboost: move compress}@*)
                        return true;				
                    } else {  // all flag operation fail
                        rFail = iFail = true;(*@\label{quadboost: rFail iFail true}@*)			
                    }
                } else {  // one of two paths fail
                    if (iFirst) {iOp = ip.op; iFail = true;} (*@\label{quadboost: iFail true}@*)
                    else {rOp = rp.op; rFail = true;}(*@\label{quadboost: rFail true}@*)
                }
            } else {  // two parents are the same
                if (helpMove(op)) return true;(*@\label{quadboost: move helpMove second}@*)
                else {rOp = iOp = rp.op; cFail = true;}(*@\label{quadboost: cFail true}@*)
            }
        }
        if (!continueFindCommon(iFail, rFail, cFail, il, rl, ip, rp, lca, rOp, iOp, oldKeyX, oldKeyY, newKeyX, newKeyY, iPath, rPath)) return false;(*@\label{quadboost: move continueFindCommon}@*)  // continuous find insert and remove paths
        cFail = iFail = rFail = false;
    }
}

bool helpMove(Move op) {
    if (op.iFirst) helpFlag(op.rParent, op.oldROp, op);  // flag the parent of the oldKey's terminal first(*@\label{quadboost: helpMove flag begin}@*)
    else helpFlag(op.iParent, op.oldIOp, op);  // flag the parent of the newKey's terminal first
    bool doCAS = op.iFirst ? op.rParent.op == op: op.iParent.op == op;  // whether the flag operation succeed(*@\label{quadboost: helpMove flag end}@*)
    if (doCAS) {  // all flags have been done
        op.allFlag = true;(*@\label{quadboost: helpMove allFlag}@*)
        op.oldRChild.op = op;(*@\label{quadboost: helpMove assign op}@*)
        if (op.oldRChild == op.oldRChild) {  // combine two CASes in two one(*@\label{quadboost: helpMove two cas begin}@*)
            helpReplace(op.rParent, op.oldRChild, op.newIChild);(*@\label{quadboost: helpMove helpReplace first}@*)
        } else {
            helpReplace(op.iParent, op.oldIChild, op.newIChild);(*@\label{quadboost: helpMove helpReplace first first}@*)
            helpReplace(op.rParent, op.oldRChild, new Empty());(*@\label{quadboost: helpMove helpReplace first second}@*)
       }(*@\label{quadboost: helpMove two cas end}@*)
    }
    if (op.iFirst) {  // unflag in a reverse order(*@\label{quadboost: helpMove reverse begin}@*)
        if (op.allFlag) helpFlag(op.rParent, op, new Clean());
        if (op.rParent != op.rParent) helpFlag(op.iParent, op, new Clean());
    } else {
        if (op.allFlag) helpFlag(op.iParent, op, new Clean());
        if (op.rParent != op.rParent) helpFlag(op.rParent, op, new Clean());
    }(*@\label{quadboost: helpMove reverse end}@*)
    return op.allFlag;
}
\end{sflisting}
\caption{quadboost \textit{move}}
\label{fig:quadboost move}
\end{figure}
Figures~\ref{fig:quadboost move} and~\ref{fig: quadboost findCommon and continueFindCommon} present the algorithm of the move operation. In contrast with the insert operation and the remove operation, the move operation begins by calling the findCommon function (line~\ref{quadboost: move findCommon}) that combines searches for two nodes together. The two searches share a common path that allows us to record them only once. We use two stacks to record nodes at line~\ref{quadboost: move two paths}. The shared path and the remove path are pushed into \textit{rPath}; and the insert path is pushed into \textit{iPath}. By doing this, we could avoid considering complicated corner cases in terms of node compressions in \textit{rPath}. The findCommon function begins by reading a child node from the parent for \textit{oldKey} and checks whether \textit{newKey} is in the same direction (line~\ref{quadboost: findCommon check}). If not, it terminates the traversal, and searches for individual keys separately (line~\ref{quadboost: findCommon find il} and line~\ref{quadboost: findCommon find rl}). If \textit{oldKey} is not in the tree or is moved, or if \textit{newKey} is in the tree and not moved, it returns false. 

The move operation then checks two parents' \textit{op} (\textit{iOp} and \textit{rOp}) before flag operations. If neither of them is \textit{Clean} (line~\ref{quadboost: move rOp check}-\ref{quadboost: move iOp check}), or \textit{rp} and \textit{ip} are the same but their \textit{op}s are different (line~\ref{quadboost: move ip rp same}), it starts the continueFindCommon function at line~\ref{quadboost: move continueFindCommon} which we will discuss later in the section. Otherwise, it creates a new node for inserting and an \textit{op} to hold essential information for a CAS at line~\ref{quadboost: move create begin}-\ref{quadboost: move create end}. There are two specific cases. If two terminal nodes share a common parent, we directly call the helpMove function at line~\ref{quadboost: move helpMove first}. Or else, to avoid live locks, we flag two nodes in a specific order. In our algorithm, we use the getSpatialOrder function at line~\ref{quadboost: move getSpatialOrder} to compare \textit{ip} and \textit{rp} in the following order: $x \rightarrow y \rightarrow w$, where $x$, $y$, and $w$ are the fields of an \textit{Internal} node. We prove that this method produces a unique order among all \textit{Internal} nodes in a quadtree in appendix~\ref{sec:appendix}. 

To start, the helpMove function flags a parent and checks whether both parents are successfully flagged at line~\ref{quadboost: helpMove flag begin}-\ref{quadboost: helpMove flag end}. If both flag operations succeed, it sets \textit{allFlag} to true at line~\ref{quadboost: helpMove allFlag}. Note that we assign the \textit{op} to \textit{oldRChild} at line~\ref{quadboost: helpMove assign op}, letting other threads know whether the CAS on \textit{iParent} has finished, which is the linearization point of the move operation. Because simultaneously $oldRChild$ is aware that it is removed, it also explains why we should examine whether a node is moved. Then, we check whether two terminal nodes are the same. If so, we could combine two CASs into a single one. If not, we should remove \textit{oldKey}'s terminal and replace \textit{newKey}'s terminal. In the end, we reset the parents' \textit{op} to \textit{Clean} in the reverse order (line~\ref{quadboost: helpMove reverse begin}-\ref{quadboost: helpMove reverse end}). 

The continueFindCommon function is also invoked when one of the flag operations fails. If \textit{rp} cannot be flagged, it pops nodes from \textit{rPath} until it reaches the LCA node (line~\ref{quadboost: continueFindCommon remove lca}) or a node's \textit{op} is \textit{Compress} (line~\ref{quadboost: continueFindCommon remove helpCompress}). It then starts from the last popped node to search for a new terminal of \textit{oldKey} at line~\ref{quadboost: continueFindCommon remove find}. If \textit{ip} cannot be flagged, it pops nodes from \textit{iPath} until it is empty (line~\ref{quadboost: continueFindCommon insert il pop}) or a node's \textit{op} is not \textit{Compress} (line~\ref{quadboost: continueFindCommon insert helpCompress}). It again continues to search for \textit{newKey} at line~\ref{quadboost: continueFindCommon insert find}. If either \textit{rPath} has popped the LCA node or \textit{iPath} is empty, it clears \textit{iPath} at line~\ref{quadboost: continueFindCommon iPath clear} and pops all nodes above the LCA node from \textit{rPath} at line~\ref{quadboost: continueFindCommon rPath set lca index}. In the end, it calls the findCommon function to locate terminal nodes again at line~\ref{quadboost: continueFindCommon findCommon}. We prove that if \textit{oldKey}'s terminal and \textit{newKey}'s terminal share an LCA node, the common path will never be changed unless the LCA is altered. 

\begin{figure}[t]
\centering
\lstset{escapeinside={(*@}{@*)}}
\begin{sflisting}[language=Java, numberblanklines=false, numbers=left,xleftmargin=2em, basicstyle=\tiny\rmfamily, breaklines=true, breakatwhitespace=true, framexleftmargin=2em]
bool findCommon(Node& il, Node& rl, Internal& lca, Operation& rOp, Operation& iOp, Stack& iPath, Stack& rPath, double oldKeyX, double oldKeyY, double newKeyX, double newKeyY) {(*@\label{quadboost: findCommon}@*)
    while (rl.class == Internal) {(*@\label{quadboost: findCommon check}@*)
        rPath.push(rl);(*@\label{quadboost: findCommon rl push}@*)
        rOp = rl.op;
        getQuadrant(rl, oldKeyX, oldKeyY);
        if (!sameDirection(oldKeyX, oldKeyY, newKeyX, newKeyY, il, rl)) 
            break // check whether two nodes are in the same direction
    }
    lca = rPath.top();(*@\label{quadboost: findCommon lca top}@*)
    find(rl, oldKeyX, oldKeyY, rOp, rPath);  // find oldKey's terminal(*@\label{quadboost: findCommon find rl}@*)
    if (!inTree(rl, oldKeyX, oldKeyY) || moved(rl)) return false;(*@\label{quadboost: findCommon intree rl}@*)
    il = lca;
    find(il, newKeyX, newKeyY, iOp, iPath);  // find newKey's terminal(*@\label{quadboost: findCommon find il}@*)
    if (inTree(il, newKeyX, newKeyY) && !moved(il)) return false;(*@\label{quadboost: findCommon intree il}@*)
}

bool continueFindCommon(Node& il, Node& rl, Internal& lca, Operation& rOp, Operation& iOp, Stack& iPath, Stack& rPath, Internal ip, Internal rp, double oldKeyX, double oldKeyY, double newKeyX, double newKeyY, bool iFail, bool rFail, bool cFail) {(*@\label{quadboost: continueFindCommon}@*)
    if (rFail && !cFail) {  // restart to find oldKey's terminal in cast that the lca's op remains the same(*@\label{quadboost: continueFindCommon rFail}@*)
    	help(rOp);(*@\label{quadboost: continueFindCommon remove help}@*)
    	if (rOp.class != Compress) rl = rp;(*@\label{quadboost: continueFindCommon remove rl rp}@*)
    	else {
    	    while (!rPath.empty()) {
    	    	if (rPath.size() <= indexOf(lca)) {  // reach the lca node, (*@\label{quadboost: continueFindCommon remove lca}@*)
    	    	    cFail = true;  // break the loop the find the common path again
    	    	    break;
    	    	}
    	    	rl = rPath.pop();(*@\label{quadboost: continueFindCommon remove rl pop}@*)
    	    	rOp = rl.op;(*@\label{quadboost: continueFindCommon remove rl op}@*)
    	    	if (rOp.class == Compress) helpCompress(rOp);	(*@\label{quadboost: continueFindCommon remove helpCompress}@*)			
    	    	else break;	
    	    }
    	}
    	if (!cFail) {  // the lca node has not been changed
    	    find(rl, oldKeyX, oldKeyY, rOp, rPath);(*@\label{quadboost: continueFindCommon remove find}@*)
    	    if (!inTree(rl, oldKeyX, oldKeyY) || moved(rl)) return false;(*@\label{quadboost: continueFindCommon intree rl}@*)
    	    else {
    	    	rp = rPath.pop();(*@\label{quadboost: continueFindCommon rp pop}@*)
    	    	return true;
    	    }
    	}
    }
    if (iFail && !cFail) {(*@\label{quadboost: continueFindCommon iFail}@*)
        help(iOp);(*@\label{quadboost: continueFindCommon insert help}@*)
        if (iOp.class != Compress) il = ip;(*@\label{quadboost: continueFindCommon insert il ip}@*)
        else {
            while (!iPath.empty()) {
                il = iPath.pop();(*@\label{quadboost: continueFindCommon insert il pop}@*)
                iOp = il.op;(*@\label{quadboost: continueFindCommon insert il op}@*)
                if (iOp.class == Compress) helpCompress(iOp);(*@\label{quadboost: continueFindCommon insert helpCompress}@*)	
                else break;
            }
        }
        if (iOp.class == Compress) cFail = true;  // check the last op which must be the op of the lca node
        if (!cFail) {
            find(il, newKeyX, newKeyY, iOp, iPath);(*@\label{quadboost: continueFindCommon insert find}@*)
            if (inTree(il, newKeyX, newKeyY) && !moved(il)) return false;(*@\label{quadboost: continueFindCommon intree il}@*)
            else {
    	    	ip = iPath.pop();(*@\label{quadboost: continueFindCommon ip pop}@*)
    	    	return true;
    	    }
        }
    }
    if (cFail) {(*@\label{quadboost: continueFindCommon cFail}@*)
        help(iOp);(*@\label{quadboost: continueFindCommon remove help common}@*)
        help(rOp);(*@\label{quadboost: continueFindCommon insert help common}@*)
        rPath.setIndex(indexOf(lca));  // pop out all nodes above the lca node(*@\label{quadboost: continueFindCommon rPath set lca index}@*)
        iPath.clear();  // clean the input path(*@\label{quadboost: continueFindCommon iPath clear}@*)
        while (!rPath.empty()) {//first time must be not empty
            rl = rPath.pop();(*@\label{quadboost: continueFindCommon rl pop common}@*)
            rOp = rl.op;
            if (rOp.getClass() == Compress) helpCompress(rOp);(*@\label{quadboost: continueFindCommon helpCompress common}@*)
            else break;
        }  // pop up nodes in the common path
        if (!findCommon(il, rl, lca, rOp, iOp, iPath, rPath, oldKeyX, oldKeyY, newKeyX, newKeyY) return false;(*@\label{quadboost: continueFindCommon findCommon}@*)
        else {  // find two nodes
            rp = rPath.pop();(*@\label{quadboost: continueFindCommon rp pop common}@*)
            ip = iPath.pop();(*@\label{quadboost: continueFindCommon ip pop common}@*)      
        }
    }
}
\end{sflisting}
\caption{quadboost findCommon and continueFindCommon}
\label{fig: quadboost findCommon and continueFindCommon}
\end{figure}

\subsection{One Parent Optimization}
In practice, we notice that pushing a whole stack during a traversal is highly expensive. When detecting a failed flag operation, we only have to restart from the parent of a terminal node because we do not change \textit{Internal} nodes unless the compress function erases them from a quadtree. Thus, to reduce the pushing cost, we could only record the parent of the terminal node during a traversal. 

Meanwhile, we have to change the continuous find mechanism. For the insert operation and the remove operation, encountering a \textit{Compress} \textit{op}, we straightforwardly restart from the root. For the move operation, if either \textit{oldKey}'s or \textit{newKey}'s parent is under compression, we restart from the LCA node. If the LCA node also has a \textit{Compress} \textit{op}, we restart from the root. 

\section{Proof Sketch}
In this section, we prove that quadboost is both linearizable and non-blocking, and we propose a lengthy proof in appendix~\ref{sec:appendix}. 

There are four kinds of \textit{basic operations} in quadboost: the insert operation, the remove operation, the move operation, and the contain operation. Other functions are called \textit{subroutines}, which are invoked by basic functions. The insert operation and the remove operation only modify one terminal node, whereas the move operation might operate two different terminals--one for inserting a node with $newKey$, the other for removing a node with $oldKey$. We call them $newKey$'s terminal and $oldKey$'s terminal, and we call their parents $newKey$'s parent and $oldKey$'s parent, respectively. We define $snapshot_{T_{i}}$ as the state of our quadtree at some time, $T_{i}$. 

In our proofs, a CAS that changes a node's $op$ is a \textit{flag} operation and a CAS that changes a node's child is a \textit{replace} operation. Specifically, we use $iflag$, $rflag$, $mflag$, and $cflag$ to denote flag operations for the insert operation, the remove operation, the move operation, and the compress operation separately. Likewise, we use $ireplace$, $rreplace$, $mreplace$, and $creplace$ for replace operations. Moreover, we specify a flag operation that attaches a \textit{Clean} $op$ on a node as an unflag operation. 

First, we present some observations from quadboost. Then, we propose lemmas to show that our subroutines satisfy their pre-conditions and post-conditions because proofs on basic operations depend on these conditions. Next, we demonstrate that there are three categories of successful CAS transitions according to Figure~\ref{fig:State transition diagram}. We also derive some invariants of these transitions. Using above post-conditions of subroutines and invariants, we could demonstrate that a quadtree's structure is maintained during concurrent modifications. In following proofs, we show that quadboost is linearizable because it can be ordered equivalently as a sequential one by its linearization points. In the last part, we prove the non-blocking progress condition of quadboost.

\begin{obse}
The key field of a \textit{Leaf} node is never changed. The $op$ field of a \textit{Leaf} node is initially null. The space information of an \textit{Internal} node is never changed.
\end{obse}
\begin{obse}
The root node is never changed.
\end{obse}
\begin{obse}
A flag operation attaches an $op$ on a \textit{Internal} node.
\end{obse}
\begin{obse}
The $allFlag$ and $iFirst$ field in \textit{Move} are initially false, and they will never be set back after assigning to true. 
\end{obse}
\begin{obse}
If a \textit{Leaf} node is moved, its $op$ is set before the replace operation on $op.iParent$, which is before the replace operation on $op.rParent$.
\end{obse}
\begin{obse}
The help function, the helpCompress function, the helpMove function, and the helpSubstitute function are not called in a mutual way. (If method A calls method B, and method B also calls method A, we say A and B are called in a mutual way.)
\end{obse}

\subsection{Basic Invariants}
We use $find(keys)$ to denote a set of find operations: $find$, $continueFind$, $findCommon$, and $continueFindCommon$. The find function and the continueFind function return a tuple $\langle l, pOp, path\rangle$. The findCommon function and the continueFindCommon function return two such tuples. We specify functions outside the \textit{while} loop in the insert operation, the remove operation, and the move operation are at $iteration_{0}$, and functions inside the \textit{while} loop are at $iteration_{i}, 0 < i$ ordered by their invocation sequence.

We suppose that $find(keys)$ executes from a valid $snapshot_{T_{i}}$ to derive the following conditions. Proofs for conditions of other subroutines are included in appendix~\ref{sec:appendix}.

\begin{lemm}
The post-conditions of $find(keys)$ returned at $T_{i}$, with tuples  $\langle l^{k}, pOp^{k}, path^{k}\rangle$, $0 \leq k < \vert keys \vert$.
\begin{enumerate}
\item $l^{k}$ is a \textit{Leaf} node or an \textit{Empty} node. 
\item At some $T_{i1} < T_{i}$, the top node in $path^{k}$ has contained $pOp^{k}$.
\item At some $T_{i2} < T_{i}$, the top node in $path^{k}$ has contained $l^{k}$.
\item If $pOp^{k}$ is read at $T_{i1}$, and $l^{k}$ is read at $T_{i2}$, then $T_{i1} < T_{i2} < T_{i}$.
\item For each node $n$ in the $path^{k}$, $size(path^{k}) \ge 2$, $n_{t}$ is on the top of $n_{t-1}$, and $n_{t}$ is on the direction $d \in \{nw, ne, sw, se\}$ of $n_{t-1}$ at $T_{i1} < T_{i}$.
\end{enumerate}
\end{lemm}

Based on these post-conditions, we show that each $op$ created at $T_{i}$ store their corresponding information.

\begin{lemm}
For $op$ created at $T_{i}$:
\begin{enumerate}
\item If $op$ is \textit{Substitute}, $op.parent$ has contained $op.oldChild$ that is a \textit{Leaf} node at $T_{i1} < T_{i}$ from the results of $find(keys)$ at the prior iteration.
\item If $op$ is \textit{Compress}, $op.grandparent$ has contained $op.parent$ that is an \textit{Internal} node at $T_{i1} < T_{i}$ from the results of $find(keys)$ at the prior iteration. 
\item If $op$ is \textit{Move}, $op.iParent$ has contained $op.oldIChild$ that is a \textit{Leaf} node and $op.oldIOp$ before $ T_{i}$, and $op.rParent$ has contained $op.oldRChild$ that is a \textit{Leaf} and $op.oldROp$ before $T_{i}$ from the results of $find(keys)$ at the prior iteration.
\end{enumerate}
\end{lemm}

Though we have not presented details of the createNode function, our implementation could guarantee that it has following conditions.

\begin{lemm}
For $createNode(l, p, newKeyX, newKeyY, value)$ that returns a $newNode$ invoked at $T_{i}$, it has the post-condition:
\begin{enumerate}
\item The $newNode$ returned is either a \textit{Leaf} node with $newKey$ and $value$, or a sub-tree that contains both $l.key$ node and $newKey$ node with the same parent.
\end{enumerate}
\end{lemm}

Using prior conditions, we derive some invariants during concurrent executions and prove that there are three kinds of successful CAS transitions. We put successful flag operations that attach $op$s on nodes at the beginning of each CAS transition. We say every successive replace operations that read $op$ belongs to it and follows the flag operations. 
 
Let $flag_{0}$, $flag_{1}$, ..., $flag_{n}$ be a sequence of successful flag operations. $flag_{i}$ reads $pOp_{i}$ and attaches $op_{i}$. $replace_{i}$ and $unflag_{i}$ read $op_{i}$ and follows it. Therefore, we say $flag_{i}$, $replace_{i}$, and $unflag_{i}$ belong to the same $op$. In addition, if there are more than one replace operation belongs to the same $op_{i}$, we denote them as $replace_{i}^{0}$, $replace_{i}^{1}$, ..., $replace_{i}^{n}$ ordered by their successful sequence. Similarly, if there is more than one flag operation that belonging to $op_{i}$ on different nodes, we denote them as $flag_{i}^{0}$, $flag_{i}^{1}$, ..., $flag_{i}^{n}$. A similar notation is used for $unflag_{i}$. 

The following lemmas prove the correct ordering of three different transitions.

\begin{lemm}
For a new node $n$:
\begin{enumerate}
\item It is created with a \textit{Clean} $op$.
\item $rflag$, $iflag$, $mflag$ or $cflag$ succeeds only if $n$'s $op$ is \textit{Clean}.
\item $unflag$ succeeds only if $n$'s $op$ is \textit{Substitute} or \textit{Move}.
\item Once $n$'s $op$ is \textit{Compress}, its $op$ will never be changed.
\end{enumerate}
\end{lemm}

\begin{lemm}
For an \textit{Internal} node $n$, it never reuses an $op$ that has been set previously. 
\end{lemm}

\begin{lemm}
$replace_{i}^{k}$ will not occur before $flag_{i}^{k}$ that belongs to the same $op$ that has been done.
\end{lemm}

\begin{lemm}
The $flag\rightarrow replace \rightarrow unflag$ transition occurs when $rflag_{i}$, $iflag_{i}$ or $mflag_{i}$ succeeds, and it has following properties:
\begin{enumerate}
	\item $replace_{i}$ never occurs before $flag_{i}$.
	\item $flag_{i}^{k}, 0 \leq k < \vert flag_{i} \vert$ is the first successful flag operation on $op_{i}.parent^{k}$ after $T_{i1}$ when $pOp_{i}^{k}$ is read.
	\item $replace_{i}^{k}, 0 \leq k < \vert replace_{i} \vert$ is the first successful replace operation on $op_{i}.parent^{k}$ after $T_{i2}$ when $op_{i}.oldChild^{k}$ is read.
	\item $replace_{i}^{k}, 0 \leq k < \vert replace_{i} \vert$ is the first successful replace operation on $op_{i}.parent^{k}$ that belongs to $op_{i}$.
	\item $unflag_{i}^{k}, 0 \leq k < \vert unflag_{i} \vert$ is the first successful unflag operation on $op_{i}.parent^{k}$ after $flag_{i}^{k}$.
	\item There is no successful unflag operation that occurs before $replace_{i}$.
	\item The first replace operation on $op_{i}.parent^{k}$ that belongs to $op_{i}$ must succeed. 
\end{enumerate} 
\end{lemm}

\begin{lemm}
The $flag\rightarrow replace$ transition occurs only when $cflag_{i}$ succeeds, and it has following properties:
\begin{enumerate}
	\item $creplace_{i}$ never occurs before $cflag_{i}$.
	\item $cflag_{i}$ is the first successful flag operation on $op_{i}.parent$ after $T_{i1}$ when $pOp_{i}$ is read.
	\item $creplace_{i}$ is the first successful replace operation on $op_{i}.grandparent$ after $T_{i2}$ when $op_{i}.parent$ is read.
	\item $creplace_{i}$ is the first successful replace operation on $op_{i}.grandparent$ that belongs to $op_{i}$.
	\item There is no unflag operation after $creplace_{i}$.
	\item The first replace operation on $op_{i}.grandparent$ that belongs to $op_{i}$ must succeed.
\end{enumerate}
\end{lemm}

\begin{lemm}
For the $flag\rightarrow unflag$ transition, it only results from $mflag$ (Suppose $iFirst$ is false) when:
\begin{enumerate}
	\item $unflag_{i}$ is the first successful unflag operation on $op_{i}.rParent$.	
	\item The first flag operation on $op_{i}.iParent$ must fail, and no later flag operation succeeds.
	\item $op_{i}.iParent$ and $op_{i}.rParent$ are different.
\end{enumerate}
\end{lemm}

\begin{claim}
There are three kinds of successful transitions belong to an $op$: (1) $flag\rightarrow replace\rightarrow unflag$, (2) $flag\rightarrow unflag$, (3) $flag\rightarrow replace$.
\end{claim}

Then, we prove that a quadtree maintains its properties during concurrent modifications.

\begin{defin}
Our quadtree has the following properties:
\begin{enumerate}
\item Two layers of dummy \textit{Internal} nodes are never changed.
\item An \textit{Internal} node $n$ has four children located in the direction $d \in \{nw, ne, sw, se\}$ according to their $\langle x, y, w, h\rangle$, or $\langle keyX, keyY\rangle$. 

For \textit{Internal} nodes residing on four directions:
\begin{itemize}
\item $n.nw.x = n.x$, $n.nw.y = n.y$; 
\item $n.ne.x = n.x + w/2$, $n.ne.y = n.y$;
\item $n.sw.x = n.x$, $n.sw.y = n.y + n.h/2$;
\item $n.se.x = n.x + w/2$, $n.se.y = n.y+h/2$,
\end{itemize}
All children have their $w' = n.w / 2$, $h'=n.h/2$. 

For \textit{Leaf} nodes residing on four directions:
\begin{itemize}
\item 
$n.x \leq n.nw.keyX < n.x + n.w/2$,\\
$n.y \leq n.nw.keyY < n.y + n.h / 2$; 
\item
$n.x + n.w / 2 \leq n.ne.keyX < n.x + n.w$,\\
$n.y \leq n.ne.keyY < n.y + n.h / 2$;
\item
$n.x \leq n.sw.keyX < n.x + n.w/2$,\\
$n.y + n.h/2 \leq n.sw.keyY < n.y + n.h$;
\item 
$n.x + n.w / 2 \leq n.se.keyX < n.x + n.w$,\\
$n.y + n.h / 2 \leq n.se.keyY < n.y + n.h$.
\end{itemize}
\end{enumerate}
\end{defin}

To help clarify a quadtree's properties during concurrent executions, we define \textit{active} set and \textit{inactive} set for different kinds of nodes. For an \textit{Internal} node or an \textit{Empty} node, if it is reachable from the root in $snapshot_{T_{i}}$, it is active; otherwise, it is inactive. For a \textit{Leaf} node, if it is reachable from the root in $snapshot_{T_{i}}$ and not moved, it is active; otherwise, it is \textit{inactive}. We say a node $n$ is moved in $snapshot_{T_{i}}$ if the function moved(n) returns true at $T_{i}$. We denote $path(keys^{k}), 0 \leq k < \vert replace_{i} \vert$ as a stack of nodes pushed by $find(keys)$ in a snapshot. We define $physical\_path(keys^{k})$ to be the path for $keys^{k}$ in $snapshot_{T_{i}}$, consisting of a sequence of \textit{Internal} nodes with a \textit{Leaf} node or an \textit{Empty} node at the end. We say a subpath of $path(keys^{k})$ is an $active\_path$ if all nodes from the root to node $n \in path(keys^{k})$ are active. Hence a $physical\_path(keys^{k})$ is active only if the end node is not moved.

\begin{lemm}
Two layers of dummy nodes are never changed.
\end{lemm}

\begin{lemm}
Children of a node with a \textit{Compress} $op$ will not be changed.
\end{lemm}

\begin{lemm}
Only an \textit{Internal} node with all children \textit{Empty} could be attached with a \textit{Compress} $op$.
\end{lemm}

\begin{lemm}
An \textit{Internal} node whose $op$ is not \textit{Compress} is active.
\end{lemm}

\begin{lemm}
After the invocation of $find(keys)$ which reads $l^{k}$, there is a snapshot in which the path from the root to $l^{k}$ is $physical\_path(keys^{k})$.
\end{lemm}

\begin{lemm}
After $ireplace$, $rreplace$, $mreplace$, and $creplace$, a quadtree's properties remain.
\end{lemm}
\begin{proof}
We shall prove that in any $snapshot_{T_{i}}$, a quadtree's properties remain.

First, Lemma 10 shows that two layers of dummy nodes remain in the tree. We have to consider other layers of nodes that are changed by replace operations.

Consider $creplace$ that replaces an \textit{Internal} node by an \textit{Empty} node. Because the \textit{Internal} node has been flagged on a \textit{Compress} $op$ before $creplace$ (Lemma 8), all of its children are \textit{Empty} and not changed (Lemma 11 and Lemma 12). Thus, $creplace$ does not affect the second claim of Definition 1. 

Consider $ireplace$, $rreplace$, or $mreplace$ that replaces a terminal node by an \textit{Empty} node, a \textit{Leaf} node, or a sub-tree. By Lemma 7, before $replace^{k}$, $op.parent^{k}$ is flagged with $op$ such that no successful replace operation could happen on $op.parent^{k}$. Therefore, if the new node is \textit{Empty}, it does not affect the tree property. If the new node is a \textit{Leaf} node or a sub-tree, based on the post-conditions of the createNode function (Lemma 3), after replace operations the second claim of Definition 1 still holds.
\end{proof}

\begin{claim}
A quadtree maintains its properties in every snapshot.
\end{claim}

\subsection{Linearizability}
In this section, we define linearization points for basic operations. As the compress function is included in the move operation and the remove operation that return true, it does not affect the linearization points of these operations. If an algorithm is linearizable, its result be ordered equivalently as a sequential history by the linearization points. Since all modifications depend on $find(keys)$, we first point out its linearization point. For $find(keys)$, we define its linearization point at $T_{i}$ such that $l^{k}$ returned is on the $physical\_path(keys^{k})$ in $snapshot_{T_{i}}$. 

For the contain operation that returns true, we show that there is a corresponding snapshot in which $l^{k}$ in $physical\_path(keys^{k})$ is active. For the contain operation, the insert operation, the remove operation, or the move operation that returns false, we show there is a corresponding snapshot in which $l^{k}$ in $physical\_path(keys^{k})$ is inactive. For the insert, remove, and move operations that returns true, we define the linearization point to be its first successful replace operation---$replace_{i}^{0}$. To make a reasonable demonstration, we first show that $replace_{i}^{k}, 0 \leq k < \vert replace_{i} \vert$ belongs to each operation that creates $op_{i}$ and illustrate that $op$ is unique for each operation.

\begin{lemm}
For $find(keys)$ that returns tuples $\langle l^{k}$, $pOp^{k}$, $path^{k}\rangle$, there is a $snapshot_{T_{i}}$ such that $path^{k}$ returned with $l^{k}$ at the end is $physical\_path(key^{k})$ in $snapshot_{T_{i}}$.
\end{lemm}

\begin{lemm}
If the insert operation, the remove operation, or the move operation returns true, the first successful replace operation occurs before returning, and it belongs to the $op$ created by the operation itself at the last iteration in the \textit{while} loop.
\end{lemm}

\begin{lemm}
If the insert operation, the remove operation, or the move operation returns false, there is no successful $replace$ happens during the execution.
\end{lemm}

The next lemma points out the \textbf{linearization points} of the contain operation.

\begin{lemm}
For the contain operation that returns true, there is a corresponding snapshot that $l^{k}$ in $physical\_path(keys^{k})$ is active in $snapshot_{T_{i}}$. For the contain operation that returns false, there is a corresponding snapshot that $l^{k}$ in $physical\_path(keys^{k})$ is inactive in $snapshot_{T_{i}}$.
\end{lemm}

We list out the \textbf{linearization points} of other operations as follows:

\begin{itemize}
\item \textit{insert(key)}. The linearization point of the insert operation that returns false is $T_{i}$ after calling $find(key)$ at which $l$ at the end of $physical\_path(key)$ contains the key and is not moved. For the insert operation that returns true, the linearization point is at the first successful replace operation (line~\ref{quadboost: helpSubstitute helpReplace}).

\item \textit{remove(key)}. The linearization point of the remove operation that returns false is $T_{i}$ after calling $find(key)$ at which $l$ at the end of $physical\_path(key)$ does not contain the key or is moved. For the successful remove operation, we define the linearization point where the node with key is replaced by an \textit{Empty} node (line~\ref{quadboost: helpSubstitute helpReplace}).

\item \textit{move(oldKey, newKey)}. For the unsuccessful move operation, the linearization point depends on both $newKey$ and $oldKey$. If $rl$ at the end of $physical\_path(oldKey)$ does not contain $oldKey$ or is moved, the linearization point is $T_{i1}$ after calling $find(keys)$. Or else, if $il$ at the end of $physical\_path(newKey)$ contains $newKey$ and is not moved, the linearization point is at $T_{i2}$ after calling $find(keys)$.

For the successful move operation, the linearization point is the first successful replace operation. (line~\ref{quadboost: helpMove helpReplace first} or line~\ref{quadboost: helpMove helpReplace first first})
\end{itemize}

\begin{claim}
quadboost is linearizable.
\end{claim}

\subsection{Non-blocking}
Finally, we prove that quadboost is non-blocking, which means that the system as a whole is making progress even if some threads are starving.

\begin{lemm}
A node with a \textit{Compress} $op$ will not be pushed into $path$ more than once.
\end{lemm}
\begin{lemm}
For $path(keys^{k})$, if $n_{t}$ is active in $snapshot_{T_{i}}$, then $n_{0}, ... , n_{t-1}$ pushed before $n_{t}$ are active.
\end{lemm}
\begin{lemm}
If $n$ is the LCA node in $snapshot_{T_{i}}$ on $physical\_path$ for $oldKey$ and $newKey$. Then at $T_{i1}, T_{i1} > T_{i}$, $n$ is still the LCA on $active\_path$ for $oldKey$ and $newKey$ if it is active.
\end{lemm}
\begin{lemm}
$path^{k}$ returned by $find(keys)$ consists of finite number of keys.
\end{lemm}
\begin{lemm}
There is a unique \textit{spatial order} among nodes in a quadtree in every snapshot.
\end{lemm}
\begin{lemm}
There are a finite number of successful $flag\rightarrow replace \rightarrow$, $flag\rightarrow replace$, $flag\rightarrow unflag$ transitions.
\end{lemm}
\begin{lemm}
If the help function returned at $T_{i}$, and $find(keys)$ at the prior iteration reads $p^{k}.op$ at $T_{i1} < T_{i}$, keys in $snapshot_{T_{i}}$ and $snapshot_{T_{i1}}$ are different.
\end{lemm}

\begin{claim}
quadboost is non-blocking.
\end{claim}
\begin{proof}
We have to prove that no process will execute loops infinitely without changing keys in a quadtree. First, we prove that $path$ is terminable. Next, we prove that $find(keys^{k})$ starts from an active node in $physical\_path(keys^{k})$ in $snapshot_{T_{i}}$ between $i_{th}$ iteration and $i+1_{th}$ iteration. Finally, some \textit{Leaf} nodes in $snapshot_{T_{i1}}$ at the returning of $find(keys)$ at $i_{th}$ iteration are different from $snapshot_{T_{i2}}$ at the returning of $find(keys)$ at $i+1_{th}$ iteration.

For the first part, we initially start from the root node. Therefore, $path$ is empty. Moreover, as Lemma 22 shows that $path^{k}$ consists of finite number of keys, we establish this part.

For the second part, the continueFind function and the continueFindCommon function pop all nodes with \textit{Compress} $op$ from $path$. For the insert operation and remove operation, since Lemma 13 shows that an \textit{Internal} nodes whose $op$ is not \textit{Compress} is active, and Lemma 21 shows that nodes above the active node are also active, there is a snapshot in which the top node of $path$ is still in $physical\_path(keys^{k})$. For the move operation, if either $rFail$ or $iFail$ is true, it is equivalent with the prior case. If $cFail$ is true, Lemma 22 illustrates that if the LCA node is active, it is in $physical\_path$ for both $oldKey$ and $newKey$. Thus, there is also a snapshot that the start node is in $physical\_path$.

The third part is proved by contradiction, assuming that a quadtree is stabilized at $T_{i}$, and all invocations after $T_{i}$ are looping infinitely without changing \textit{Leaf} nodes.

For the insert operation and the remove operation, before the invocation of $find(keys)$ at the next iteration, they must execute the help function at line~\ref{quadboost: insert help} and line~\ref{quadboost: remove help} accordingly. In both cases, the help function changes the snapshot (Lemma 26).

For the move operation, different situations of the continueFindCommon function are considered. If $rFail$ or $iFail$ is true, two situations arise: (1) $iOp$ or $rOp$ is \textit{Clean} but $mflag$ fails; or (2) $iOp$ or $rOp$ is not \textit{Clean}. In both cases, before the invocation of $find(keys^{k})$, the help function is performed at line~\ref{quadboost: continueFindCommon insert help} and line~\ref{quadboost: continueFindCommon remove help}. Thus by Lemma 26, the snapshot is changed between two iterations. Now consider if $cFail$ is true. If $ip$ and $rp$ are the same, it could result from the difference between $rOp$ and $iOp$. In this case, the snapshot might be changed between reading $rOp$ and $iOp$. It could also result from the failure of $mflag$. For the above cases, the help function at line~\ref{quadboost: continueFindCommon remove help common} would change the quadtree. If $ip$ and $rp$ are different, it results from either $iPath$ or $rPath$ that have popped the LCA node. We have proved the case in which either $iFail$ or $rFail$ is true and derives a contradiction.

From the above discussions, we prove that quadboost is non-blocking.
\end{proof}

\section{Evaluation}

We ran experiments on a machine with 64GB main memory and two 2.6GHZ Intel(R) Xeon(R) 8-core E5-2670 processors with hyper-threading enabled, rendering 32 hardware threads in total. We used the RedHat Enterprise Server 6.3 with Linux core 2.6.32, and all experiments were ran under Sun Java SE Runtime Environment (build 1.8.0\_65). To avoid significant run-time garbage collection cost, we set the initial heap size to 6GB.

For each experiment, we ran eight 1-second cases, where the first 3 cases were performed to warm up JVM, and the median of the last 5 cases was used as the real performance. Before the start of each case, we inserted half keys from the key set into a quadtree to guarantee that  the insert and the remove operation have equal success opportunity initially. 

We apply uniformly distributed key sets that contain two-dimensional points within a square. We use the $range$ to denote the border of a square. Thus, points are located inside a $range * range$ square. In our experiments, we used two different key sets: $10^{2}$ keys to measure the performance under high contention, and $10^{6}$ keys to measure the performance under low contention. For simplicity, we let the $range$ of the first category experiment be $10$, rendering $1-10^{2}$ consecutive keys for one-dimensional structure. For the second category experiment, we let the $range$ be $1000$, generating $1-10^{6}$ consecutive keys.

\begin{table}[htbp]
\caption{The concurrent quadtree algorithms with different optimization strategies.}
\label{tab:quadtree}
\centering
\tiny
\begin{tabular}{|c|c|c|c|}
\hline
 \textbf{type} & \textbf{insert} & \textbf{remove} & \textbf{move}\\
\hline
 \textbf{qc} & single CAS & single CAS & \textbf{not support}\\
\hline
 \textbf{qb-s} & \begin{tabular}{@{}c@{}} flag CAS, \\continuous find, \\stack \end{tabular}& \begin{tabular}{@{}c@{}}flag CAS, \\continuous find, \\decoupling stack, \\recursive compression\end{tabular} & \begin{tabular}{@{}c@{}}flag CAS, \\continuous find, \\decoupling stack, \\recursive compression\end{tabular}\\
\hline
 \textbf{qb-o} & \begin{tabular}{@{}c@{}}flag CAS,\\ continuous find  \end{tabular}& \begin{tabular}{@{}c@{}}flag CAS, \\continuous find, \\decoupling compression\end{tabular}& \begin{tabular}{@{}c@{}}flag CAS,\\continuous find, \\decoupling compression\end{tabular}\\
\hline
\end{tabular}
\end{table}

We evaluate quadboost algorithms via comparing with the state-of-the-art concurrent trees (kary, ctrie, and patricia) for throughput (Section~\ref{subsec:throughput}) and presenting the incremental effects of the optimization strategies proposed in this work (Section~\ref{subsec:analysis}). Table~\ref{tab:quadtree} lists the concurrent quadtree algorithms. \textbf{qb-o} (quadboost-one parent) is the one parent optimization based on \textbf{qb-s} (quadboost-stack) mentioned in Section 4.4. \textbf{qc} is the CAS quadtree introduced in Section 3.

\subsection{Throughput}
\label{subsec:throughput}
To the best of our knowledge, a formal concurrent quadtree has not been published yet. Hence, we compare our quadtrees with three one-dimensional non-blocking trees:
\begin{itemize}
\item \textbf{kary} is a non-blocking k-way search tree, where $k$ represents the number of branches maintained by an internal node. Like the non-blocking BST~\cite{ellen2010non}, keys are kept in leaf nodes. When $k=2$, the structure is similar as the non-blocking BST; when $k=4$, each internal node has four children, it has a similar structure to quadtree. However, kary's structure depends on the modification order, and it does not have a series of internal nodes representing the two-dimensional space hierarchy.
\item \textbf{ctrie} is a concurrent hash trie, where each node can store up to $2^{k}$ children. We use $k=2$ to make a 4-way hash trie that resembles quadtree. The hash trie also incorporates a compression mechanism to reduce unnecessary nodes. Different from quadtree, it uses a control node (\texttt{INODE} as the paper indicates) to coordinate concurrent updates. Hence, the search depth could be longer than quadtree.
\item \textbf{patricia} is a binary search tree, which adopts Ellen's BST techniques~\cite{ellen2010non}. As the author points out, it can be used as quadtree by interleaving the bits of $x$ and $y$. It also supports the move (replace) operation like quadboost. Unlike our LCA-based operation, it searches two positions separately without a continuous find mechanism. 
\end{itemize}

Since the above structures only store one-dimensional keys, we had to transform a two-dimensional key into a one-dimensional key for comparison. Though patricia could store two-dimensional keys using the above method, ctrie and kary cannot use it. Thus, we devised a general formula: $key^{1} = key^{2}_{x} * range + key^{2}_{y}$. Given a two-dimensional key-$key^{2}$, and $range$, we transformed it into a one-dimensional key-$key^{1}$. To refrain from trivial transformations by floating numbers, we only considered integer numbers in this section.

\begin{figure}[htbp]
\centering
\begin{subfigure}[b]{0.25\textwidth}
\centering
\includegraphics[width=\linewidth]{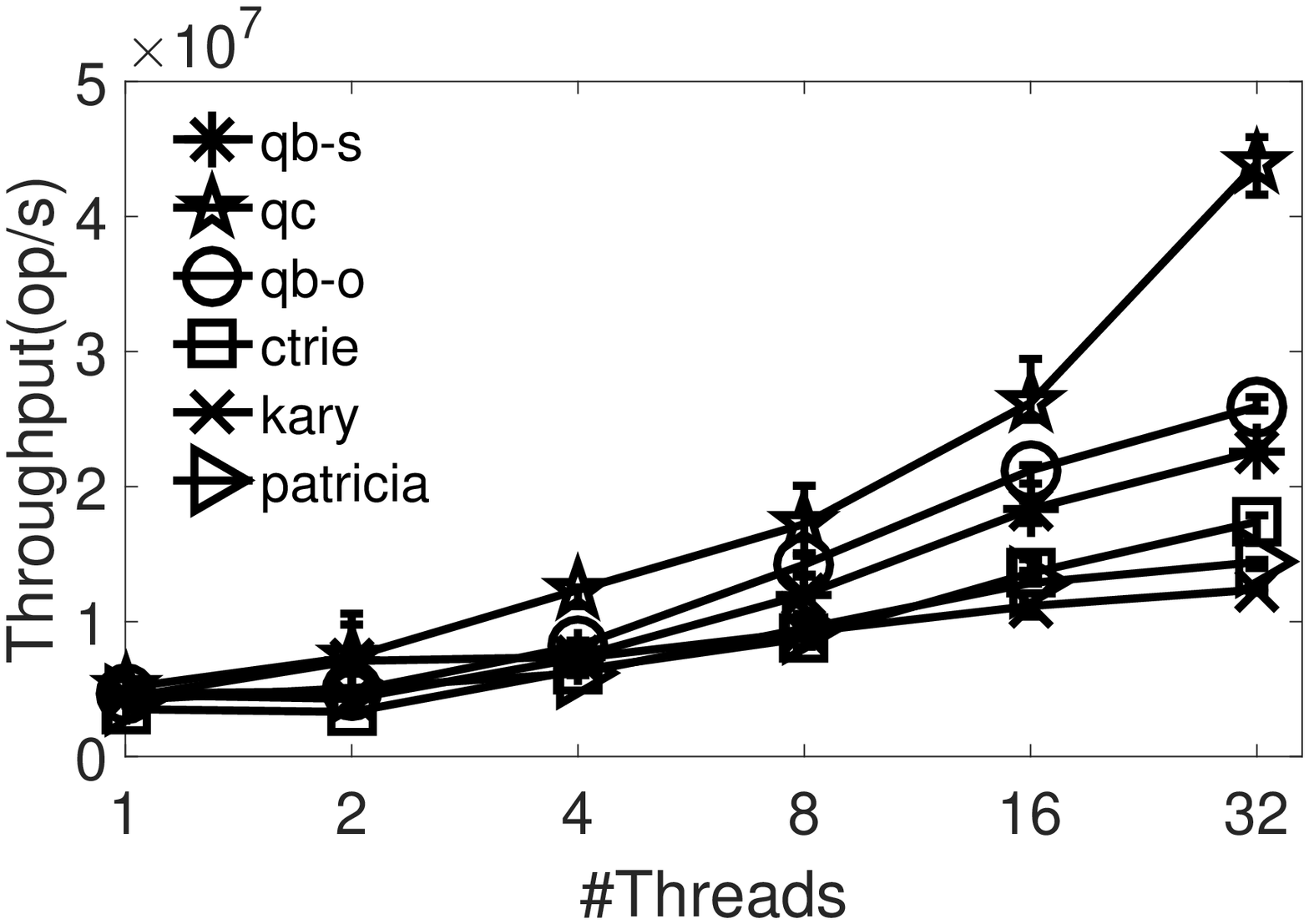}
\caption{$10^{2}$ keys}
\label{fig:hybrid_small}
\end{subfigure}%
\begin{subfigure}[b]{0.25\textwidth}
\centering
\includegraphics[width=\linewidth]{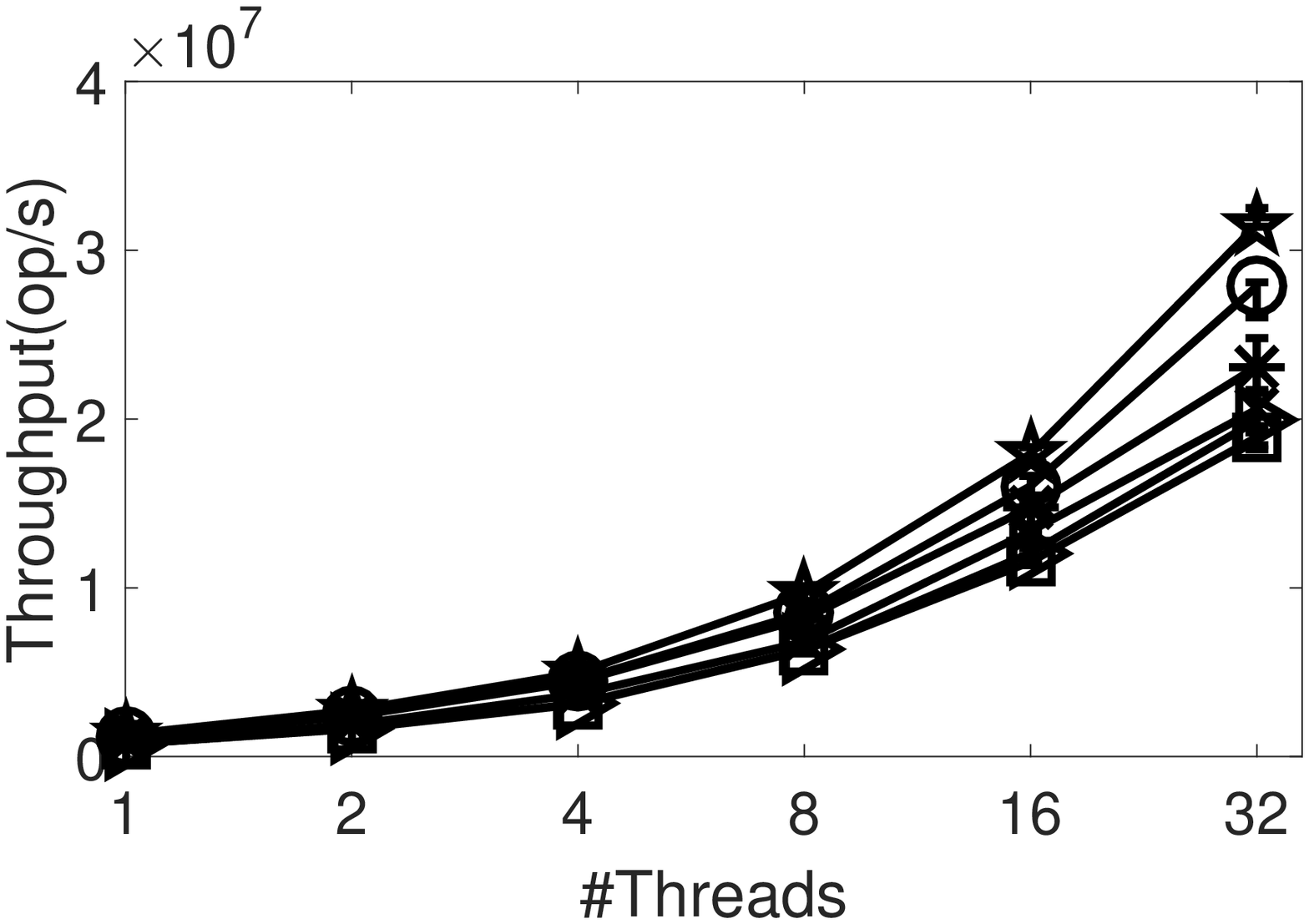}
\caption{$10^{6}$ keys}
\label{fig:hybrid_large}
\end{subfigure}
\caption{Throughput of different concurrent trees under both high and low contention (50\% \textit{insert}, 50\% \textit{remove}). }
\label{fig:hybrid}
\end{figure}

Due to the lack of the move operation in CAS quadtree (\textbf{qc}), we compare throughput with/without the move operation respectively. Figure~\ref{fig:hybrid} plots throughput without any move operation for the concurrent algorithms. It is not surprising to observe that \textbf{qc} achieves the highest throughput. To some extent, \textbf{qc} represents an upper bound of throughput because it maintains the hierarchy without physical removal, i.e., its remove operation only applies a CAS on edges to change links, which leads to less contention than other practical concurrent algorithms. However, both \textbf{qb-s} and \textbf{qb-o} can achieve comparable throughput when the key set becomes larger. This phenomenon occurs due to: (i) given the large key set, fewer thread interventions results in a lesser number of CAS failures on nodes; and (ii) both algorithms compress nodes and use the continuous find mechanism to reduce the length of traverse path.

As a comparison, \textbf{ctrie}, \textbf{kary}, and \textbf{patricia} show lower performances with the increasing number of threads. For instance, in Figure~\ref{fig:hybrid_small} at 32 threads, \textbf{qb-o} outperforms \textbf{ctrie} by 49\%, \textbf{patricia} by 79\%, and \textbf{kary} by 109\%. Note that \textbf{qb-o} and \textbf{qb-s} incorporate the continuous find mechanism to reduce the length of traverse path. Further, both \textbf{kary} and \textbf{patricia} flag the grandparent node in the remove operation, which allows less concurrency than \textbf{ctrie}, \textbf{qb-o} and \textbf{qb-s} with the decoupling approach shown in Figure~\ref{fig:decoupling}. \textbf{qb-s} is worse than \textbf{qb-o} due to its extra cost of recording elements and compressing nodes recursively. Figure~\ref{fig:hybrid_large} exhibits results from when the key set was large. There is less collision among threads but deeper depths of trees than the smaller key set. In the scenario, \textbf{qb-o} and \textbf{qb-s} show a similar performance as \textbf{qc} because of less number of CAS failures caused by thread interventions. \textbf{qb-o} is only 12\% worse than \textbf{qc} at 32 threads, but 47\% better than \textbf{ctrie}, 35\% better than \textbf{kary}, and 39\% better than \textbf{patricia} mainly for its shorter traversal paths caused by its static representation and the continuous find mechanism. As we discussed in the next section, \textbf{qb-s} and \textbf{qb-o} save a significant number of nodes as shown in Figure~\ref{fig:nodes}. It implies that \textbf{qb-s} and \textbf{qb-o} occupy less memory and result in a shorter path for traversal.

\begin{figure}[htbp]
\centering
\begin{subfigure}[t]{0.25\textwidth}
\includegraphics[width=\linewidth]{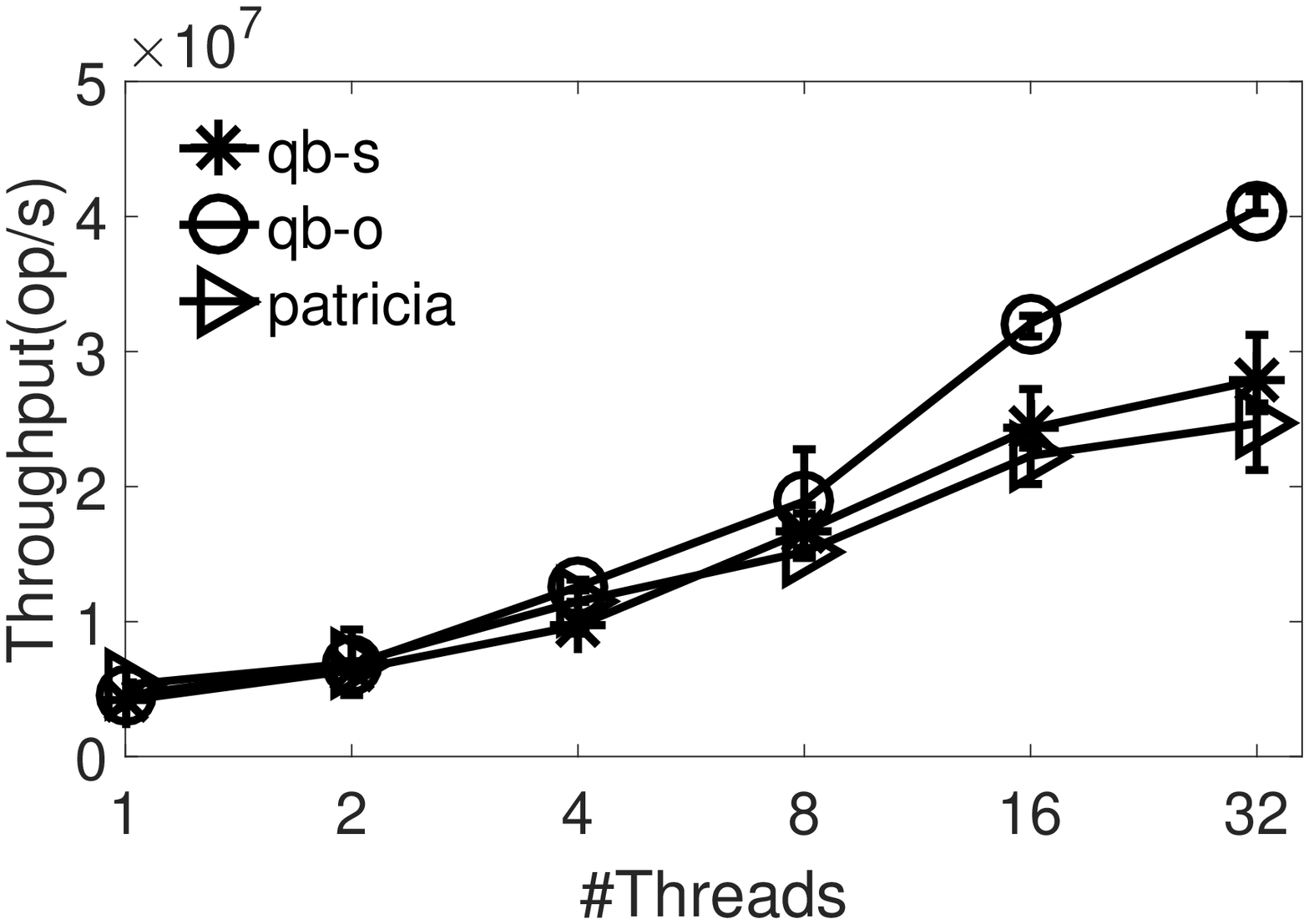}
\caption{$10^{2}$ keys}
\label{fig:move_small}
\end{subfigure}%
\begin{subfigure}[t]{0.25\textwidth}
\includegraphics[width=\linewidth]{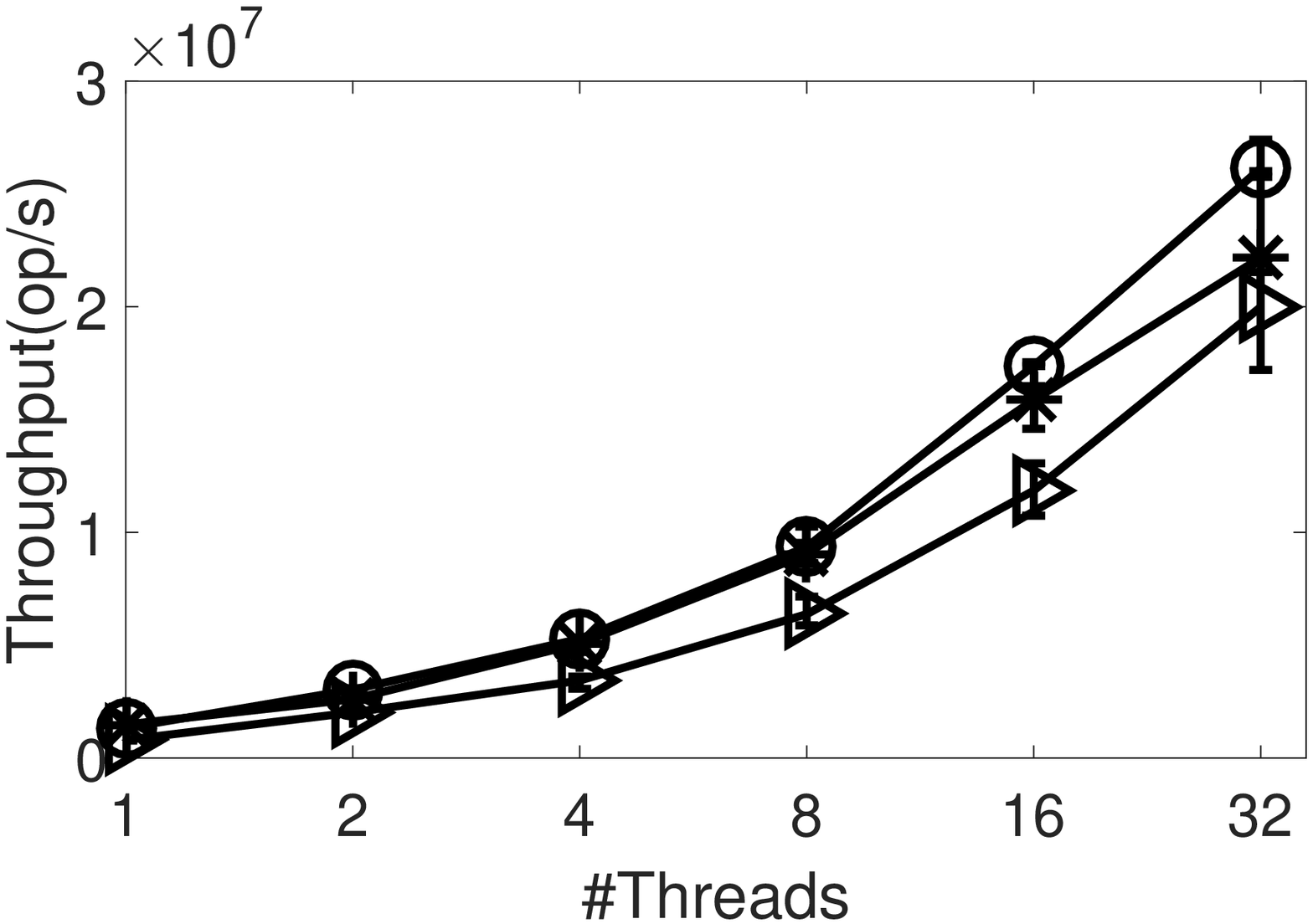}
\caption{$10^{6}$ keys}
\label{fig:move_large}
\end{subfigure}
\caption{Comparison of the move operation's throughput between quadboost and patricia in both small and large ranges (10\% \textit{insert}, 10\% \textit{remove}, 80\% \textit{move}). }
\label{fig:move}
\end{figure}

Figure~\ref{fig:move_large} demonstrates that quadboost has an efficient move operation. Using the small key set, where the depth is not a significant impact, Figure~\ref{fig:move_small} shows that \textbf{qb-o} is more efficient than \textbf{patricia} especially when contention is high. For example, it performs better than \textbf{patricia} by 47\% at 32 threads because it adopts the continuous find mechanism to traverse less path and decouples physical adjustment for higher concurrency. However, \textbf{qb-s} is similar to \textbf{patricia} since it has to maintain a stack and recursively compress nodes in a quadtree. Figure~\ref{fig:move_large} illustrates that \textbf{qb-s} and \textbf{qb-o} have a similar throughput for the large key set. \textbf{qb-o} outperforms \textbf{patricia} by 31\% at 32 threads. When the key set is large, the depth becomes a more significant factor due to less contention. Since each \textit{Internal} node in a quadtree maintains four children while \textbf{patricia} maintains two, the depth of \textbf{patricia} is deeper than quadboost. Further, the combination of the LCA node and the continuous find mechanism ensures that \textbf{qb-o} and \textbf{qb-s} do not need to restart from the root even if flags on two different nodes fail.
\subsection{Analysis}
\label{subsec:analysis}

To determine how quadboost algorithms improve the performance,  we devised two algorithms that incrementally use parts of techniques in \textbf{qb-o}:
\begin{itemize}
\item \textbf{qb-f} flags the parent of a terminal node in the move operation, the insert operation, and the remove operation. It restarts from the root without a continuous find mechanism. Further, it adopts the traditional remove mechanism mentioned in Figure~\ref{fig:BST's remove}.
\item \textbf{qb-d} decouples the physical adjustment in the remove operation based on \textbf{qb-f}.
\end{itemize}

$range$ here was set to $2^{32}-1$, and both $key_{x}$ and $key_{y}$ could be floating numbers. We used an insert dominated experiment and a remove dominated experiment to demonstrate the effects of different techniques. We used a remove dominated experiment to show the effect of decoupling, where {\em insert:remove} ratio was 1:9. In the insert dominated experiment, the {\em insert:remove} ratio was 9:1; hence, there were far more insert operations. Since fewer compress operations were induced, the experiment showed the effect of the continuous find mechanism.
\begin{figure}[htbp]
\centering
\begin{subfigure}[t]{0.25\textwidth}
\includegraphics[width=\linewidth]{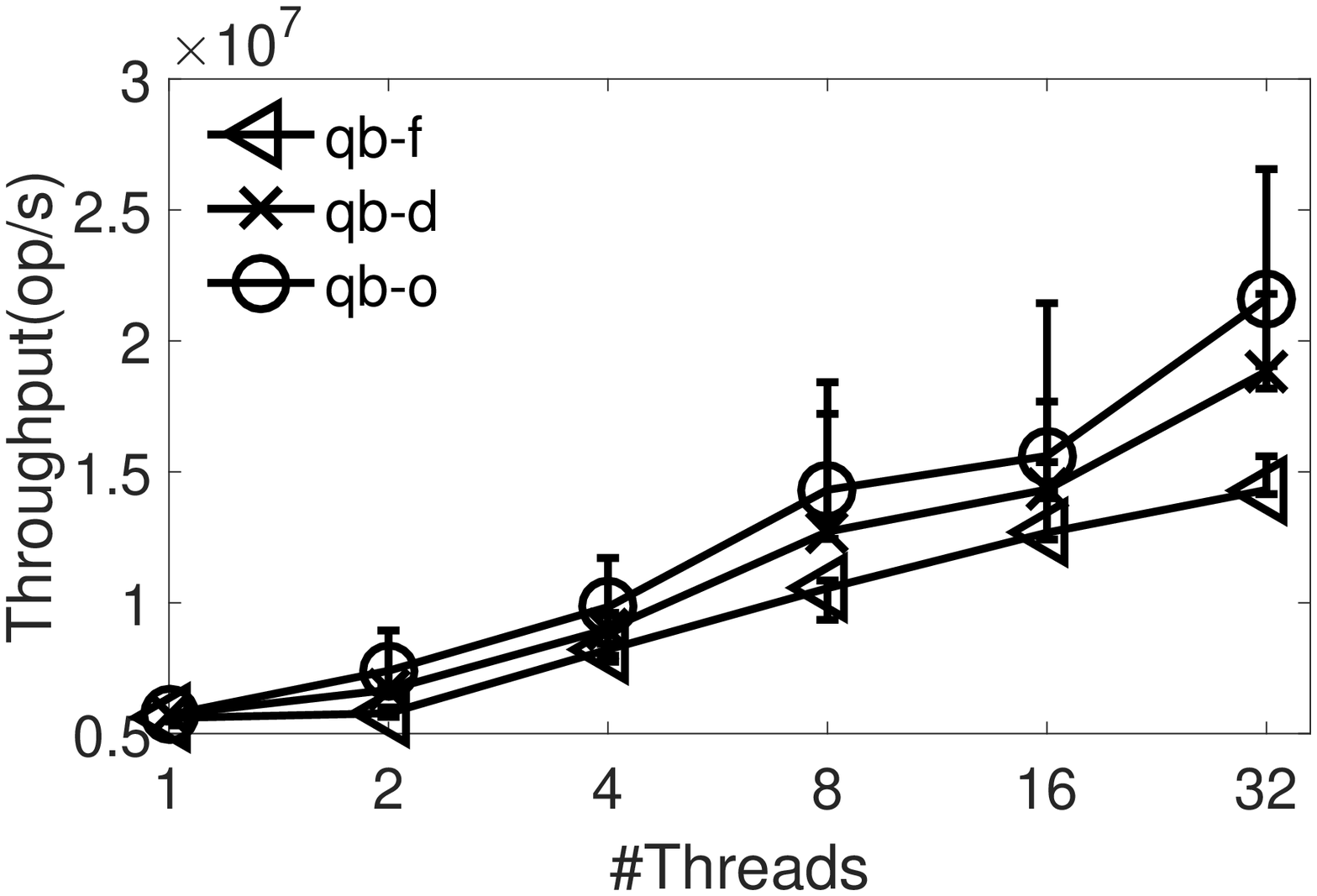}
\caption{10\% \textit{insert}, 90\% \textit{remove}}
\label{fig:dominate_remove}
\end{subfigure}%
\begin{subfigure}[t]{0.25\textwidth}
\includegraphics[width=\linewidth]{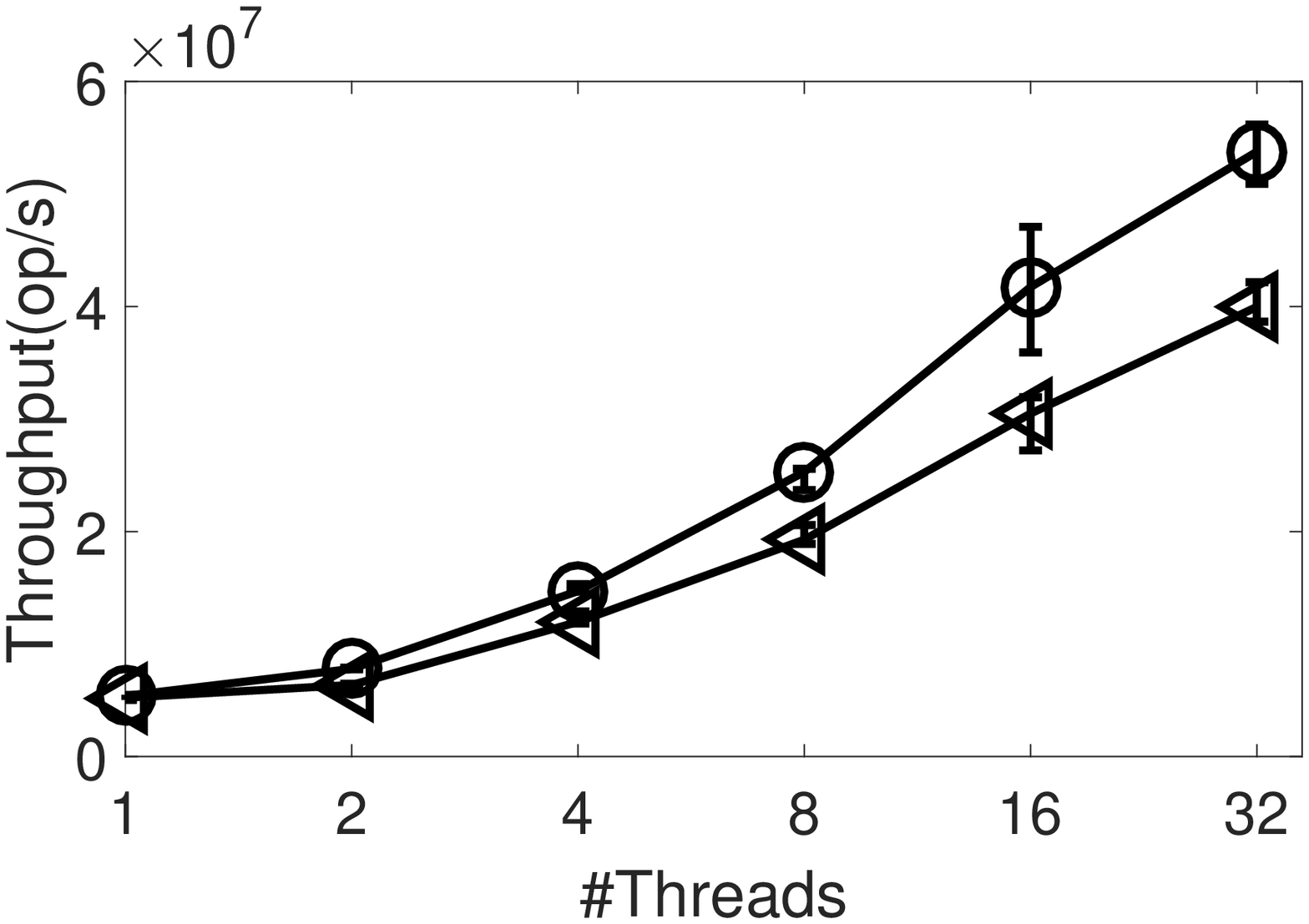}
\caption{90\% \textit{insert}, 10\% \textit{remove}}
\label{fig:dominate_insert}
\end{subfigure}
\caption{Throughput comparison in insert dominated and remove dominated cases ($10^{2}$ keys). }
\label{fig:dominate}
\end{figure}
Figure~\ref{fig:dominate_remove} illustrates that quadtrees with decoupling exhibit a higher throughput than \textbf{qb-f}, the basic flag concurrent quadtree. Besides, \textbf{qb-o} which incorporates the continuous find is more efficient than \textbf{qb-d}. Specifically, at 32 threads, \textbf{qb-o} performs 15\% better than \textbf{qb-d} and 51\% better than \textbf{qb-f}. From Figure~\ref{fig:dominate_insert}, we determine that \textbf{qb-o} outperforms \textbf{qb-f} by up to 35\%. Therefore, it demonstrates that the continuous find mechanism and the decoupling approach play a significant role in our algorithm.

\begin{figure}[htbp]
\centering
\begin{subfigure}[t]{0.25\textwidth}
\includegraphics[width=\linewidth]{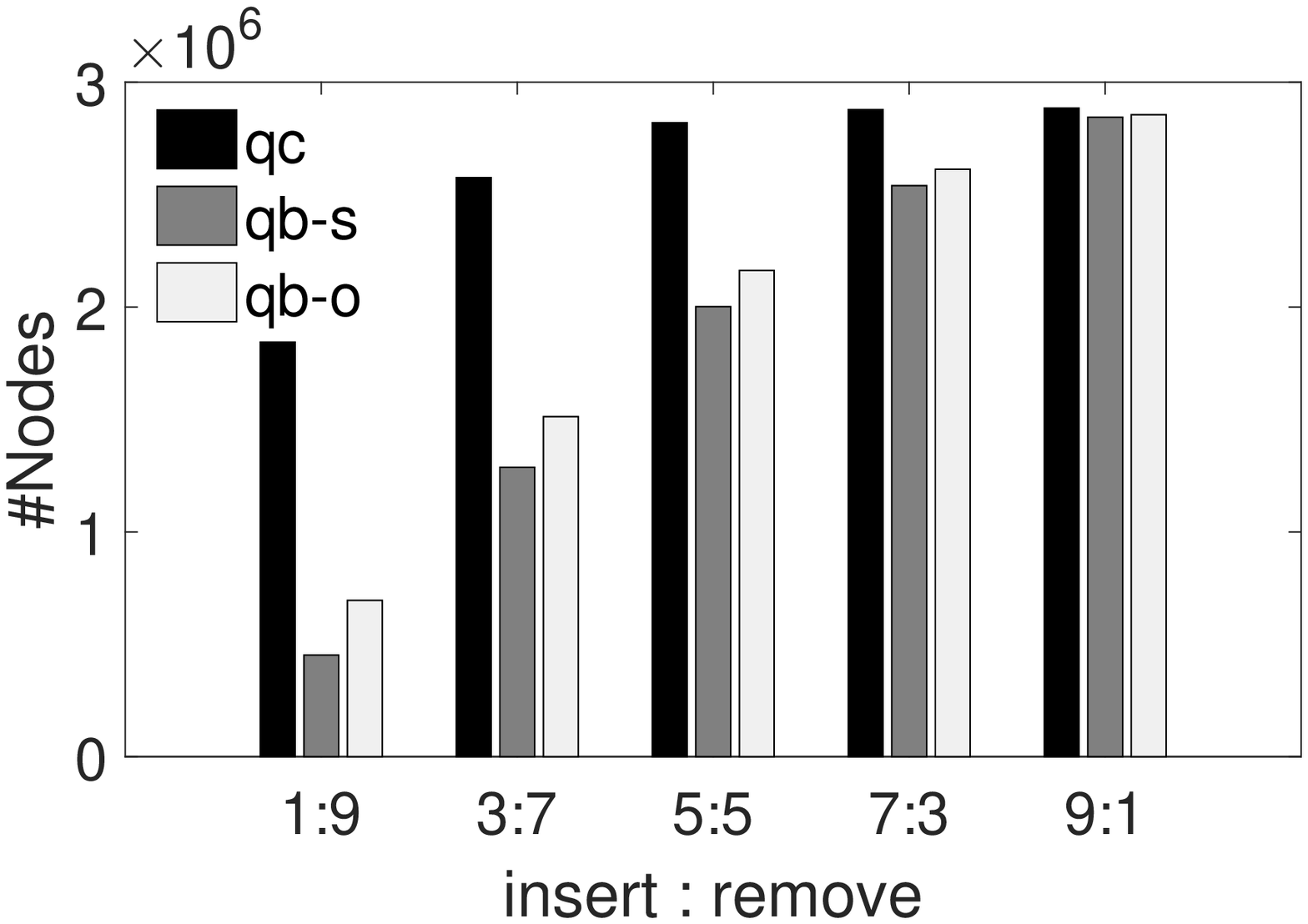}
\caption{Nodes}
\label{fig:nodes nodes}
\end{subfigure}%
\begin{subfigure}[t]{0.25\textwidth}
\includegraphics[width=\linewidth]{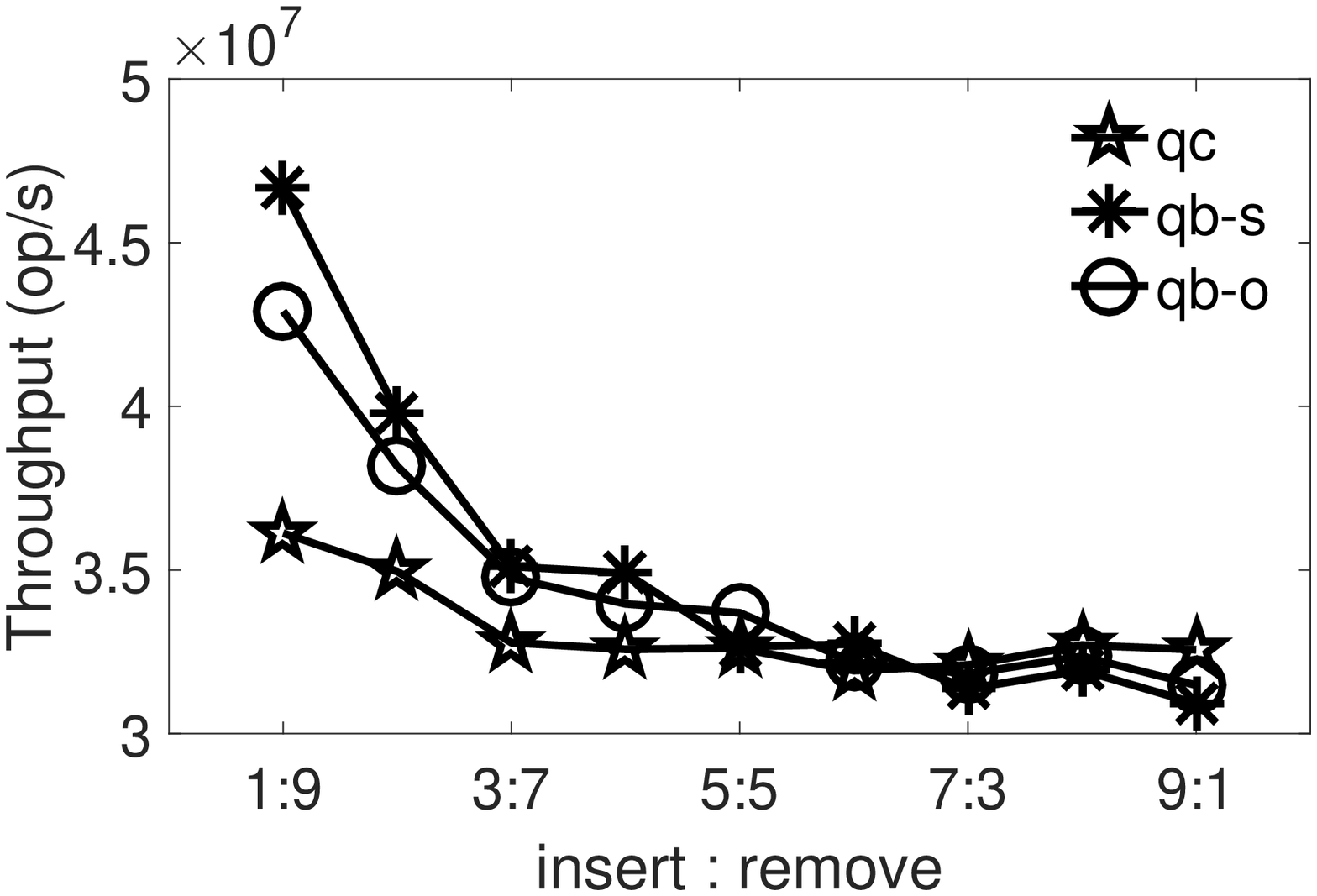}
\caption{Throughput}
\label{fig:nodes throughput}
\end{subfigure}
\caption{Number of nodes left and throughput under different ratios of insert:remove from 9:1 to 1:9, with fixed 90\% contain, under 32 threads and $10^{6}$ keys.}
\label{fig:nodes}
\end{figure}

Another advantage of quadboost results from the compression technique, which reduces the search path for each operation and the memory consumption. Figure~\ref{fig:nodes}\footnote{\footnotesize Unlike previous experiments, we run eight 3-second cases in the experiment to ensure stable amount of modifications} plots the number of nodes left and the throughput of each quadtree at different {\em insert:remove} ratios. As \textbf{qc} only replaces a terminal node with an \textit{Empty} node without compression, it results in the greatest number of nodes in the memory (Figure~\ref{fig:nodes nodes}). In contrast, \textbf{qb-s} and \textbf{qb-o} compress a quadtree if necessary. When \textbf{qc} contains much more nodes than other quadtrees, the remove operation dominates (the first groups bars to the left) and has three times the amount of nodes of \textbf{qb-s}. The result also indicates that \textbf{qb-o} contains a similar number of nodes to \textbf{qb-s} despite it only compresses one layer of nodes. Figure~\ref{fig:nodes throughput} illustrates the effectiveness of compression in the face of tremendous contain operations. \textbf{qb-s} outperforms \textbf{qc} by 30\% at 9:1 {\em insert:remove} ratio because \textbf{qb-s} adjusts the quadtree's structure by compression to reduce the length of the search path. With the increment of the insert ratio, \textbf{qb-s} performs similar to \textbf{qc} due to the extra cost of maintaining a stack and the recursive compression. However, \textbf{qb-o} achieves good balance between \textbf{qb-s} and \textbf{qc}, which compresses one layer of nodes without recording the whole traverse path.

\section{Related Works}
Because there are few formal works related to concurrent quadtrees, we present a roadmap to show the development of state-of-the-art concurrent trees in Figure~\ref{fig:trees roadmap}.

\begin{figure}[htbp]
\centering
\includegraphics[scale=0.4]{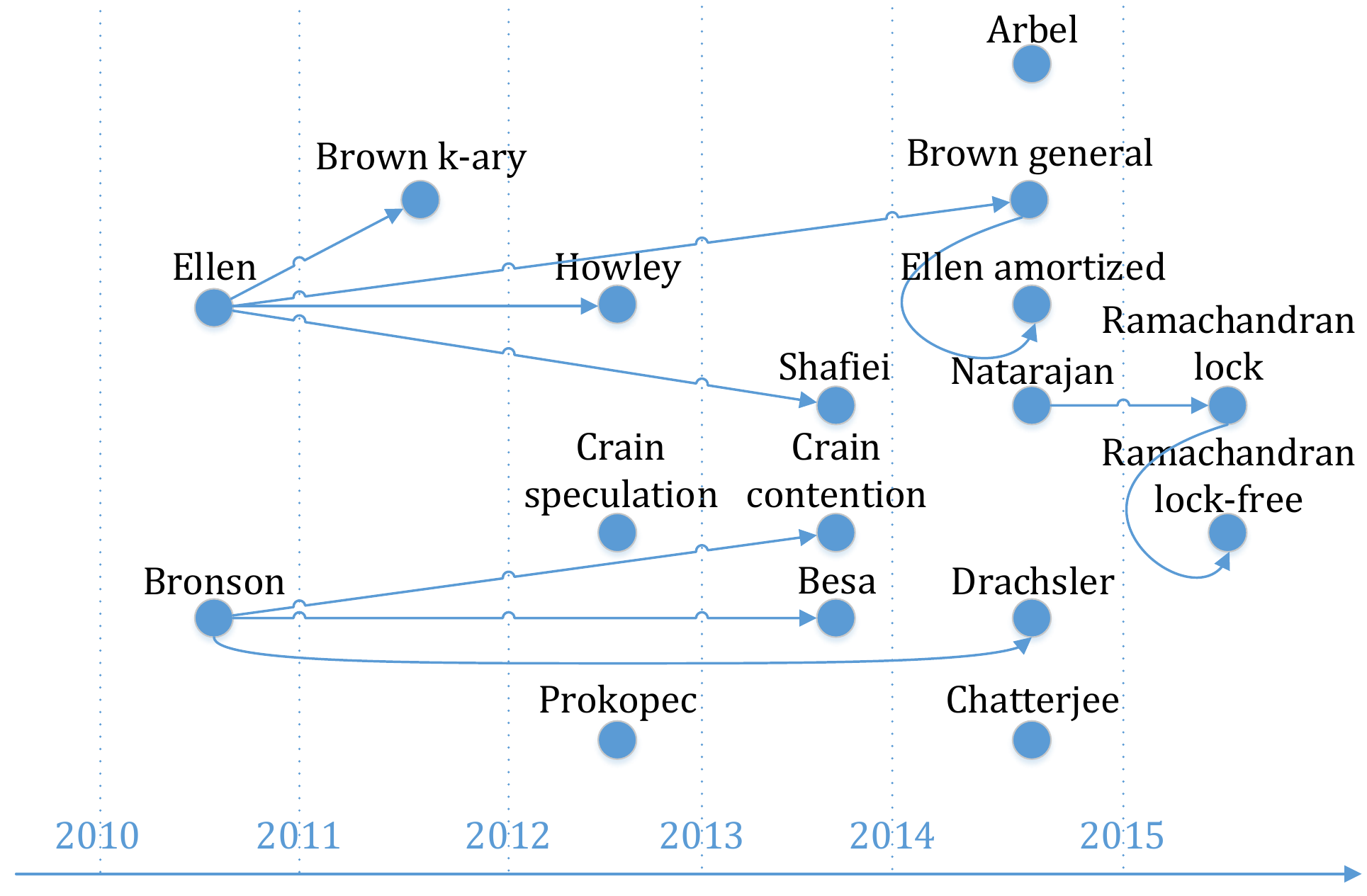}
\caption{Concurrent trees roadmap}
\label{fig:trees roadmap}
\end{figure}

Ellen~\cite{ellen2010non} provided the first non-blocking BST and proved it correct. This work is based on the cooperative method described in Turek~\cite{turek1992locking} and Barnes~\cite{barnes1993method}. Brown~\cite{brown2011non} used a similar approach for the concurrent k-ary tree. Shafiei~\cite{shafiei2013non} also applied the method for the concurrent patricia trie. It also showed how to design a concurrent operation where two pointers need to be changed. The above concurrent trees have an external structure, where only leaf nodes contain actual keys. Howley~\cite{howley2012non} designed the first internal BST by the cooperative method. Recently, Brown~\cite{brown2014general} presented a generalized template for all concurrent down-trees. Ellen~\cite{ellen2014amortized} exhibited how to incorporate a stack to reduce the original complexity from $O(ch)$ to $O(h+c)$. Our quadboost is a hybrid of these techniques. It uses a cooperative method for concurrent coordination, changes two different positions with atomicity, and devises a continuous find mechanism to reduce restart cost.

Different from the previously mentioned methods that apply flags on nodes, Natarajan~\cite{natarajan2014fast} illustrated how to apply flags on edges for a non-blocking external BST. Ramachandran~\cite{ramachandran2015castle} adopted CAS locks on edges to design a concurrent internal BST and later extended the work to non-blocking~\cite{ramachandran2015fast}. Their experiments showed that internal trees are more scalable than external trees with a large key range. On one hand, internal trees take up less memory; on the other hand, the remove operation is more complicated than that in external trees. Chatterjee~\cite{chatterjee2014efficient} provided a threaded-BST with edge flags and claimed it has a low theoretical complexity. Unlike the above trees that have to flag their edges before removal, our CAS quadtree uses a single CAS in both the insert operation and the remove operation.

The first balanced concurrent BST was proposed by Bronson~\cite{bronson2010practical} in which they used an optimistic and relaxed balance method to build an AVL tree. Besa~\cite{besa2013concurrent} employed a similar method for a red-black tree. Crain~\cite{crain2013contention} proposed a method that decouples physical adjustment from logical removal by a background thread. Drachsler~\cite{drachsler2014practical} mentioned an alternative technique called logical ordering, which uses the key order of the BST to optimize the contain operation. All of these works were constructed on fine-grained locks and are deadlock free. Based on  the special properties of quadtree, we also decoupled physical adjustment from logical removal and achieved a high throughput.

There are numerous studies on concurrent trees in addition to ours and the ones mentioned in this work. Crain~\cite{crain2012speculation} designed a concurrent AVL tree based on STM. Prokopec~\cite{prokopec2012concurrent} used a control node, which is similar to our \textit{Operation} object to develop a concurrent trie. Arbel~\cite{arbel2014concurrent} provided a balanced BST with RCU and fine-grained locks.

\section{Conclusions}
In this paper, we present a set of concurrent quadtree algorithms--{\em quadboost}, which support concurrent insert, remove, contain, and move operations. In the remove operation, we decouple physical updates from the logical removal to improve concurrency. The continuous find mechanism analyses flags on nodes to decide whether to move down or up. Further, our LCA-based move operation can modify two pointers with atomicity. The experimental results demonstrate that quadboost outperforms existing one-dimensional tree structures while maintaining a two-dimensional hierarchy. The quadboost algorithms are scalable with a variety of workloads and thread counts.

\bibliographystyle{IEEEtran}
\bibliography{IEEEabrv,ref}

\clearpage
\appendix
\section{appendix}
\label{sec:appendix}

\subsection{Basic Invariants}

We provide a detailed proof of quadboost in this section. We follow the same naming convention in the paper. There are four kinds of \textit{basic operations} in quadboost, i.e. the insert operation, the remove operation, the move operation, and the contain operation. Other functions are called \textit{subroutines}, which are invoked by \textit{basic operations}. The insert operation and the remove operation only operate on one terminal node, whereas the the move operation operates two different terminals--one for inserting a node with $newKey$, the other is for removing a node with $oldKey$. We call them $newKey$'s terminal and $oldKey$'s terminal, and we call their parents $newKey$'s parent and $oldKey$'s parent accordingly. Moreover, we name a CAS that changes a node's $op$ a $flag$ operation and a CAS that changes a node's child a $replace$ operation. We define $snapshot_{T_{i}}$ as the state of our quadtree at some time $T_{i}$.

\subsubsection{Subroutines}

To begin with, we have following observations from quadboost.

\begin{obs}
\label{obs: leaf key}
The key field of a \textit{Leaf} node is never changed. The $op$ field of a \textit{Leaf} node is initially null.
\end{obs}
\begin{obs}
\label{obs: space info}
The space information--$\langle x, y, w, h\rangle$ of an \textit{Internal} node is never changed.
\end{obs}
\begin{obs}
\label{obs: root}
The root node is never changed.
\end{obs}

Based on Observation 3, we derive a Corollary as follows:
\begin{cor}
\label{cor: common path}
Two nodes in quadtree must share a common search path starting from the root.
\end{cor}

Using these observations, we prove that each subroutine satisfies their specific pre-conditions and post-conditions. Because all basic operations invoke the find function at line~\ref{quadboost: find} that returns $\langle l, pOp, path\rangle$ and the findCommon function at line~\ref{quadboost: findCommon} that returns $\langle il, iOp, iPath, rl, rOp, rPath\rangle$, we prove that the two functions satisfy their pre-conditions and post-conditions beforehand. We suppose that they execute from a $snapshot_{T_{i}}$ to derive these conditions.

\begin{defi}
\label{defi: find conditions}
The pre-condition of find(l, pOp, path, keyX, keyY) invoked at $T_{i}$:
\begin{enumerate}
\item If $path$ is not empty, $l$ was on the direction $d \in \{nw, ne, sw, se\}$ of the top node in $path$ at $T_{i1} \leq T_{i}$.
\end{enumerate}

The post-conditions of find(l, pOp, path, keyX, keyY) that returns a tuple $\langle l, pOp, path\rangle$ at $T_{i}$: 
\begin{enumerate}
\item $l$ is a $Leaf$ node or an $Empty$ node. 
\item At some $T_{i1} < T_{i}$, the top node in $path$ has contained $pOp$.
\item At some $T_{i1} < T_{i}$, the top node in $path$ has contained $l$.
\item If $pOp$ was read at $T_{i1}$, and $l$ was read at $T_{i2}$, then $T_{i1} < T_{i2} < T_{i}$.
\item For each node $n$ in $path$, $size(path) \ge 2$, $n_{t}$ is on the top of $n_{t-1}$, and $n_{t}$ was on the direction $d \in \{nw, ne, sw, se\}$ of $n_{t-1}$ at $T_{i1} leq T_{i}$. 
\end{enumerate}
\end{defi}

\begin{defi}
\label{defi: findCommon}
The pre-conditions of findCommon(il, rl, lca, rOp, iOp, iPath, rPath, oldKeyX, oldKeyY, newKeyX, newKeyY) invoked at $T_{i}$:
\begin{enumerate}
\item If $iPath$ is not empty, $il$ was on the direction $d \in \{nw, ne, sw, se\}$ of the top node in $iPath$ at $T_{i1} \leq T_{i}$.
\item If $rPath$ is not empty, $rl$ was on the direction $d \in \{nw, ne, sw, se\}$ of the top node in $rPath$ at $T_{i1} \leq T_{i}$.
\end{enumerate}

The post-conditions of findCommon(il, rl, lca, rOp, iOp, iPath, rPath, oldKeyX, oldKeyY, newKeyX, newKeyY) that returns two tuples $\langle il, iOp, iPath\rangle$ and $\langle rl, rOp, rPath\rangle$ at $T_{i}$:
\begin{enumerate}
\item $rl$ was a \textit{Leaf} or an \textit{Empty} node. 
\item $il$ was a \textit{Leaf} or an \textit{Empty} node. 
\item At some $T_{i1} < T_{i}$, the top node in $rPath$ has contained $rOp$.
\item At some $T_{i1} < T_{i}$, the top node in $iPath$ has contained $iOp$.
\item At some $T_{i1} < T_{i}$, the top node in $rPath$ has contained $rl$.
\item At some $T_{i1} < T_{i}$, the top node in $iPath$ has contained $il$.
\item If $iOp$ was read at $T_{i1}$, and $il$ was read at $T_{i2}$, then $T_{i1} < T_{i2} < T_{i}$.
\item If $rOp$ was read at $T_{i1}$, and $rl$ was read at $T_{i2}$, then $T_{i1} < T_{i2} < T_{i}$.
\item For each node $n$ in $rPath$, $size(rPath) \ge 2$, $n_{t}$ is on the top of $n_{t-1}$, and $n_{t}$ was on the direction $d \in \{nw, ne, sw, se\}$ of $n_{t-1}$ at $T_{i1} \leq T_{i}$.  
\item For each node $n$ in $iPath$, $size(iPath) \ge 2$, $n_{t}$ is on the top of $n_{t-1}$, and $n_{t}$ was on the direction $d \in \{nw, ne, sw, se\}$ of $n_{t-1}$ at $T_{i1} \leq T_{i}$. 
\end{enumerate}
\end{defi}
By observing the find function and the findCommon function, we have following lemmas. 
\begin{lem}
\label{lem: find first}
At the first time calling the find function, the contain operation, the insert operation, and the remove function start with $l$ as an \textit{Internal} node and an empty path.
\end{lem}
\begin{proof}
The contain operation at line~\ref{quadboost: contain root}, the insert operation at line~\ref{quadboost: insert root} and the remove operation  at line~\ref{quadboost: remove root} start with the root node by $l = root$, an \textit{Internal} node that is never changed (Observation~\ref{obs: root}).

Besides, the contain operation at line~\ref{quadboost: contain path}, the insert operation at line~\ref{quadboost: insert path}, and the remove operation at line~\ref{quadboost: remove path} start with an empty stack to record nodes.
\end{proof}

\begin{lem}
\label{lem: findCommon first}
At the first time the move operation calls the findCommon function, it starts with \textit{rl} as an \textit{Internal} node and $rPath$ and $iPath$ are empty.
\end{lem}
\begin{proof}
The move operation begins with the root node at line by $rl = root$ that is never changed (Observation~\ref{obs: root}) at line~\ref{quadboost: move root}. Both $rPath$ and $iPath$ are initialized as empty stacks at line~\ref{quadboost: move two paths}.
\end{proof}

\begin{lem}
\label{lem: find continueFind push nodes}
All nodes pushed by the find function (line~\ref{quadboost: find}) and the continueFind function (line~\ref{quadboost: continueFind}) are \textit{Internal} nodes.
\end{lem}
\begin{proof}
Before pushing into the stack, the find function at line~\ref{quadboost: find check} first check the class of a node. Also, the continueFind function calls the find function at line~\ref{quadboost: continueFind find} to push nodes into $path$. Thus, nodes other than \textit{Internal} cannot be pushed.
\end{proof}

\begin{lem}
\label{lem: continueFind continueFindCommon push nodes}
All nodes pushed by the findCommon function (line~\ref{quadboost: findCommon}) the continueFindCommon function (line~\ref{quadboost: continueFindCommon}) are \textit{Internal}.
\end{lem}
\begin{proof}
Before pushing into the stack, the findCommon function first checks the class of a node at line~\ref{quadboost: findCommon check}. Also, the continueFindCommon function calls the findCommon function at line~\ref{quadboost: continueFindCommon findCommon}. Or it calls the find function at line~\ref{quadboost: continueFindCommon insert find} and line~\ref{quadboost: continueFindCommon remove find} to push \textit{Internal} nodes according to Lemma~\ref{lem: find continueFind push nodes}. Thus, nodes other than \textit{Internal} cannot be pushed.
\end{proof}

\begin{lem}
\label{lem: find findCommon process}
For the loop in the find function at line~\ref{quadboost: find} and the findCommon function at line~\ref{quadboost: findCommon}, we suppose that $path$($rPath$) is empty and refer $l$($rl$) and $pOp$($rOp$) to each field updated by the loop. If the loop executes at least once and breaks at $T_{i}$, $l$($rl$), $pOp$($rOp$), and $path$($rPath$) satisfy following conditions:
\begin{enumerate}
\item $l$($rl$) is a \textit{Leaf} node or an \textit{Empty} node. 
\item At some $T_{i1} < T_{i}$, the top node in $path$($rPath$) has contained $pOp$($rOp$).
\item At some $T_{i1} < T_{i}$, the top node in $path$($rPath$) has contained $l$($rl$).
\item If $pOp$($rOp$) was read at $T_{i1}$, and $l$($rl$) was read at $T_{i2}$, then $T_{i1} < T_{i2} < T_{i}$.
\item For each node $n$ in $path$($rPath$), $size(path) \ge 2$, $n_{t}$ is on the top of $n_{t-1}$, and $n_{t}$ was on the direction $d \in \{nw, ne, sw, se\}$ of $n_{t-1}$ at $T_{i1} \leq T_{i}$. 
\end{enumerate}
\end{lem}
\begin{proof}
Apart from the terminate condition that the findCommon function judges whether two directions are the same or not, it performs in the same pattern as the find function. Both functions first push the previous \textit{Internal} node $l'$ into $path$, then get its $op$, and read a child pointer $l$ finally. 

The first condition always holds because by Lemma~\ref{lem: continueFind continueFindCommon push nodes} and Lemma~\ref{lem: find continueFind push nodes} \textit{Leaf} nodes and \textit{Empty} nodes will never be pushed into the $path$($rPath$). 

We then prove part 2-4. At the last iteration, $pOp$($rOp$) and $l$($rl$) are read from $l'$, which has already been pushed into the $path$($rPath$). Hence, part 2 and part 3 are correct. Besides, $pOp$($rOp$) is read before $l$($rl$) so that part 4 holds.

For the last part, we assume that at $T_{i1}$ the lemma holds, so at $T_{i2} > T_{i1}$ $l$ is read from $l'$ which has already been pushed into the $path$($rPath$). As $path$($rPath$) is initially empty, and $T_{i}$ is at the end of the last iteration, $l$ which on the top of $l'$ has been a child of it.
\end{proof}

Lemma~\ref{lem: find findCommon process} proves loop within the find function and the findCommon function satisfy post-conditions in Definition~\ref{defi: find conditions} and Definition~\ref{defi: findCommon}. Next parts are going to show that other statements do not change these properties. We consider the first invocation as the base case, and prove lemmas by induction.

\begin{lem}
\label{lem: find pre-condition post-conditions}
Every call to the find function satisfies its pre-condition and post-conditions.
\end{lem}
\begin{proof}
Now we prove the base case. We prove that initially the lemma holds. 

The pre-condition:

By Lemma~\ref{lem: find first}, the contain operation, the insert operation, and the remove operation start with an empty stack to record nodes.

For the move operation, by Lemma~\ref{lem: findCommon first} it starts with $rl$ as an \textit{Internal} node and empty paths. Therefore it executes the loop more than once to push \textit{Internal} nodes into $path$. In the last iteration $rl$ is pushed at line~\ref{quadboost: findCommon rl push}, and the parent of it has already been push at the prior iteration. Then, $rl$ is read from at line~\ref{quadboost: findCommon lca top}. Hence, $rl$ at line~\ref{quadboost: findCommon find rl} is a child of the top node in $rPath$ if it is not empty. $iPath$ is empty since it starts with an empty stack.

Therefore, initially the pre-conditions are satisfied.

The post-conditions:

The first condition always satisfies as the class of a node is checked at line~\ref{quadboost: find check}. 

By Lemma~\ref{lem: find first}, the contain operation, the insert operation, and the remove operation start with an empty stack. Thu Lemma~\ref{lem: find findCommon process} indicates that post-conditions are satisfied. 

For the find function within the findCommon function, by Lemma~\ref{lem: find findCommon process} it satisfies post-conditions when it breaks out.
Hence, at line~\ref{quadboost: findCommon find il}, post-conditions are satisfied because $iPath$ is empty. At line~\ref{quadboost: findCommon find rl}, if $rl$ starts as a \textit{Leaf} node, the post-conditions are also satisfied. Otherwise, there are two parts of $rPath$ at the time the loop breaks out, where we have proved each of them satisfy post-conditions by Lemma~\ref{lem: find findCommon process}. Moreover, by the pre-condition, $rl$ has been a child of the top node in $rPath$ of the find function. Thus, all nodes pushed later also satisfy the post-conditions.

We have proved the base case. Then we assume that for invocations $find_{0}, find_{1}, find_{2}...find_{k}$, the first $k$ invocations satisfy the pre-conditions and post-conditions. We should prove that $find_{k+1}$ also satisfy the conditions. We have to consider all places where the find function is invoked. 

The pre-condition:

First we consider the contain operation, the insert operation, and the move operation.

At line~\ref{quadboost: contain find}, line~\ref{quadboost: insert find}, line~\ref{quadboost: remove find}, the invocations follow the base case which we have proved satisfy the pre-condition.
 
At line~\ref{quadboost: continueFind find}, there are two scenarios $l$ is read. $l$ could either be read from the $path$ at line~\ref{quadboost: continueFind path pop} or be assigned to its parent node at line~\ref{quadboost: continueFind p}. In the latter case, $p$ was at the top of $path$ at $T_{i1} < T_{i}$ (line~\ref{quadboost: insert pop} and line~\ref{quadboost: remove pop}) by the hypothesis. Otherwise $l$ is assigned to a node in $path$. As there's no other push operations, $n_{t}$ on the top of $n_{t-1}$ must be its child at some $T_{i1} < T_{i}$ by the hypothesis.

Then, we discuss the move operation.

At line~\ref{quadboost: findCommon find rl}, the find function is wrapped by the findCommon function. Since Lemma~\ref{lem: find findCommon process} shows that the loop set $rl$ as a child of the top node in $rPath$, we have to prove that either it is an \textit{Internal} node or an \textit{Leaf} node that was a child of the top node before entering the loop. At line~\ref{quadboost: move findCommon}, the findCommon function follows the base case. At line~\ref{quadboost: continueFindCommon findCommon}, Corollary~\ref{cor: common path} shows that the root node will always be a LCA node. Hence, after setting $rPath$ to its lca index at line~\ref{quadboost: continueFindCommon rPath set lca index}, it contains at least one node. By Lemma~\ref{lem: find continueFind push nodes} and Lemma~\ref{lem: continueFind continueFindCommon push nodes}, $rl$ is an \textit{Internal} node popped from $rPath$ at line~\ref{quadboost: continueFindCommon rl pop common}. 

At line~\ref{quadboost: findCommon find il}, the find function is also wrapped by the findCommon function. We could prove that $iPath$ is empty. At line~\ref{quadboost: move findCommon}, the findCommon function follows the base case. At line~\ref{quadboost: continueFindCommon findCommon}, $iPath$ is set to empty before the invocation (line~\ref{quadboost: continueFindCommon iPath clear}).

At line~\ref{quadboost: continueFindCommon remove find}, $rl$ is either assigned to $rp$ (line~\ref{quadboost: continueFindCommon remove rl rp}) or a node popped from $rPath$ (line~\ref{quadboost: continueFindCommon remove rl pop}). If it is popped from $rPath$, our hypothesis ensures that it was a child of the top node from $rPath$. If it is read from $rp$, which is an \textit{Internal} node popped from $rPath$ at line~\ref{quadboost: move rp pop}, line~\ref{quadboost: continueFindCommon rp pop common}, or line~\ref{quadboost: continueFindCommon rp pop}, it also satisfies the claim.

At line~\ref{quadboost: continueFindCommon insert find}, likewise, $il$ is either assigned to $ip$ (line~\ref{quadboost: continueFindCommon insert il ip}) or a node popped from $iPath$ (line~\ref{quadboost: continueFindCommon insert il pop}). If it is popped from $iPath$, our hypothesis ensures that $il$ was a child node of the top node from $iPath$ if it is not empty. If it is read from $ip$, which is an \textit{Internal} node popped from $iPath$ at line~\ref{quadboost: move ip pop}, line~\ref{quadboost: continueFindCommon ip pop common}, or line~\ref{quadboost: continueFindCommon ip pop}, it also satisfies the claim.

The post-conditions:

The first condition always holds. Because by Lemma~\ref{lem: find continueFind push nodes}, if $l$ is an \textit{Internal} node, it will be pushed into $path$.

We then prove other conditions.

First we consider the contain operation, the insert operation, and the move operation.

At line~\ref{quadboost: contain find}, line~\ref{quadboost: insert find}, line~\ref{quadboost: remove find}, the invocations follow the base case which we have proved satisfy the post-conditions.
 
At line~\ref{quadboost: continueFind find}, there are two scenarios $l$ is read. $l$ could either be read from the $path$ at line~\ref{quadboost: continueFind path pop} or be assigned to its parent node at line~\ref{quadboost: continueFind p}. In both cases, $l$ was an \textit{Internal} node in $path$. Therefore Lemma~\ref{lem: find findCommon process} indicates that all nodes followed by $l$ satisfy the post-conditions.

Then, we discuss the move operation.

At line~\ref{quadboost: findCommon find rl}, the find function is wrapped by the findCommon function. The pre-condition indicates $rl$ was a child of $rPath$ if it is not empty. Hence, it suffices to show that the post-conditions are satisfied by Lemma~\ref{lem: find findCommon process}.

At line~\ref{quadboost: findCommon find il}, the find function is also wrapped by the findCommon function. The pre-condition indicates $iPath$ is an empty. In this way, by Lemma~\ref{lem: find findCommon process} we prove the post-conditions. 

At line~\ref{quadboost: continueFindCommon remove find}, likewise, the pre-condition suggests that $rl$ was an \textit{Internal} node in $rPath$. By Lemma~\ref{lem: find findCommon process} we also prove the post-conditions. 

At line~\ref{quadboost: continueFindCommon insert find}, the pre-condition suggests that $il$ was an \textit{Internal} node in $iPath$. Therefore, by Lemma~\ref{lem: find findCommon process} we prove the post-conditions. 
\end{proof}

After showing the pre-condition and post-conditions of the find function, we the prove that the findCommon function that returns $\langle il, iOp, iPath, rl, rOp, rPath\rangle$ satisfies its pre-conditions and post-conditions.
\begin{lem}
\label{lem: findCommon pre-condition post-conditions}
Every call to the findCommon function satisfies its pre-conditions and post-conditions.
\end{lem}
\begin{proof}
The pre-condition:

We prove the lemma by induction. First we prove the base case.

At line~\ref{quadboost: move findCommon}, the findCommon function is initially invoked, and both $iPath$ and $rPath$ are empty (line~\ref{quadboost: move two paths}). 

Then we assume that for invocations $findCommon_{0},$ $findCommon_{1},$ $findCommon_{2}$ ... $findCommon_{k}$, the first $k$ invocations satisfy the pre-conditions. We prove that $findCommon_{k+1}$ also satisfies the conditions. 

There are two places the findCommon function is invoked. At line~\ref{quadboost: move findCommon}, the base case proves the post-conditions. At line~\ref{quadboost: continueFindCommon findCommon}, Corollary~\ref{cor: common path} shows that the root node will always be a LCA node. Hence, after setting $rPath$ to its lca index at line~\ref{quadboost: continueFindCommon rPath set lca index}, it contains at least one node. By Lemma~\ref{lem: find continueFind push nodes} and Lemma~\ref{lem: continueFind continueFindCommon push nodes}, $rl$ is an \textit{Internal} node popped from $rPath$ at line~\ref{quadboost: continueFindCommon rl pop common}. Hence, we prove the pre-condition.

The post-conditions:

The findCommon function wraps two find functions at line~\ref{quadboost: findCommon find il} and line~\ref{quadboost: findCommon find rl}. Because $\langle rl, rOp, rPath\rangle$ that passed from the find function for $rl$ and $\langle il, iOp, iPath\rangle$ that passed from the find function for $il$ hold the post-conditions, the lemma is true.
\end{proof}

After proving the pre-conditions and post-conditions of the find function and the findCommon function, we shall prove that the continuous find mechanism (i.e. the continueFind function and the continueFindCommon function) satisfies its pre-conditions and post-conditions. 

\begin{defi}
The pre-conditions of continueFind(pOp, path, l, p):
\begin{enumerate}
	\item $l$ was child of $p$.
	\item $p$ was child of the top node in $path$ if it is not empty.
\end{enumerate} 
The post-conditions are the same as the find function.
\end{defi}

\begin{defi}
The pre-conditions of continueFindCommon(il, rl, lca, rOp, iOp, iPath, rPath, ip, rp, oldKeyX, oldKeyY, newKeyX, newKeyY, iFail, rFail, cFail):
\begin{enumerate}
	\item At least one of $iFail$, $rFail$ or $cFail$ is true.
	\item $il$ was a child of $ip$.
	\item $ip$ was a child of the top node in $iPath$ if it is not empty.
	\item $rl$ was a child of $rp$.
	\item $rp$ was a child of the top node in $rPath$ if it is not empty.
\end{enumerate} 
The post-conditions are the same as the findCommon function.
\end{defi}

We now prove that the continueFind function satisfies its pre-conditions and post-conditions. 

\begin{lem}
\label{lem: continueFind pre-condition post-conditions}
Every call to the continueFind function satisfies its pre-condition and post-conditions.
\end{lem}
\begin{proof}
The post-conditions:

At line~\ref{quadboost: continueFind find} of the continueFind function, it executes the find function. Because the find operation satisfies its post-conditions by Lemma~\ref{lem: find pre-condition post-conditions}, the continueFind function also satisfies the same conditions.

The pre-condition:

The continueFind function is invoked at line~\ref{quadboost: insert continueFind} and line~\ref{quadboost: remove continueFind}. 

Initially, the continueFind function follows the find function at line~\ref{quadboost: insert find} and line~\ref{quadboost: remove find}. For the first part, the post-condition of the find function shows that $l$ was a child of $p$ which popped at line~\ref{quadboost: insert pop} or line~\ref{quadboost: remove pop}(Lemma~\ref{lem: find findCommon process}). For the second part, after popping $p$, it becomes a child of the top node in $path$ if it is not empty. 

Otherwise, it reads results from the continueFind invocation at the prior iteration. As the post-conditions of the continueFind function is the same as the find function, we prove the lemma.
\end{proof}

We now prove that the continueFindCommon function satisfies its pre-conditions and post-conditions.

\begin{lem}
\label{lem: continueFindCommon pre-conditions post-conditions}
Every call to the continueFindCommon function satisfies its pre-conditions and post-conditions
\end{lem}
\begin{proof}

The post-conditions:

There are three cases, the first part of the pre-condition shows that it must enter one of the loop. We consider each case by the program execution order. 

\begin{enumerate}
	\item If $rFail$ is true and $cFail$ is false, it executes the case starts from line~\ref{quadboost: rFail true}. If $cFail$ is not set to true, it executes the find function at line~\ref{quadboost: continueFindCommon remove find}. Because the find function satisfies its post-conditions (Lemma~\ref{lem: find pre-condition post-conditions}), the lemma holds. If $iFail$ is true, it comes to the second part of the proof. If $cFail$ is set to true, it comes to the third part of the proof. 
	\item If $iFail$ is true and $cFail$ is false, it executes the case starts from line~\ref{quadboost: iFail true}. If $cFail$ is not set to true, it executes the find function at line~\ref{quadboost: continueFindCommon insert find}. Because the find function satisfies its post-conditions, part 1 and Lemma~\ref{lem: find pre-condition post-conditions} indicate the Lemma holds. If $cFail$ is true, it comes to the third part of the proof.
	\item If $cFail$ is true, it executes the case starts from line~\ref{quadboost: cFail true}. Because the findCommon function satisfies its post-conditions (Lemma~\ref{lem: findCommon pre-condition post-conditions}), it establishes this part of the lemma.
\end{enumerate}

The pre-condition:

The continueFindCommon function is invoked at line~\ref{quadboost: move continueFindCommon}.

For the first condition, line~\ref{quadboost: move rOp check}, line~\ref{quadboost: move iOp check}, line~\ref{quadboost: move ip rp same}, line~\ref{quadboost: rFail iFail true}, line~\ref{quadboost: iFail true}, line~\ref{quadboost: rFail true}, line~\ref{quadboost: cFail true} indicate that at least on of three flags is set to true.

For the next conditions, initially the continueFindCommon follows the findCommon function at line~\ref{quadboost: move findCommon}. As $ip$ and $rp$ are popped at line~\ref{quadboost: move ip pop} and line~\ref{quadboost: move rp pop}, $il$ was a child of $ip$ and $rl$ was a child of $rp$. Further, $ip$ and $rp$ was a child of the top node in $iPath$ and $rPath$. All these claims are based on Lemma~\ref{lem: continueFind pre-condition post-conditions}. Otherwise, the continueFindCommon function reads results from  its invocation at the prior iteration. In such a case, our post-conditions prove the pre-conditions.

\end{proof}

We observe that the find function, the continueFind function, the findCommon function,  and the continueFindCommon function return a similar pattern of results. The find function and the continueFind function return a tuple $\langle l, pOp, path\rangle$. The findCommon function and the continueFindCommon function return two such tuples. Hence We use $find(keys)$ to denote such a set of find operations for searching keys. Further, we number $find(keys)$ by their invocation orders. Invocations at line~\ref{quadboost: insert find}, line~\ref{quadboost: remove find}, and line~\ref{quadboost: move findCommon} are at $iteration_{0}$. Invocations at line~\ref{quadboost: insert find}, line~\ref{quadboost: remove find}, and line~\ref{quadboost: move findCommon} are at $iteration_{i}, 0 < i$. 

\begin{defi}
For $compress(path, p)$ invoked at $T_{i}$, it has following pre-conditions:
\begin{enumerate}
\item $p$ is an \textit{Internal} node.
\item $path$ is read from the result of $find(keys)$ at the prior iteration.
\item If $path$ is not empty, $p$ was a child node of the top node in $path$ at $T_{i1} < T_{i}$.
\end{enumerate}
\end{defi}

\begin{lem}
\label{lem: compress pre-conditions}
Every call to the compress function satisfies its pre-conditions.
\end{lem}
\begin{proof}
The compress function is called at line~\ref{quadboost: remove compress} or line~\ref{quadboost: move compress}. At line~\ref{quadboost: remove compress}, $p$ is a node popped from $path$ (line~\ref{quadboost: remove pop}. At line~\ref{quadboost: move compress}, $rp$ is an \textit{Internal} node popped from $rPath$ (line~\ref{quadboost: move rp pop}, line~\ref{quadboost: continueFindCommon rp pop} or line~\ref{quadboost: continueFindCommon rp pop common}. This proves the first part.

For last two parts, since the compress function reads result following $find(keys)$, it suffices to prove the post-conditions of $find(keys)$ satisfy the preconditions. By Lemma~\ref{lem: continueFind pre-condition post-conditions}, Lemma~\ref{lem: continueFindCommon pre-conditions post-conditions}, Lemma~\ref{lem: find pre-condition post-conditions}, and Lemma~\ref{lem: findCommon pre-condition post-conditions}, the post-conditions of the find function, the findCommon function, the continueFind function, and the continueFindCommon function establish the lemma.
\end{proof}

\begin{cor}
\label{cor: check pre-conditions}
Every call to the check function satisfies its pre-conditions that $p$ is an \textit{Internal} node.
\end{cor}

\begin{defi}
For $moved(node)$ invoked at $T_{i}$, it has following the pre-conditions that $node$ is a \textit{Leaf} node or an \textit{Empty} node.
\end{defi}

\begin{lem}
\label{lem: moved pre-condition}
Every call to the moved function satisfies its pre-condition.
\end{lem}
\begin{proof}
The moved function is called at line~\ref{quadboost: insert intree}, line~\ref{quadboost: remove intree}, line~\ref{quadboost: findCommon intree il}, line~\ref{quadboost: findCommon intree rl}, line~\ref{quadboost: continueFindCommon intree il}, line~\ref{quadboost: continueFindCommon intree rl}. 

At line~\ref{quadboost: insert intree}, it reads $l$ from line~\ref{quadboost: insert find} or line~\ref{quadboost: insert continueFind}. Hence, according to Lemma~\ref{lem: find pre-condition post-conditions}, $l$ is a \textit{Leaf} node or an \textit{Empty} node.

At line~\ref{quadboost: remove intree}, it reads $l$ from line~\ref{quadboost: remove find} or line~\ref{quadboost: remove continueFind}. Hence, according to Lemma~\ref{lem: find pre-condition post-conditions}, $l$ is a \textit{Leaf} node or an \textit{Empty} node.

At line~\ref{quadboost: findCommon intree il} and line~\ref{quadboost: findCommon intree rl}, it reads $il$ from line~\ref{quadboost: findCommon find il} and $rl$ from line~\ref{quadboost: findCommon find rl}. Therefore, according to Lemma~\ref{lem: find pre-condition post-conditions}, $il$ and $rl$ are \textit{Leaf} node or \textit{Empty} nodes.

At line~\ref{quadboost: continueFindCommon intree il}, it reads $il$ from line~\ref{quadboost: continueFindCommon insert find}. Therefore, according to Lemma~\ref{lem: find pre-condition post-conditions}, $il$ is a \textit{Leaf} node or an \textit{Empty} node.

At line~\ref{quadboost: continueFindCommon intree rl}, it reads $rl$ from line~\ref{quadboost: continueFindCommon remove find}. Therefore, according to Lemma~\ref{lem: find pre-condition post-conditions}, $rl$ is a \textit{Leaf} node or an \textit{Empty} node.
\end{proof}

Based on these post-conditions, we show that each $op$ created at $T_{i}$ store their corresponding information.

\begin{lem}
\label{lem: op create}
For $op$ created at $T_{i}$:
\begin{enumerate}
\item If $op$ is a \textit{Substitute} object, $op.parent$ has contained $op.oldChild$, a \textit{Leaf} node, at $T_{i1} < T_{i}$ from the result of $find(keys)$ at the previous iteration.
\item If $op$ is a \textit{Compress} object, $op.grandparent$ has contained $op.parent$, an \textit{Internal} node, at $T_{i1} < T_{i}$ from the result of $find(keys)$ at the previous iteration. 
\item If $op$ is a \textit{Move} object, $op.iParent$ has contained $op.oldIChild$, a \textit{Leaf} node, $op.rParent$ has contained $op.oldRChild$, a \textit{Leaf} node, $op.iParent$ has contained $op.oldIOp$, and $op.rParent$ has contained $op.oldROp$ before $T_{i}$ from the result of $find(keys)$ at the previous iteration.
\end{enumerate}
\end{lem}
\begin{proof}
\begin{enumerate}
\item A \textit{Substitute} object is created at line~\ref{quadboost: insert Substitute} or line~\ref{quadboost: remove Substitute} by assigning $\langle p, l, newNode\rangle$. Since $p$ read from the top node in $path$ at line~\ref{quadboost: insert pop} or line~\ref{quadboost: remove pop}, and $l$ is returned by the find function or findCommon function, the post-conditions of them establish the claim according to Lemma~\ref{lem: find pre-condition post-conditions}, Lemma~\ref{lem: findCommon pre-condition post-conditions}, Lemma~\ref{lem: continueFindCommon pre-conditions post-conditions}, and Lemma~\ref{lem: continueFind pre-condition post-conditions}.

\item A \textit{Compress} object is created at line~\ref{quadboost: compress Compress} by assigning $\langle path, p\rangle$. $p$ is passed from the top of $path$ at line~\ref{quadboost: remove pop} or line~\ref{quadboost: move rp pop}, the post-conditions of $find(keys)$ and the property of $path$ establishes the claim (Lemma~\ref{lem: find pre-condition post-conditions}, Lemma~\ref{lem: findCommon pre-condition post-conditions}, Lemma~\ref{lem: continueFindCommon pre-conditions post-conditions}, Lemma~\ref{lem: continueFind pre-condition post-conditions}, Lemma~\ref{lem: find continueFind push nodes}, and Lemma~\ref{lem: continueFind continueFindCommon push nodes}).

\item A \textit{Move} object is created at line~\ref{quadboost: move create end} by assigning $\langle ip, rp, il, rl, newNode, iOldOp, rOldOp\rangle$. $ip$ is assigned to the top node of $iPath$ at line~\ref{quadboost: move ip pop}, $rp$ is assigned to the top node of $rPath$ at line~\ref{quadboost: move rp pop}, $il$ was a child of the top node of $iPath$, $rl$ was a child of the top node of $rPath$, and $iOp$ and $rOp$ are read from $ip$ and $rp$ by the post-conditions of the findCommon function and the continueFindCommon function (Lemma~\ref{lem: continueFindCommon pre-conditions post-conditions} and Lemma~\ref{lem: findCommon pre-condition post-conditions}.
\end{enumerate}
\end{proof}

Then we prove the pre-conditions of other functions which use the result of above functions to modify quadtree.

The following lemmas satisfy a basic pre-condition that their arguments are the same type as indicated in algorithms.
\begin{lem} 
\label{lem: same type}
\begin{enumerate}
\item Every call to the help function satisfies its pre-conditions.
\item Every call to the helpSubstitute function satisfies its pre-conditions.
\item Every call to the helpCompress function satisfies its pre-conditions.
\item Every call to the helpMove function satisfies its pre-conditions.
\end{enumerate}
\end{lem}
\begin{proof}
\begin{enumerate}
\item 
The help function is called at line~\ref{quadboost: insert help}, line~\ref{quadboost: remove help}, line~\ref{quadboost: continueFindCommon insert help}, line~\ref{quadboost: continueFindCommon remove help}, line~\ref{quadboost: continueFindCommon insert help common}, and line~\ref{quadboost: continueFindCommon remove help common}.

At line~\ref{quadboost: insert help} and line~\ref{quadboost: remove help},  $pOp$ might be obtained from the result of the find function and the continueFind function. Hence, the post-conditions of them ensures that $pOp$ is an \textit{Operation} object read from $p$. Otherwise, $pOp$ might be obtained from line~\ref{quadboost: insert p op}, or line~\ref{quadboost: remove p op} by reading a node's $op$. 

At line~\ref{quadboost: continueFindCommon insert help} and line~\ref{quadboost: continueFindCommon remove help}, $rOp$ and $iOp$ are passed the move operation. The post-conditions of the continueFindCommon function and the findCommon function (Lemma~\ref{lem: findCommon pre-condition post-conditions} and Lemma~\ref{lem: continueFindCommon pre-conditions post-conditions}), with line~\ref{quadboost: iFail true}, line~\ref{quadboost: rFail true}, and line~\ref{quadboost: cFail true} guarantee that arguments' type are the same. 

At line~\ref{quadboost: continueFindCommon insert help common} and line~\ref{quadboost: continueFindCommon remove help common}, $rOp$ and $iOp$ could be passed from the move function or read from $ip$ and $rp$ in the continueFindCommon function at lines that have been mentioned.

\item The helpSubstitute function is called at line~\ref{quadboost: insert helpSubstitute}, line~\ref{quadboost: remove helpSubstitute}, and line~\ref{quadboost: help helpSubstitute}. At line~\ref{quadboost: insert helpSubstitute} and line~\ref{quadboost: remove helpSubstitute}, $op$ is created at line~\ref{quadboost: insert Substitute} or line~\ref{quadboost: remove Substitute} respectively. At line~\ref{quadboost: help helpSubstitute}, it checks whether $op$ is \textit{Substitute} before calling the helpSubstitute function.

\item The helpCompress function is called at line~\ref{quadboost: compress Compress}, line~\ref{quadboost: continueFindCommon insert helpCompress}, line~\ref{quadboost: continueFindCommon remove helpCompress}, and line~\ref{quadboost: continueFindCommon helpCompress common}. For each invocation, it checks whether $op$ is \textit{Compress} beforehand.

\item The helpMove function is invoked at line~\ref{quadboost: move helpMove first}, line~\ref{quadboost: move helpMove second}, and line~\ref{quadboost: help helpMove}. At line~\ref{quadboost: move helpMove first} and line~\ref{quadboost: move helpMove second}, $op$ is created at line~\ref{quadboost: move create end}. At line~\ref{quadboost: help helpMove}, it checks whether $op$ is \textit{Move} before the invocation.

\end{enumerate}
\end{proof}

\begin{cor}
\label{cor: help pOp}
For $help(op)$, $op$ is read from an \textit{Internal} node.
\end{cor}

\begin{defi}
For $helpFlag(p, oldOp, newOp)$ invoked at $T_{i}$, it has the pre-conditions:
\begin{enumerate}
\item $p$ is an \textit{Internal} node.
\item For $cflag$, $p$ has contained $pOp$ at $T_{i1} < T_{i}$ after reading $p$; otherwise, $p$ has contained $pOp$ at $T_{i1} < T_{i}$ from the result of $find(keys)$ at the prior iteration.
\end{enumerate}
\end{defi}

\begin{lem}
\label{lem: helpFlag}
Every call to the helpFlag function satisfies its pre-conditions. 
\end{lem}
\begin{proof}
The helpFlag function is invoked at line~\ref{quadboost: insert helpFlag}, line~\ref{quadboost: remove helpFlag}, line~\ref{quadboost: compress check}, line~\ref{quadboost: move helpFlag}, line~\ref{quadboost: helpMove flag begin}-line~\ref{quadboost: helpMove flag end}, and line~\ref{quadboost: helpMove reverse begin}-line~\ref{quadboost: helpMove reverse end}.

At line~\ref{quadboost: insert helpFlag} and line~\ref{quadboost: remove helpFlag}, the post-conditions of the find function and the continueFind function (Lemma~\ref{lem: find pre-condition post-conditions} and Lemma~\ref{lem: continueFind pre-condition post-conditions} imply that $p$ is an \textit{Internal} node popped at line~\ref{quadboost: insert pop} and has contained $pOp$.

At line~\ref{quadboost: compress check}, the pre-condition of the compress function (Lemma~\ref{lem: compress pre-conditions} shows that $p$ and nodes from $path$ are read from $find(keys)$ at the prior iteration, and $pOp$ is read at line~\ref{quadboost: compress p op}. Thus, $p$ has $pOp$ at $T_{i1} < T_{i}$ after reading $p$.

At line~\ref{quadboost: move helpFlag}, $ip$ and $rp$ are popped from $iPath$ at line~\ref{quadboost: move ip pop} and line~\ref{quadboost: move rp pop} or read from the continueFindCommon function. By Lemma~\ref{lem: findCommon pre-condition post-conditions} and Lemma~\ref{lem: continueFindCommon pre-conditions post-conditions}, either $op.iParent$ has contained $op.iOp$ or $op.rParent$ has contained $op.rOp$.

At line~\ref{quadboost: helpMove flag begin}-line~\ref{quadboost: helpMove flag end} and line~\ref{quadboost: helpMove reverse begin}-line~\ref{quadboost: helpMove reverse end}, by the pre-condition of the helpMove function~\ref{lem: same type}, $op$ is a \textit{Move} object. Therefore, according to Lemma~\ref{lem: op create}, either $op.iParent$ has contained $op.iOp$ or $op.rParent$ has contained $op.rOp$.
\end{proof}

\begin{lem}
Every call to the hasChild function satisfies its pre-conditions that $parent$ is an \textit{Internal} node. 
\end{lem}
\begin{proof}
According to Lemma~\ref{lem: moved pre-condition}, $node$ is a \textit{Leaf} node or an \textit{Empty} node. According to Lemma~\ref{lem: helpFlag}, flag operations only perform on \textit{Internal} nodes. Therefore line is the only place that set a \textit{Move} $op$ on a \textit{Leaf} node. By Observation~\ref{obs: leaf key}, $op$ is initially null. After checking the condition, $op$ is assigned to a \textit{Leaf} node where $op.iParent$ is an \textit{Internal} node by Lemma~\ref{lem: op create}
\end{proof}

\begin{defi}
For $helpReplace(p, oldChild, newChild)$ invoked at $T_{i}$, it has the pre-conditions:
\begin{enumerate}
\item $p$, $oldChild$, and $newChild$ are read from the same $op$.
\item $newChild$ is a node that has not been in quadtree.
\item $p$ is an $Internal$ node that has contained $oldChild$ at $T_{i1} < T_{i}$.
\end{enumerate}
\end{defi}
\begin{lem}
\label{lem: helpReplace}
Every call to the helpReplace function satisfies its pre-conditions.
\end{lem}
\begin{proof}
For the first condition:

The helpReplace function is invoked at line~\ref{quadboost: helpSubstitute helpReplace}, line~\ref{quadboost: helpCompress helpReplace}, line~\ref{quadboost: helpMove two cas begin}-\ref{quadboost: helpMove two cas end}.

By Lemma~\ref{lem: same type}, at line~\ref{quadboost: helpCompress helpReplace}, it reads a \textit{Compress} object; at line~\ref{quadboost: helpSubstitute helpReplace}, it reads a \textit{Substitute} object; at line~\ref{quadboost: helpMove two cas begin}-\ref{quadboost: helpMove two cas end}. it reads a \textit{Move} object.

For the second condition:

Since part 1 illustrates that $p$, $oldChild$, and $newChild$ are read from the same $op$, if $op$ is a \textit{Substitute} object, $op.newNode$ is created at line~\ref{quadboost: insert createNode} or line~\ref{quadboost: remove root}; if $op$ is a \textit{Compress} object, each time the helpCompress function creates a new Empty node at line~\ref{quadboost: helpCompress helpReplace}; if $op$ is a \textit{Move} object, $op.newIChild$ is created at line~\ref{quadboost: move create begin}, and an empty node is created at line~\ref{quadboost: helpSubstitute helpReplace}.

For the third condition:

Since part 1 illustrates that $p$, $oldChild$, and $newChild$ are read from the same $op$, Lemma~\ref{lem: op create} shows that $p$ has contained $oldChild$.
\end{proof}

\begin{defi}
For $createNode(l, p, newKeyX, newKeyY, value)$ that returns a $newNode$ invoked at $T_{i}$, it has the pre-conditions:
\begin{enumerate}
\item $p$ is an \textit{Internal} node.
\item $p$ has contained $l$ at $T_{i1} < T_{i}$ from the result of $find(keys)$ at the prior iteration.
\end{enumerate}

post-conditions:

The $newNode$ returned is either a \textit{Leaf} node with $newKey$ and $value$, or a sub-tree that contains both $l.key$ node and $newKey$ node with the same parent.

\end{defi}

\begin{lem}
\label{lem: createNode pre-conditions and post-condition}
Every call to the createNode function satisfies its pre-conditions and post-condition.
\end{lem}
\begin{proof}
The createNode function is invoked at line~\ref{quadboost: insert createNode} and line~\ref{quadboost: move createNode}. 

At line~\ref{quadboost: insert createNode}, the post-conditions of the find function and the continueFind function ensure that $l$ has been a child of $p$ that popped from $path$ at line~\ref{quadboost: insert pop} (Lemma~\ref{lem: find pre-condition post-conditions} and Lemma~\ref{lem: continueFind pre-condition post-conditions}).

At line~\ref{quadboost: move createNode}, the post-conditions of the findCommon function and the continueFindCommon function ensure that $l$ has been a child of $p$ popped at line~\ref{quadboost: move ip pop}. (Lemma~\ref{lem: findCommon pre-condition post-conditions}, Lemma~\ref{lem: continueFindCommon pre-conditions post-conditions}).

Though we have not presented the details of the createNode function, our implementation guarantee that the post-condition must be satisfied.
\end{proof}

\subsubsection{Flag and replace operations}
We argue that each successful CAS operation occurs in a correct order. At outset, we argue that CAS behaviours for each helpFlag function. A flag operation involves three arguments: $node$, $oldOp$, $newOp$. A replace operation involves three arguments: $p$, $oldChild$, $newChild$. From Figure~\ref{fig:State transition diagram}, we denote each flag operation accordingly. $iflag$ occurs in the insert operation, $rflag$ occurs in the remove operation, and $mflag$ occurs in the move operation. Moreover, we specify a flag operation which attaches a \textit{Clean} $op$ on a node as an unflag operation. We call $ireplace$ as the replace operation occurs in the insert operation, $rreplace$ as the replace operation occurs in the remove operation, $mreplace$ as the replace operation occurs in the move operation, and $creplace$ as the replace operation occurs in the compress operation. The next lemmas describe the behaviours of the helpFlag function according to the state transition diagram. 

\begin{obs}
$op$ is only attached on an \textit{Internal} node.
\end{obs}

\begin{lem}
\label{lem: op properties}
For a new node $n$:
\begin{enumerate}
\item It is created with a \textit{Clean} $op$.
\item $rflag$, $iflag$, $mflag$ or $cflag$ succeeds only if $n$'s $op$ is \textit{Clean}.
\item $unflag$ succeeds only if $n$'s $op$ is \textit{Substitute} or \textit{Move}.
\item Once $n$'s $op$ is \textit{Compress}, its $op$ will never be changed.
\end{enumerate}
\end{lem}
\begin{proof}
\begin{enumerate}
\item As shown at line~\ref{quadboost: structure Internal op Clean}, \textit{Internal} nodes are assigned a \textit{Clean} $op$ when they are created.

\item Before $iflag$, $op$ is checked at line~\ref{quadboost: insert pop Clean check}; before $rflag$, $op$ is checked at line~\ref{quadboost: remove pop Clean check}; before $cflag$, $op$ is checked at line~\ref{quadboost: compress clean}. For the move operation, before $mflag$ on $oldKey$'s parent and $newKey$'s parent, they are checked at line~\ref{quadboost: move iOp check} and line~\ref{quadboost: move rOp check}.

\item $unflag$ is called at line~\ref{quadboost: helpSubstitute helpFlag} and line~\ref{quadboost: helpMove reverse begin}-line~\ref{quadboost: helpMove reverse end}, where the pre-conditions specify that $op$ is \textit{Move} or \textit{Substitute} according to Lemma~\ref{lem: same type}.

\item By part 3, there's no unflag operation happens on an  \textit{Internal} node with a \textit{Compress} $op$.
\end{enumerate}
\end{proof}

Next we illustrate that there's no ABA problem on any $op$. That is to say, in terms of an \textit{Internal} node $n$, its $op_{i}$ at $T_{i}$ has not appeared at $T_{i1} < T_{i}$, $op_{i1} \neq op_{i}$. We prove the following lemma:

\begin{lem}
\label{lem: op ABA}
For an \textit{Internal} node $n$, it never reuse an $op$ that has been set previously. 
\end{lem}
\begin{proof}
If a node's $op$ is set to $op_{i}$ at $T_{i}$, it has never been appeared before. We prove the lemma by discussing different types of $flag$ operations.

For $iflag$, a \textit{Substitute} $op$ is created at line~\ref{quadboost: insert Substitute} before its invocation. For $rflag$, a \textit{Substitute} $op$ is created at line~\ref{quadboost: remove Substitute} before its invocation. For a $cflag$, a Compress $op$ is created at line~\ref{quadboost: compress Compress} before its invocation. For $mflag$, a \textit{Move} $op$ is created at line~\ref{quadboost: move create end} before its invocation. Because each $newOp$ is newly created, it has not been set before. For $unflag$, each time it creates a new \textit{Clean} $op$ to replace the prior \textit{Move} $op$ or \textit{Substitute} $op$ by Lemma~\ref{lem: op properties}. Therefore, the new \textit{Clean} $op$ has never appeared before. 
\end{proof}
We have shown that for each successful flag operation, it never set an $op$ that has been used before. Every \textit{Internal} node is initialized with \textit{Clean} $op$ by Lemma~\ref{lem: op properties} and we use a flag operation to change $op$ at the beginning of each CAS transition. We say successive replace operations and unflag operations that read the $op$ belongs to it. According to Figure~\ref{fig:State transition diagram}, there are threes categories of successful CAS transitions. We denote them as $flag\rightarrow unflag$, $flag\rightarrow replace$, and $flag\rightarrow  replace\rightarrow unflag$. 

Let $flag_{0}$, $flag_{1}$, ..., $flag_{n}$ be a sequence of successful flag operations; let $unflag_{0}$, $unflag_{1}$, ..., $unflag_{n}$ be a sequence of successful unflag operations. $flag_{i}$ attaches $op_{i}$, and $replace_{i}$ and $unflag_{i}$ read $op_{i}$ and come after it. Therefore we say $flag_{i}$, $replace_{i}$, and $unflag_{i}$ belong to the same $op$. In addition, if there are more than one replace operations belong to the same $op_{i}$, we denote them as $replace_{i}^{0}$, $replace_{i}^{1}$, ..., $replace_{i}^{k}$ ordered by their successful sequence. Similarly, if there are more than one flag operations that belong to $op_{i}$ on different nodes, we denote them as $flag_{i}^{0}$, $flag_{i}^{1}$, ..., $flag_{i}^{k}$. A similar notation is used for $unflag_{i}$. The following lemmas prove the correct ordering of three different transitions.

\begin{lem}
\label{lem: replace Substitute op before}
\begin{enumerate}
\item $ireplace$ will not be done before the success of $iflag$ which belongs to the same $op$.
\item $rreplace$ will not be done before the success of $rflag$ which belongs to the same $op$.
\end{enumerate}
\end{lem}
\begin{proof}
For the insert operation and the remove operation, at line~\ref{quadboost: helpSubstitute helpReplace} the helpReplace function is wrapped by the helpSubstitute function. The helpSubstitute function is called at line~\ref{quadboost: insert helpSubstitute}, line~\ref{quadboost: remove helpSubstitute}, and line~\ref{quadboost: help helpSubstitute}. At line~\ref{quadboost: insert helpSubstitute}, it is called after the success of $iflag$ at line~\ref{quadboost: insert helpFlag}. At line~\ref{quadboost: remove helpSubstitute}, it is called after the success of $rflag$ at line~\ref{quadboost: remove helpFlag}. At line~\ref{quadboost: help helpSubstitute}, by Corollary~\ref{cor: help pOp} and Lemma~\ref{lem: op properties}, there must have been $iflag$ or $rflag$ that change the node's $op$ from \textit{Clean} to \textit{Substitute}.
\end{proof}

\begin{lem}
\label{lem: replace Compress op before}
$creplace$ will not be done before the success of $cflag$ which belongs to the same $op$.
\end{lem}
\begin{proof}
The helpCompress function which wraps $creplace$ is called at line~\ref{quadboost: compress Compress}, line~\ref{quadboost: continueFindCommon insert helpCompress}, line~\ref{quadboost: continueFindCommon remove helpCompress}, and line~\ref{quadboost: continueFindCommon helpCompress common}, and line~\ref{quadboost: help helpCompress}. At line~\ref{quadboost: compress Compress}, the helpCompress function follows the successful helpFlag function at line~\ref{quadboost: compress check}. At line~\ref{quadboost: continueFindCommon insert helpCompress}, line~\ref{quadboost: continueFindCommon remove helpCompress}, and line~\ref{quadboost: continueFindCommon helpCompress common}, a $pOp$ is read from \textit{Internal} nodes. At line~\ref{quadboost: help helpCompress}, Lemma~\ref{cor: help pOp} shows that $op$ is read from an \textit{Internal} node.  Hence, by Lemma~\ref{lem: op properties}, successful $cflag$ must have changed the node's $op$ from \textit{Clean} to \textit{Compress}.
\end{proof}

\begin{lem}
\label{lem: replace Move op before}
$mreplace$ will not be done before the success of $mflag$ which belongs to the same $op$.
\end{lem}
\begin{proof}
The helpMove function which wraps two $mreplace$s is called at line~\ref{quadboost: move helpMove first}, line~\ref{quadboost: move helpMove second}, and line~\ref{quadboost: help helpMove}. At line~\ref{quadboost: move helpMove first} and line~\ref{quadboost: move helpMove second}, the \textit{Move} $op$ is created at line~\ref{quadboost: helpMove reverse end}. Otherwise, at line~\ref{quadboost: help helpMove}, by Corollary~\ref{cor: help pOp}, $op$ is read from an \textit{Internal} node. Hence, Lemma~\ref{lem: op properties} demonstrates that a successful $mflag$ must have changed the node's $op$ from \textit{Clean} to \textit{Move}.
\end{proof}

From the above lemmas, we have the following corollary:
\begin{cor}
\label{cor: flag before replace}
A successful replace operation which belongs to an $op$ will not succeed before it's flag operations have been done.
\end{cor}

We discuss each CAS transition accordingly. We begin by discussing the $flag \rightarrow replace \rightarrow unflag$ transition. We refer $\langle oldIOp, oldIChild, newIChild\rangle$ and $\langle oldROp, oldRChild, newRChild\rangle$ to different tuples of $\langle oldOp, oldChild, newChild\rangle$.

\begin{lem}
\label{lem: flag replace unflag transition}
$flag\rightarrow replace \rightarrow unflag$ occurs when $rflag_{i}$, $iflag_{i}$ or $mflag_{i}$ succeeds, and it has following properties:
\begin{enumerate}
	\item $replace_{i}$ never occurs before $flag_{i}$.
	\item $flag_{i}^{k}, 0 \leq k < \vert flag_{i} \vert$ is the first successful flag operation on $op_{i}.parent^{k}$ after $T_{i1}$ when $pOp_{i}$ is read.
	\item $replace_{i}^{k}, 0 \leq k < \vert replace_{i} \vert$ is the first successful replace operation on $op_{i}.parent^{k}$ after $T_{i2}$ when $op_{i}.oldChild^{k}$ is read.
	\item $replace_{i}^{k}, 0 \leq k < \vert replace_{i} \vert$ is the first successful replace operation on $op_{i}.parent^{k}$ that belongs to $op_{i}$.
	\item $unflag_{i}^{k}, 0 \leq k < \vert unflag_{i} \vert$ is the first successful unflag operation on $op_{i}.parent^{k}$ after $flag_{i}^{k}$.
	\item There is no successful unflag operation occurs before $replace_{i}$.
	\item The first replace operation on $op_{i}.parent^{k}$ that belongs to $op_{i}$ must succeed. 
\end{enumerate} 
\end{lem}
\begin{proof}
	\begin{enumerate}
		\item Corollary~\ref{cor: flag before replace} proves the claims.			
		\item By Lemma~\ref{lem: helpFlag}, for each flag operation, $p$ has contained $pOp$ at some time $T_{i1}$. Suppose that there's another $flag^{k}$ that occurs after $T_{i1}$ but before $flag_{i}^{k}$, then $flag_{i}^{k}$ fails because Lemma~\ref{lem: op ABA} shows that $op$ is not reused. Therefore, it contradicts to the definition that $flag_{i}^{k}$ is a successful flag operation.
		\item By Lemma~\ref{lem: helpReplace}, for each replace operation, $p$ has contained $oldChild$ at some time $T_{i1}$. Suppose that there's another $replace$ that occurs after $T_{i1}$ but before $replace_{i}^{k}$, then $replace_{i}^{k}$ fails because Lemma~\ref{lem: helpReplace} shows that $newChild$ is never a node in the tree. Therefore, it contradicts to the definition that $replace_{i}^{k}$ is a successful flag operation.
		\item Suppose there's $replace^{k}$ that belongs to $op_{i}$ and occurs before $creplace_{i}^{k}$. It must follow $flag^{k}$ according to Corollary~\ref{cor: flag before replace}. If $flag^{k}$ happens before $flag_{i}^{k}$, it fails because $fail_{i}^{k}$ will fail by consequence (Lemma~\ref{lem: op ABA}). If $flag^{k}$ happens after $flag_{i}^{k}$, it contradicts Lemma~\ref{lem: op properties} that a flag operation starts Because $flag^{k}$ is done after reading $op_{i}.oldChild^{k}$, $replace^{k}$ is also after reading it. Therefore, if $flag^{k}$ succeeds, it contradicts part 3 that $replace_{i}^{k}$ is the first successful replace operation after reading $op_{i}.oldChild^{k}$.
		\item Based on Lemma~\ref{lem: helpFlag}, if there is another $unflag$ changes $p.op$, $unflag_{i}^{k}$ that use $op_{i}$ as the $oldOp$ will fail. This contradicts to the definition that $unflag_{i}^{k}$ is a successful unflag operation that reads $op_{i}$.
		\item We consider each operation separately.
		
		For the insert operation and the remove operation, $unflag$ follows $replace$ immediately at line~\ref{quadboost: helpSubstitute helpFlag}. Assuming that $replace_{i}^{k}$ occurs after $unflag_{i}^{k}$, there must be another $replace$ reads $op_{i}$ and changes the link. It contradicts part 3 of the proof that $replace_{i}^{k}$ is the first successful replace operation that belongs to $op_{i}$. 
		
		For the move operation, because we assume that replace operations  exist, $doCAS$ is set to true before performing replace operations. An unflag operation follows replace operations from line~\ref{quadboost: helpMove reverse begin}-\ref{quadboost: helpMove reverse end}. Consider if a $replace_{i}^{k}, 0 \leq k < \vert replace_{i} \vert$ happens after any unflag operation, then another $replace$ which use the same $op$ must have succeed as the program order specifies that $replace$ execute before $unflag$ for a process. Hence, $replace_{i}^{k}$ fails and contradicts the definition.
		
		Therefore all successful unflag operations occur after $replace$.
		
		\item According to Corollary~\ref{cor: flag before replace}, $replace_{i}^{k}$ follows $flag_{i}^{k}$ belongs to $op_{i}$. Suppose there exists $replace$ after reading $op_{i}.parent^{k}$ and before $replace_{i}^{k}$. Then, there must be $flag$ that precedes $replace$ by Corollary~\ref{cor: flag before replace}. By part 1, $flag_{i}$ is the first successful flag operation after reading $op_{i}$. By Lemma~\ref{lem: helpFlag}, $op_{i}$ is read before $op_{i}.l$. Hence, there does not exist any replace operation between $T_{i1}$ reading $op_{i}.oldChild^{k}$ and $flag_{i}^{k}$. If $flag$ happens between $flag_{i}^{k}$ and $replace_{i}^{k}$, it will also fail because $op_{i}.parent^{k}.op$ is not \textit{Clean}. Then, the first replace operation that belongs to $op_{i}$ must succeed. 
	\end{enumerate}
\end{proof}

We then discuss the ordering of the $flag\rightarrow replace$ transition.  

\begin{lem}
\label{lem: flag replace transition}
$flag\rightarrow replace$ occurs only when $cflag_{i}$ succeeds, it has following properties:
\begin{enumerate}
	\item $creplace_{i}$ never occurs before $cflag_{i}$.
	\item $cflag_{i}$ is the first successful flag operation on $op_{i}.parent$ after $T_{i1}$ when $rOp_{i}$ is read.
	\item $creplace_{i}$ is the first successful replace operation on $op_{i}.grandparent$ after $T_{i1}$ when $op_{i}.parent$ is read.
	\item $creplace_{i}$ is the first successful replace operation on $op_{i}.grandparent$ that belongs to $op_{i}$.
	\item There is no unflag operation after $creplace_{i}$.
	\item The first replace operation that belongs to $op_{i}$ must succeed.
\end{enumerate}
\end{lem}
\begin{proof}
	\begin{enumerate}
		\item Corollary~\ref{cor: flag before replace} proves the claims.			
		\item By Lemma~\ref{lem: helpFlag}, for each flag operation, $p$ has contained $pOp$ at some time $T_{i1}$. Suppose that there's another $cflag$ that occurs after $T_{i1}$ but before $cflag_{i}$, then $cflag_{i}$ fails because Lemma~\ref{lem: op ABA} shows that $op$ is not reused. Therefore, it contradicts to the definition that $cflag_{i}$ is a successful flag operation.
		\item By Lemma~\ref{lem: helpReplace}, for each replace operation, $p$ has contained $oldChild$ at some time $T_{i1}$. Suppose that there's another $creplace$ that occurs after $T_{i1}$ but before $replace_{i}$, then $creplace_{i}$ fails because $oldChild$ remains the same but the link has been changed to another node. Therefore, it contradicts to the definition that $creplace_{i}$ is a successful flag operation.
		\item Suppose there's a replace operation that belongs to $op_{i}$ and occurs before $creplace_{i}$. It must follow $cflag$ according to Corollary~\ref{cor: flag before replace}. If $cflag$ happens before $cflag_{i}$, it fails because $cfail_{i}$ will fail by consequence (Lemma~\ref{lem: op ABA}. If $cflag$ happens after $cflag_{i}$, it contradicts Lemma~\ref{lem: op properties} that a flag operation starts with a \textit{Clean} $op$. Therefore, $cflag$ is the same as $cflag_{i}$. Because $cflag$ is done after reading $op_{i}.parent$, $creplace$ is also after reading it. Therefore, if $cflag$ succeeds, it contradicts part 3 that $creplace_{i}$ is the first successful replace operation after reading $op_{i}.parent$.
		\item By Lemma~\ref{lem: op properties}, once a node's $op$ is \textit{Compress}, it will never be changed.
		\item Suppose the first replace operation that belongs to $op_{i}$ fails, there must exist some successful $replace$ after reading $op_{i}.parent$ and before $creplace_{i}$. Further, there is $flag$ that precedes $replace$ by Corollary~\ref{cor: flag before replace}. If $flag$ happens between $creplace_{i}$ and $cflag_{i}$, $flag$ will fail by Lemma~\ref{lem: op properties}. If $flag$ happens between reading $op_{i}$ and $cflag_{i}$, it contradicts part 1 that $cflag_{i}$ is the first successful flag operation after reading $op_{i}$. Because $op_{i}$ is read after $op_{i}.parent$, we also have to consider if $flag$ happens before reading $op_{i}$. Consider if $cflag$ happens, $flag_{i}$ will fail because $op_{i}.parent.op$ will not be set back to \textit{Clean} according to the fourth part. Consider if $mflag$, $iflag$ or $rflag$ happens, it will not change the link from $op_{i}.grandparent$ to $op_{i}.p$ according to Lemma~\ref{lem: flag replace unflag transition}.
	\end{enumerate}
\end{proof}

\begin{cor}
\label{cor: first replace succeed}
The first replace operation belongs to $op_{i}$ must succeed.
\end{cor}

Next we discuss the $flag\rightarrow unflag$ transition, where replace operation occurs between a successful flag operation and an unflag operation that belong to the same $op$. We suppose that $op.iFirst$ is true. The claim for $op.iFirst = false$ is symmetric. 

\begin{obs}
\label{obs: allFlag iFirst}
$allFlag$ and $iFirst$ in a \textit{Move} $op$ are initially false, and they will never be set false after assigning to true. 
\end{obs}

\begin{lem}
\label{lem: flag unflag transition}
For $flag\rightarrow unflag$ circle, it only results from a \textit{Move} object $op_{i}$ such that:
\begin{enumerate}
	\item $unflag_{i}$ is the first successful unflag operation on $op_{i}.rParent$.	
	\item The first flag operation on $op_{i}.iParent$ must fail, and no later flag operation will succeed.
	\item $op_{i}.iParent$ and $op_{i}.rParent$ are different.
\end{enumerate}
\end{lem}
\begin{proof}
	\begin{enumerate}
		\item We prove it by contradiction. Suppose there's another unflag operation that reads $op_{i}$ and succeeds before $unflag_{i}$ on $op_{i}.rParent$. It must have changed $op_{i}.rParent.op$ to a new $op$, by Lemma~\ref{lem: op ABA} $unflag_{i}$ will fail. This contradicts the definition of $unflag_{i}$.
		\item 
		For the former part of the claim, we prove it by contradiction. Assuming that $op_{i}.iParent$ is flagged on $op_{i}$, and will be unflagged without performing replace operations. However, unflag operations occur after the judgement at line~\ref{quadboost: helpMove two cas begin}. The first process must set $doCAS$ to true. In the next step, it executes line~\ref{quadboost: helpMove two cas begin}-\ref{quadboost: helpMove two cas end}. According to Corollary~\ref{cor: first replace succeed}, the first $replace_{i}^{t}$ must succeed. Therefore it contradicts the definition that there's no replace operation before an unflag operation.
		For the latter part, by Lemma~\ref{lem: op create} $op_{i}.iParent$ has contained $op_{i}.oldIOp$. Because the first flag operation on $op_{i}.iParent$ fails, latter flag operations read the $op_{i}.oldIOp$ also fail.
		\item Assuming $op_{i}.iParent$ and $op_{i}.rParent$ are the same. By part 2, flag operations on $op_{i}.iParent$ must fail, hence $doCAS$ will never be true. $op_{i}.allFlag$ is always false by by Observation~\ref{obs: allFlag iFirst}. In this way, line~\ref{quadboost: helpMove reverse begin}-\ref{quadboost: helpMove reverse end} prevents any unflag operation. This derives a contradiction.
	\end{enumerate}
\end{proof}

\subsubsection{Quadtree properties}
In the Section, we use above lemmas to show that quadtree's properties are maintained during concurrent modifications. 

\begin{defi}
\label{defi: quadtree structure}
Our quadtree has these properties:
\begin{enumerate}
\item Two layers of dummy \textit{Internal} nodes are never changed.
\item An \textit{Internal} node $n$ has four children, which locate in the direction $d \in \{nw, ne, sw, se\}$ respectively according to their $\langle x, y, w, h\rangle$, or $\langle keyX, keyY\rangle$. 

For \textit{Internal} nodes reside on four directions:
\begin{itemize}
\item $n.nw.x = n.x$, $n.nw.y = n.y$; 
\item $n.ne.x = n.x + w/2$, $n.ne.y = n.y$;
\item $n.sw.x = n.x$, $n.sw.y = n.y + n.h/2$;
\item $n.se.x = n.x + w/2$, $n.se.y = n.y+h/2$,
\end{itemize}
and all children have their $w' = n.w / 2$, $h'=n.h/2$. 

For \textit{Leaf} nodes reside on four directions:
\begin{itemize}
\item 
$n.x \leq n.nw.keyX < n.x + n.w/2$, $n.y \leq n.nw.keyY < n.y + n.h / 2$; 
\item
$n.x + n.w / 2 \leq n.ne.keyX < n.x + n.w$, $n.y \leq n.ne.keyY < n.y + n.h / 2$;
\item
$n.x \leq n.sw.keyX < n.x + n.w/2$, $n.y + n.h/2 \leq n.sw.keyY < n.y + n.h$;
\item 
$n.x + n.w / 2 \leq n.se.keyX < n.x + n.w$, $n.y + n.h / 2 \leq n.se.keyY < n.y + n.h$.
\end{itemize}
\end{enumerate}
\end{defi}

\begin{lem}
\label{lem: dummy not change}
Two layers of dummy nodes are never changed.
\end{lem}
\begin{proof}
We prove the lemma by induction. Only replace operations will affect the structure of quadtree as we define in Definition~\ref{defi: quadtree structure}.

Initially the property holds as we initialize quadtree with two layers of dummy nodes and a layer of \textit{Empty} nodes.

Suppose that after $replace_{i}$ the lemma holds, we shall prove that after $replace_{i+1}$ the lemma is still true. For $ireplace$, $rrepalce$, and $mreplace$, the pre-condition~\ref{lem: helpReplace} ensures that three nodes are from the same $op$. Thus, as Lemma~\ref{lem: op create} further ensures that they only swing \textit{Leaf} nodes and \textit{Empty} node, dummy node are not changed. For $creplace$, since at $replace_{i}$ the property is true, the root is connected to a layer of dummy nodes. Lemma~\ref{lem: op create} shows that $op.oldChild$ is an \textit{Internal} node, and line~\ref{quadboost: compress check} shows that the root node's child pointer will never be changed. Therefore, $replace_{i+1}$ will not occur on the root node. 
\end{proof}

\begin{lem}
\label{lem: all children not changed}
Children of a node with \textit{Compress} $op$ will not be changed.
\end{lem}
\begin{proof}
Based on Lemma~\ref{lem: op properties}, the node's $op$ will never be changed after being set to \textit{Compress}. Lemma~\ref{lem: flag replace transition} and Lemma~\ref{lem: flag replace unflag transition} indicate that for a node flagged with \textit{Compress} $op$, only itself will be unlinked. This establishes the Lemma.
\end{proof}

\begin{lem}
\label{lem: all children empty}
Only an \textit{Internal} node with all children \textit{Empty} could be attached with a \textit{Compress} object.
\end{lem}
\begin{proof}
Corollary~\ref{cor: check pre-conditions} ensures that the pre-conditions of the check function is satisfied. It checks whether all children are \textit{Empty} before calling the helpFlag function at line~\ref{quadboost: compress check}. This establishes the Lemma.
\end{proof}

\begin{lem}
\label{lem: op Compress reachable}
An \textit{Internal} node whose $op$ is not \textit{Compress} is reachable from the root.
\end{lem}
\begin{proof}
We have to consider all kinds of replace operations.

For $ireplace$, $rrepalce$, and $mreplace$, by Lemma~\ref{lem: helpReplace} and Lemma~\ref{lem: op create} they use a new node to replace a \textit{Leaf} node or an \textit{Empty} node. Thus, \textit{Internal} nodes are still reachable from the root.

For $creplace$, by Lemma~\ref{lem: helpReplace} and Lemma~\ref{lem: op create}, it replaces \textit{Internal} nodes. Lemma~\ref{lem: flag replace transition} shows that it is preceded by $cflag$. Lemma~\ref{lem: all children empty} implies that all children are \textit{Empty} nodes. Therefore, an \textit{Internal} node with a  \textit{Compress} $op$ and its children are not reachable from the root only after the success of $creplace$.
\end{proof}

Next we define \textit{active} set and \textit{inactive} set for different kinds of nodes. We say a node is \textit{moved} if the moved function turns true. For an \textit{Internal} node or an \textit{Empty} node, if it is reachable from the root in $snapshot_{T_{i}}$, it is active; otherwise, it is \textit{inactive}. For a \textit{Leaf} node, if it is reachable from the root in $snapshot_{T_{i}}$ and not moved, it is active; otherwise, it is \textit{inactive}. We denote $path(keys^{k}), 0 \leq k < \vert replace_{i} \vert$ as a stack of nodes pushed by $find(keys)$ in a snapshot. We define $physical\_path(keys^{k})$ to be the path for $keys^{k}$ in $snapshot_{T_{i}}$, consisting of a sequence of \textit{Internal} nodes with a \textit{Leaf} node or an \textit{Empty} node at the end, which is the actual path in the snapshot. We say a subpath of $path(keys^{k})$ is an $active\_path$ if all nodes from the root to node $n \in path(keys^{k})$ are active due to their pushing order. Hence a physical path with a \textit{Leaf} node is active only if $path(keys^{k})$ is an active path and the \textit{Leaf} node is active.

\begin{lem}
\label{lem: path Compress end}
There's at most a node with \textit{Compress} $op$ in the subpath of $path(keys^{k})$ that is active, and it resides in the end of the path if exists.
\end{lem}
\begin{proof}
Assume that there are two nodes $n$ and $n'$ reside in the $active\_path(keys^{k})$ with \textit{Compress} $op$. We prove the lemma by contradiction. Because $n$ and $n'$ are active with \textit{Compress} $op$, all children of $n$ and $n'$ are \textit{Empty} (Lemma~\ref{lem: all children empty}. Since $n$ and $n'$ are \textit{Internal} nodes in the same $path$, it derives a contradiction. 

For the second part of the proof, if $n$ resides in other places in  $path$, its children should be \textit{Empty} by Lemma~\ref{lem: all children empty}. Therefore, only if $n$ is on the end of path, conditions are satisfied.
\end{proof}

\begin{lem}
\label{lem: nodes above active}
For $path(keys^{k})$, if $n_{t}$ is active in $snapshot_{T_{i}}$, then $n_{0}, ... , n_{t-1}$ above $n_{t}$ are active.
\end{lem}
\begin{proof}
By Lemma~\ref{lem: op Compress reachable}, nodes with $op$ other than \textit{Compress} are reachable from the root. Lemma~\ref{lem: path Compress end} indicates that a node with a \textit{Compress} $op$ is on the end of the path. Therefore, other nodes do not have a \textit{Compress} $op$, which establishes the lemma.
\end{proof}

\begin{lem}
\label{lem: lca active}
If in the $snapshot_{T_{i}}$, $n$ is the LCA on $active\_path$ for $oldKey$ and $newKey$. Then at $T_{i1}, T_{i1} > T_{i}$, $n$ is still the LCA on $active\_path$ for both $oldKey$ and $newKey$ if it is active.
\end{lem}
\begin{proof}
Nodes from root to the LCA node shares a common path (Observation~\ref{cor: common path}. Because the LCA node is active, nodes above the LCA node are active. Hence, the subpath from the root to the LCA node is active by Lemma~\ref{lem: nodes above active}.
\end{proof}

\begin{lem}
\label{lem: compress push more once}
Nodes a with a \textit{Compress} $op$ will not be pushed into $path$ more than once.
\end{lem}
\begin{proof}
Supposing that a node $n$ is active in $snapshot_{T_{i}}$, and it becomes inactive at $T_{i1} > T_{i}$. Besides, suppose that $op$ of $n$ is \textit{Compress} and pushed at $T_{i2}, T_{i} < T_{i2} < T_{i1}$. We prove that after $T_{i2}$, $n$ will never be pushed again.

Lemma~\ref{lem: flag replace transition} shows that a node with a \textit{Compress} $op$ will never be unflagged. Hence, based on Lemma~\ref{lem: find pre-condition post-conditions} and Lemma~\ref{lem: find pre-condition post-conditions}), we have to prove that nodes in $path$ with a \textit{Compress} $op$ will not be pushed into again. For the insert operation and the remove operation, $pOp$ is checked at line~\ref{quadboost: continueFind p} and line~\ref{quadboost: continueFind pOp} before the next find operation. For the move operation, $rOp$ is checked at line~\ref{quadboost: continueFindCommon insert il ip}, line~\ref{quadboost: continueFindCommon insert il op}, line~\ref{quadboost: continueFindCommon helpCompress common}, and $iOp$ is checked at line~\ref{quadboost: continueFindCommon remove rl op} and line~\ref{quadboost: continueFindCommon remove rl rp} before calling the findCommon function and the find function. Hence, at $T_{i2} < T_{i1}$, $op$ must be checked before $find(keys)$ so that nodes with a \textit{Compress} $op$ will not be pushed more than once.
\end{proof}

Lemma~\ref{lem: path Compress end} shows that for $active\_path$, active nodes with a \textit{Compress} $op$ will only reside at the end. If other nodes pushed before $n$ will become inactive, we can pop them beforehand until reaching the last node an $op$ other than \textit{Compress}. 

\begin{lem}
\label{lem: find snapshot}
After the invocation of $find(keys)$ which reads $l^{t}$, there is a snapshot, where the path from the root to $l^{t}$ is $physical\_path(keys^{t})$.
\end{lem}
\begin{proof}
We prove the lemma by induction.

In the base case, where $path$ starts from the root node, the claim is true. 

We consider the pushing sequence as $n_{0}, n_{1}, ..., n_{k}$.  Suppose that for first $k$ nodes, the path from the root to $n_{k}$ is the $physical\_path$. We shall prove that for $n_{k+1}$, the lemma is true. If there is no $replace$ operation before reading $n_{k+1}$, it is obvious that we can linearized it as the same snapshot before pushing $n_{k}$.

Next, we assume that $replace$ occurs on $n_{k}$ before reading $n_{k+1}$, and results in $snapshot$. There are two cases:
\begin{enumerate}
\item Consider if $replace$ occurs after reading $n_{k+1}$. We could have $snapshot_{k+1} = snapshot_{k}$ because $n_{k+1}$ is connected to an active node in a snapshot so that $n_{k+1}$ is also reachable from the root.

\item Consider if $replace$ occurs before reading $n_{k+1}$. If the replace operation changes $n_{k}$, then $n_{k}$ is flagged with a \textit{Compress} $op$ by Lemma~\ref{lem: flag replace transition}. Then we have $snapshot_{k+1} = snapshot_{k}$ because Lemma~\ref{lem: all children not changed} demonstrates that a node flagged with a \textit{Compress} has no successive replace operations on its child pointers. Otherwise, if the replace operation changes $n_{k+1}$, $n_{k}$ is not changed by Lemma~\ref{lem: flag replace unflag transition}. Hence, we have $snapshot_{k+1} = snapshot$, where $n_{k+1}$ is reachable from the root in the $snapshot$ just after a replace operation.
\end{enumerate}
\end{proof}

\begin{lem}
\label{lem: replace properties remain}
After $ireplace$, $rreplace$, $mreplace$, and $creplace$, quadtree's properties remain.
\end{lem}
\begin{proof}
We shall prove that in any $snapshot_{T_{i}}$, quadtree's properties remain.

First, Lemma 10 shows that two layers of dummy nodes remain in the tree. We have to consider other layers of nodes which are changed by replace operations.

Consider $creplace$ that replaces an \textit{Internal} node by an \textit{Empty} node. Because the \textit{Internal} node has been flagged on a \textit{Compress} $op$ before $creplace$ (Lemma~\ref{lem: flag replace transition}), all of its children are \textit{Empty} and not changed (Lemma~\ref{lem: all children empty} and Lemma~\ref{lem: all children not changed}). Thus, $creplace$ does not affect the second claim of Definition~\ref{defi: quadtree structure}. 

Consider $ireplace$, $rreplace$, or $mreplace$ that replaces a terminal node by an \textit{Empty} node, a \textit{Leaf} node, or a sub-tree. By Lemma~\ref{lem: flag replace unflag transition}, before $replace^{k}$, $op.parent^{k}$ is flagged with $op$ such that no successful replace operation could happen on $op.parent^{k}$. Therefore, if the new node is \textit{Empty}, it does not affect the tree property. If the new node is a \textit{Leaf} node or a sub-tree, based on the post-conditions of the createNode function (Lemma~\ref{lem: createNode pre-conditions and post-condition}), after replace operations the second claim of Definition~\ref{defi: quadtree structure} still holds.
\end{proof}

\subsection{Linearizability}
In this Section, we define linearization points for \textit{basic operations}. As the compress operation is included in the move operation and the remove operation that returns true, it does not affect the linearization points of them. If an algorithm is linearizable, it could be ordered equivalently as a sequential one by their linearization points. Since all the modifications depend on $find(keys)$, we first point out its linearization point. 

To prove that our linearization points are correct, we shall demonstrate that for some time $T_{i}$, the key set in $snapshot_{T_{i}}$ is the same as the results of modifications linearized before $T_{i}$. For $find(keys)$, we should define the  linearization point at some $snapshot_{T_{i}}$ so that $l^{t}$ returned is on the $pysical\_path$ for $key^{k}$ in $snapshot_{T_{i}}$. 

\begin{lem}
\label{lem: find physical path}
For $find(keys)$ that returns tuples $\langle l^{k}, pOp^{k}, path^{k}\rangle$, there's a snapshot after its invocation and before reading $l^{k}$, such that $path(keys^{k})$ returned with $l^{k}$ at the end is a $physical\_path$ in $snapshot_{T_{i}}$.
\end{lem}
\begin{proof}
Since $l^{k}$ is a \textit{Leaf} node or an \textit{Empty} node, by Lemma~\ref{lem: find snapshot} there is a $snapshot_{T_{i}}$ that $l^{k}$ as the last node from $l_{i-1}$, and $l_{0}$ to $n_{i-1}$ in the $path(keys^{k})$ are active. Therefore $path(keys^{k})$ with $l^{k}$ is $physical\_path(keys^{k})$ in $snapshot_{T_{i}}$. We define $snapshot_{T_{i}}$ as $find(keys^{k})$'s linearization point.
\end{proof}

In the next, we define linearization points for \textit{basic operations}. For the insert operation, the remove operation, and the move operation that returns true, we define it to be the first $replace_{i}^{k}, 0 \leq k < \vert replace_{i} \vert$ belongs to $op_{i}$. To make a reasonable demonstration, we first show that the first $replace_{i}^{k}$ belongs the $op_{i}$ created by each operation itself, and then illustrate that $op$ is unique for each operation.

\begin{lem}
If the insert operation, the remove operation, or the move operation that returns true, the first $replace_{i}^{k}$ occurs before the returning, and it belongs to $op_{i}$ created by the operation itself.
\end{lem}
\begin{proof}
First we prove that $replace_{i}^{k}$ belongs to $op_{i}$ created by the operation itself.

For the insert operation, it returns true after successful $iflag$ and $ireplace$ using $op$ created at line~\ref{quadboost: insert Substitute}.

For the remove operation, it returns true after successful $rflag$ and $rreplace$ using $op$ created at line~\ref{quadboost: remove Substitute}.

For the move operation, it returns true after successful $mflag$ and $mreplace$ using $op$ created at line~\ref{quadboost: move create end}.

Corollary~\ref{cor: first replace succeed} indicates that the first replace operation must succeed, therefore all successful replace operations happen before returning.
\end{proof}

Then, we define the linearization points of the insert operation, the remove operation, and the move operation that returns true are before  returning. We should also point out that there's only one $op$ lead to the linearization point for each operation.

\begin{lem}
\label{lem: op last}
For the insert operation, the remove operation, and the move operation that returns true, $op$ is created at the last iteration in the while loop.
\end{lem}
\begin{proof}
Based on Lemma~\ref{lem: flag replace unflag transition}, all successful insert, remove and move operations follow the $flag\rightarrow replace \rightarrow unflag$ transition. $flag$ is the first successful flag operation after reading $op_{i}.parent$, so there is no other flag operation further. This establishes the claim.
\end{proof}

\begin{lem}
\label{lem: replace success before return}
If the insert operation, the remove operation, and the move operation that return false, successful $replace$ happens before returning.
\end{lem}
\begin{proof}
As Lemma~\ref{lem: op last} points out that $op$ is created at the last iteration. Moreover, before returning true, $mflag$, $iflag$ and $rflag$ must succeed according to the program order. As the first replace operation that belongs to $op_{i}$ must succeed (Lemma~\ref{cor: first replace succeed}, $replace$ happens before returning.
\end{proof}

\begin{lem}
\label{lem: replace false not success}
If the insert operation, the remove operation, and the move operation that return false, there's no successful $replace$ ever happened.
\end{lem}
\begin{proof}
Consider the execution of the insert operation, it returns false at line~\ref{quadboost: insert intree} where $op$ created are not flagged on nodes. Consider the remove operation, it returns false at line~\ref{quadboost: remove intree}. Also, $op$ created are not flagged on nodes. For the move operation, it returns false by calling the findCommon function or the continueFindCommon function so that the helpMove function must fail or it is not called. If the helpMove function is not called, there is no replace operations. If the helpMove function returns false, by the ordering of $mflag$ and $unflag$ (Lemma~\ref{lem: flag replace unflag transition} and Lemma~\ref{lem: flag unflag transition}), $op.iParent$ and $op.rParent$ are different. Moreover, the flag operation on the latter parent will fail such that $doCAS$ is false and no replace operation will perform.
\end{proof}

In the next step, we show that how the algorithm could be ordered as a sequential execution. First we linearize the insert operation, the remove operation and the move operation that returns true, where $replace$ results in changing the key set. Let $s_{i}$ be a set of \textit{Leaf} nodes that are active in $snapshot_{T_{i}}$ after performing operations $o_{0}$, $o_{1}$, ..., $o_{n}$, ordered by their linearization points $replace_{0}$, $repalce_{1}$, ..., $replace_{n}$ sequentially.

\begin{obs}
\label{lem: moved replace order}
If a \textit{Leaf} node is \textit{moved}, its $op$ is set before the replace operation on $op.iParent$, which is before the replace operation on $op.rParent$.
\end{obs}

\begin{lem}
\label{lem: contain snapshot}
For the contain operation that returns true, there is a corresponding snapshot that $l^{k}$ is active in $snapshot_{T_{i}}$. For the contain operation that returns false, there is a corresponding snapshot that $l^{k}$ is inactive in $snapshot_{T_{i}}$.
\end{lem}
\begin{proof}
For the first part, we prove that in a snapshot, the end node $l^{k}$ contains $keys^{k}$ and is not \textit{moved}. Lemma~\ref{lem: find physical path} shows that there is a snapshot consists of $physical\_path(keys^{k})$ after calling, thus $l^{k}$ that contains $keys^{k}$ is in $snapshot_{T_{i}}$ before judging whether the node is \textit{moved}. We only have to discuss $mreplace$ in the next cases because other replace operations will not affect $l^{k}.op$. If $l^{k}.op$ is null at verifying, $mreplace$ must have not been done according to Observation~\ref{lem: moved replace order}. In this case, we could linearize it at $snapshot_{T_{i}}$. If $l^{k}.op$ is not null, but $mreplace$ occurs after verifying whether $l.op.iParent$ has $l.op.oldIChild$, then it is also before the second $mreplace$ on $l.op.rParent$ (Observation~\ref{lem: moved replace order}). We could also linearize it $snapshot_{T_{i}}$.

For the second part, Lemma~\ref{lem: find physical path} shows that there is a snapshot consists $physical\_path(keys^{k})$. By Lemma~\ref{lem: replace properties remain}, in every snapshot our quadtree's property maintains. Hence, if $l^{k}$ does not contain $keys^{k}$, we could linearize it at $snapshot_{T_{i}}$. Or else, if $l^{k}$ is in $physical\_path$, but $l^{k}$ is \textit{moved} at verifying, we could linearize at $snapshot_{T_{i1}}$ after $mreplace$.
\end{proof}

\begin{lem}
\label{lem: lineariztion true}
\begin{enumerate}
\item For the insert operation, it returns true when $key \notin s_{i-1}$, $key \in s_{i}$.
\item For the remove operation, it returns true and $key \in s_{i-1}$, $key \notin s_{i}$.
\item For the move operation, it returns true and $oldKey \in s_{i-1}$, $newKey \notin  s_{i-1}$, $oldKey \notin s_{i}$, and $newKey \in s_{i}$.
\end{enumerate}
\end{lem}
\begin{proof}
\begin{enumerate}
Lemma~\ref{lem: replace success before return} shows that there is successful $replace$ occurs before the insert operation, the remove operation, and the move operation that returns true. 

If a replace operation succeeds, $op.parent$ has been flagged with $op$ other than \textit{Compress} before $replace$ (Lemma~\ref{lem: flag replace transition} and Lemma~\ref{lem: flag replace unflag transition}). Lemma~\ref{lem: op Compress reachable} implies that $op.parent$ is reachable from the root. Therefore, it is active. Besides, Corollary~\ref{lem: nodes above active} shows that the path from the root to $op.parent$ is an active. Therefore a snapshot exists such that the path from the root to $op.oldChild$ is a $physical\_path(key)$.

\item 
First we prove that $key \notin s_{i-1}$. By Lemma~\ref{lem: contain snapshot}, in the snapshot $l$ is inactive. Otherwise, replace operations will not execute. Hence, $key \notin s_{i-1}$. 

Second, by the post-conditions of the createNode function~\ref{lem: createNode pre-conditions and post-condition}, it creates a node or a sub-tree that contains $newKey$. According to Corollary~\ref{lem: replace properties remain}, after $ireplace$ the node contains $newKey$ is inserted and quadtree's properties maintained. Hence, $key \in s_{i}$.

\item 
First we prove that $key \in s_{i}$. By Lemma~\ref{lem: contain snapshot}, in the snapshot $l$ is active. Otherwise, replace operations will not execute. Hence, $key \notin s_{i-1}$. 

Second, the remove operation creates an \textit{Empty} node. According to Corollary~\ref{lem: replace properties remain}, after $rreplace$ the node contains $oldKey$ is removed so that $key \in s_{i}$.

\item 
First we prove that $oldKey \in s_{i-1}$, $newKey \notin s_{i-1}$. By Lemma~\ref{lem: contain snapshot}, $rl$ is active at $snapshot_{T_{i}}$, $il$ is inactive at $snapshot_{T_{i1}}, T_{i} < T_{i1}$. We have to  demonstrate that $rl$ is still active at $snapshot_{T_{i1}}$ so that we can use $snapshot_{T_{i1}}$ as $s_{i-1}$. 

We prove it by contradiction. Assuming that at $snapshot_{T_{i1}}$ $rl$ is inactive, hence there's some replace operation happens before $T_{i1}$ and after $T_{i}$. However, based on Lemma~\ref{lem: flag replace unflag transition}, $replace_{i}$ is the first successful replace operation that after reading $rl$. Therefore, it derives a contradiction. Hence we prove that in $snapshot_{T_{i1}}$ $rl$ is still active.

\end{enumerate}
\end{proof}

Finally we order linearization points for the insert operation, the remove operation and the move operation that returns false. We also order the contain operation. We consider they are linearized between $replace_{i-1}$ and $replace_{i}$ if exists.

\begin{lem}
\label{lem: lineariztion false}
\begin{enumerate}
\item For the insert operation returns false, $key \in s_{i-1}$.
\item For the remove operation returns false, $key \notin s_{i-1}$.
\item For the move operation returns false, either $oldKey \notin s_{i-1}$ or $newKey \in s_{i-1}$.
\end{enumerate}
\end{lem}
\begin{proof}
The first two parts are equivalent as the case in Lemma~\ref{lem: contain snapshot}.
\end{proof}

\begin{lem}
Our quadboost is a linearizable implementation.
\end{lem}
\begin{proof}
By lemma~\ref{lem: lineariztion true} and Lemma~\ref{lem: lineariztion false}, our algorithm returns the same result as they are finished in the linearized order. Therefore we prove our algorithm is linearizable and our linearization points are correct. 
\end{proof}

\subsection{Progress Condition}
We say an algorithm is non-blocking if the system as a whole is making progress even if some threads are starving. We prove our quadboost is non-blocking by following set of lemmas. We assume that there are finite number of \textit{basic operations} invocations.

\begin{obs}
\label{obs:finite basic operations}
There are finite number of \textit{basic operations}.
\end{obs}

\begin{lem}
\label{lem: path terminable}
$path(key^{k})$ returned by $find(keys)$ consists of finite number of keys.
\end{lem}
\begin{proof}
Observation~\ref{obs:finite basic operations} shows that the number of successful insert operation and move operation is limited. Lemma~\ref{lem: replace false not success} illustrates that for \textit{basic operations} return false, there is no successful replace operation. Hence, only successful insert and move operations add nodes into quadtree. The post-conditions of the createNode function (Lemma~\ref{lem: createNode pre-conditions and post-condition}) and the effects of $replace$ (Lemma~\ref{lem: replace properties remain}) demonstrate that they will add finite number of nodes into quadtree. We then prove that $path^{k}$ is terminable. By Lemma~\ref{lem: find snapshot}, $find(keys)$ returns $path(key^{k})$ which is a subpath of  $physical\_path(keys^{k})$ in $snapshot_{T_{i}}$ which is terminable. Thus, there are finite number of nodes in $path^{k}$.
\end{proof}

\begin{cor}
\label{cor: compress terminable}
The compress function must terminate.
\end{cor}
\begin{proof}
By Lemma~\ref{lem: path terminable}, $path^{k}$ is terminable. Since other functions called by the compress function do not consist loops, the compress function must terminate.
\end{proof}

The we shall demonstrate that there are finite number of three different CAS transitions.

\begin{lem}
\label{lem: spatial order}
There is an unique \textit{spatial order} among nodes in quadtree in $snapshpt_{T_{i}}$
\end{lem}
\begin{proof}
As illustrated by the algorithm at line~\ref{quadboost: move getSpatialOrder}, the getSpatialOrder function compares $ip$ with $rp$ by the order: $x \rightarrow y \rightarrow w$. In our quadtree, we only consider square partitions on two-dimensional space. Hence, $w$ is always equal to $h$. We prove the lemma by contradiction. Assuming there are two different \textit{Internal} nodes with the same $\langle x, y, w\rangle$, they represent the same square starting with left coroner $\langle x, y\rangle$ with width $w$. By Lemma~\ref{lem: replace properties remain}, in every snapshot quadtree's properties remain. There cannot be two squares with the same left corner and $w$. Hence, in $snapshpt_{T_{i}}$, a \textit{Internal} node consists of a unique $\langle x, y, w\rangle$ tuple that can be ordered correctly.
\end{proof}

\begin{lem}
\label{lem: finite transitions}
\begin{enumerate}
	\item There are finite number of successful $flag\rightarrow replace \rightarrow unflag$ transitions.
	\item There are finite number of successful $flag\rightarrow replace$ transitions.
	\item There are finite number of successful $flag\rightarrow unflag$ transitions.
\end{enumerate}
\end{lem}
\begin{proof}
\begin{enumerate}
\item 
By Lemma~\ref{lem: replace success before return}, replace operations only occur in \textit{basic operations} that return true. By Lemma~\ref{lem: op last}, there's a unique $op$ created for operations that return true. Hence, for the move operation, the insert operation and the remove operation, there are finite number of $flag\rightarrow replace \rightarrow$ circles.

\item 
By Lemma~\ref{lem: replace success before return}, replace operations only occur in \textit{basic operations} that return true. Lemma~\ref{lem: flag replace transition} shows that the $flag\rightarrow replace$ transition happens only in the compress function which is called by the successful remove operation or the move operation. Hence, there are a finite number of $flag\rightarrow replace$ transitions, as Corollary~\ref{cor: compress terminable} shows that the compress function is terminable.

\item
The $flag\rightarrow unflag$ transition executes only in the helpMove function where $op.iParent$ cannot be flagged according to Lemma~\ref{lem: flag unflag transition} (if $op.iFirst$ is true). The case that $op.iFirst$ is false can be proved symmetrically.

From first two parts of the lemma, there are finite number of replace operations that change quadtree's structure. If the $flag\rightarrow replace \rightarrow$ transition happens simultaneously, it will reset a node's $op$ to \textit{Clean}. If the $flag\rightarrow replace \rightarrow unflag$ transition happens in the meantime. Lemma~\ref{lem: compress push more once} indicates that a node with \textit{Compress} $op$ will not be pushed twice. Consider if the helpMove function detects $op.iParent.op$ as \textit{Compress}, then $op.iParent$ is not reachable in the next time. Hence, quadtree is stabilized and only when $flag\rightarrow unflag$ transitions execute infinitely and never set $doCAS$ as true.

According to Lemma~\ref{lem: spatial order}, we order \textit{spatial order} on both $op.iParent$ and $op.rParent$ as $ip_{0} > ip_{1} ... > ip_{n}$ and $rp_{0} > rp_{1} ... > rp_{n}$. We assume that there's a ring such that $move_{i}$ which flags $rp_{i}$ but fails on $ip_{i}$, and $move_{i}'$ which flags $ip_{i}$ but fails on $rp_{i}$. By assumption we order $rp_{i}$ and $ip_{i}$ as $rp_{i} > ip_{i}$ by $move_{i}$'s order, and $ip_{i} > rp_{i}$ by $move_{i}'$'s order. But by Lemma~\ref{lem: spatial order}, all \textit{Internal} nodes in a snapshot can be ordered uniquely. Therefore it derives a contradiction. 

Then, we prove that there's no ring among all $move_{i}$ in the same snapshot, so at least one of $move_{i}$ could set $doCAS$ to true, resulting in the $flag\rightarrow replace\rightarrow unflag$ transition. Let's consider the dependency among all $move_{i}$. If $move_{i}$ which fails on $p_{i}$ has been flagged by $move_{i}'$, then there is a directed edge from $move_{i}$ to $move_{i}'$ by order. In this way, the head node in the graph must have no out edge as proved. It will set $doCAS$ to true, all move operations that are directly connected to it will help it finish. After erase the former head node from the graph, there will be other nodes that have no out edge. Finally, all dependency edges are removed. Therefore, there are a finite number of $flag\rightarrow unflag$ transitions.
\end{enumerate}
\end{proof}

\begin{obs}
\label{obs: help mutual}
The help function, the helpCompress function, the helpMove function, and the helpSubstitute function are not called in a mutual way. (If method A calls method B, and method B calls method A reciprocally, we say method A and B are called in a mutual way.)
\end{obs}

\begin{lem}
\label{lem: help snapshot}
If help function returned at $T_{i}$, and $find(keys)$ at the prior iteration reads $p^{0}.op$ at $T_{i1} < T_{i}$, \textit{Leaf} nodes in $snapshot_{T_{i}}$ and $snapshot_{T_{i1}}$ are different.
\end{lem}
\begin{proof}
We prove the lemma by contradiction, there are two possible scenarios. First, the help function might be interminable. But based on Observation~\ref{obs: help mutual}, the help function must terminate. Second, the help function might execute without changing quadtree. According to Lemma~\ref{lem: finite transitions}, there are finite number of CAS transitions. So if some help function does not change quadtree, the structure might be changed before $T_{i1}$. For the move function, if $ip$ and $rp$ are the same but their $op$s are different, the help function might not change the structure. Nevertheless, $snapshot_{T_{i}}$ and $snapshot_{T_{i1}}$ are different. Or else, there are some interval that all transitions are $flag \rightarrow unflag$. However, Lemma~\ref{lem: finite transitions} also shows that all transitions form a acyclic graph that there is no interval that all transitions are $flag \rightarrow unflag$. Therefore, a $flag\rightarrow replace \rightarrow unflag$ transition or a $flag\rightarrow unflag$ transition is executed to change the snapshot. 
\end{proof}

\begin{lem}
Our quadboost algorithms are non-blocking.
\end{lem}
\begin{proof}
We have to prove that no process will execute loops infinitely without changing keys in quadtree. First, we prove that $path$ is terminable. Next, we prove that $find(keys^{k})$ starts from an active node in $physical\_path(keys^{k})$ in $snapshot_{T_{i}}$ between $i_{th}$ iteration and $i+1_{th}$ iteration. Finally, \textit{Leaf} nodes in $snapshot_{T_{i1}}$ at the returning of $find(keys)$ at $i_{th}$ iteration is different from $snapshot_{T_{i2}}$ that at the returning of $find(keys)$ at the $i+1_{th}$ iteration.

For the first part, initially we start from the root node. Therefore, $path$ is empty. Moreover, as Lemma~\ref{lem: path terminable} shows that $path^{k}$ consists of finite number of keys, we establish this part.

For the second part, the continueFind function and the continueFindCommon function pop all nodes with a \textit{Compress} $op$ from $path$. For the insert operation and the remove operation, since Lemma~\ref{lem: op Compress reachable} shows that an \textit{Internal} nodes whose $op$ is not \textit{Compress} is active , and Lemma~\ref{lem: nodes above active} shows that nodes above the active node are also active, there is a snapshot in which the top node of $path$ is still in $physical\_path(keys^{k})$. For the move operation, if either $rFail$ or $iFail$ is true, it is equivalent with the prior case. If $cFail$ is true, Lemma~\ref{lem: lca active} illustrates that if the LCA node is active, it is in $physical\_path$ for both $oldKey$ and $newKey$. Thus, there is also a snapshot in which the start node is in $physical\_path$.

We prove the third part by contradiction, assuming that quadtree is stabilized at $T_{i}$, and all invocations after $T_{i}$ are looping infinitely without changing \textit{Leaf} nodes.

For the insert and remove operations, before the invocation of $find(keys)$ at the next iteration, they must execute the help function at line~\ref{quadboost: insert help} and line~\ref{quadboost: remove help} accordingly. In both cases, the help function changes the snapshot (Lemma~\ref{lem: help snapshot}).

For the move operation, consider different situations of the continueFindCommon function. If $rFail$ or $iFail$ is true, two situations arise: (1) $iOp$ or $rOp$ is \textit{Clean} but $mflag$ fails. (2) $iOp$ or $rOp$ is not \textit{Clean}. Consider the first case that $mflag$ fails, $iOp$ and $rOp$ are updated respectively at line~\ref{quadboost: rFail true} and line~\ref{quadboost: iFail true}. If $ip$ and $rp$ are different, before its invocation of $find(keys^{k})$, the help function is performed at line~\ref{quadboost: continueFindCommon insert help} and line~\ref{quadboost: continueFindCommon remove help}. Thus by Lemma~\ref{lem: help snapshot}, the snapshot is changed between two iterations. Now consider if $cFail$ is true. If $ip$ and $rp$ are the same, it could result from the difference between $rOp$ and $iOp$. In this case, the snapshot might be changed between reading $rOp$ and $iOp$. It could also result from the failure of $mflag$. For above cases, the help function at line~\ref{quadboost: continueFindCommon remove help common} would change quadtree. If $ip$ and $rp$ are different, it results from that either $iPath$ or $rPath$ has popped the LCA node. We have proved the case in the former paragraph. Thus, we derives a contradiction.

From above discussions, we prove that quadboost is non-blocking.
\end{proof}



\ifCLASSOPTIONcaptionsoff
  \newpage
\fi

\end{document}